\documentclass[mnsc,nonblindrev]{informscopy}

\OneAndAHalfSpacedXI

\TheoremsNumberedThrough 
\EquationsNumberedThrough

\usepackage[breakable, theorems, skins]{tcolorbox}
\definecolor{cornellred}{rgb}{0.7, 0.11, 0.11}
\definecolor{dgreen}{rgb}{0.0, 0.5, 0.0}
\definecolor{ballblue}{rgb}{0.13, 0.67, 0.8}
\definecolor{royalblue(web)}{rgb}{0.25, 0.41, 0.88}
\definecolor{bleudefrance}{rgb}{0.19, 0.55, 0.91}
\definecolor{royalazure}{rgb}{0.0, 0.22, 0.66}

\usepackage{hyperref}
\hypersetup{
	colorlinks = true,
	linkcolor=royalazure,
	citecolor=royalazure,
	urlcolor=black,
	linkbordercolor = {white}
}
\usepackage{enumitem}
\usepackage{subfiles}
\usepackage{bm}
\usepackage{xfrac}
\usepackage{bbm}
\usepackage{booktabs}
\usepackage{csquotes}
\usepackage{pst-all}
\usepackage{accents}
\usepackage{subcaption}
\usepackage{perpage}

\makeatletter
\renewenvironment*{displayquote}
  {\begingroup\setlength{\leftmargini}{0.6cm}\csq@getcargs{\csq@bdquote{}{}}}
  {\csq@edquote\endgroup}
\makeatother
\DeclareRobustCommand{\mybox}[2][gray!15]{
\begin{tcolorbox}[  
        left=0pt,
        right=0pt,
        top=0pt,
        bottom=0pt,
        colback=#1,
        colframe=#1,
        enlarge left by=0mm,
        boxsep=10pt,
        arc=2pt,outer arc=2pt,
        ]
        #2
\end{tcolorbox}
}
\usepackage[ruled,vlined,linesnumbered]{algorithm2e}
\SetNoFillComment

\usepackage{cleveref}
\usepackage{natbib}
 \bibpunct[, ]{(}{)}{,}{a}{}{,}%
 \def\bibfont{\small}%
 \def\bibsep{\smallskipamount}%

\usepackage{thm-restate}

\setcitestyle{authoryear}

 \newcommand{\prob}[2][]{\text{\bf Pr}\ifthenelse{\not\equal{}{#1}}{_{#1}}{}\!\left[{\def\givenn{\middle|}#2}\right]}
\newcommand{\expect}[2][]{\text{\bf E}\ifthenelse{\not\equal{}{#1}}{_{#1}}{}\!\left[{\def\givenn{\middle|}#2}\right]}
\newcommand{\condition}{\,\mid\,}

\makeatletter
\newsavebox\myboxA
\newsavebox\myboxB
\newlength\mylenA

\newcommand*\xoverline[2][0.75]{%
    \sbox{\myboxA}{$\m@th#2$}%
    \setbox\myboxB\null
    \ht\myboxB=\ht\myboxA%
    \dp\myboxB=\dp\myboxA%
    \wd\myboxB=#1\wd\myboxA
    \sbox\myboxB{$\m@th\overline{\copy\myboxB}$}
    \setlength\mylenA{\the\wd\myboxA}
    \addtolength\mylenA{-\the\wd\myboxB}%
    \ifdim\wd\myboxB<\wd\myboxA%
       \rlap{\hskip 0.5\mylenA\usebox\myboxB}{\usebox\myboxA}%
    \else
        \hskip -0.5\mylenA\rlap{\usebox\myboxA}{\hskip 0.5\mylenA\usebox\myboxB}%
    \fi}
\makeatother

\newcommand{\Demand}{d}
\newcommand{\Demandi}[1]{\Demand_{#1}}
\newcommand{\DemandVec}{\vec{\Demand}}
\newcommand{\DemandJoinDist}{\vec{\mathcal{F}}}

\newcommand{\NumAgents}{n}
\newcommand{\ExpDemand}{\mu}
\newcommand{\ExpDemandi}[1]{\ExpDemand_{#1}}
\newcommand{\Alloc}{x}
\newcommand{\Alloci}[1]{\Alloc_{#1}}
\newcommand{\Policy}{\pi}
\newcommand{\Supply}{s}
\newcommand{\Supplyi}[1]{\Supply_{#1}}

\newcommand{\ExpostObj}{W_{\textsc{p}}}
\newcommand{\ExanteObj}{W_{\textsc{a}}}
\newcommand{\NonNegReals}{\mathbb{R}_{\geq 0}}
\newcommand{\NonNegRealsN}{\mathbb{R}_{\geq 0}^{n}}

\newcommand{\LowerboundFunA}{\kappa_{\textsc{a}}(\ExpDemand,\NumAgents)}
\newcommand{\LowerboundFunP}{\kappa_{\textsc{p}}(\ExpDemand,\NumAgents)}
\newcommand{\LowerboundFunEx}{\kappa_{\textsc{p}}(\ExpDemand,\NumAgents)}
\newcommand{\LowerboundFunPj}{\kappa_{\textsc{p}}(\ExpDemand^j,\NumAgents)}
\newcommand{\LowerboundFunPnoarg}{\kappa_{\textsc{p}}}
\newcommand{\LowerboundFunExAnte}{\kappa_{\textsc{a}}(\ExpDemand, \NumAgents)}
\newcommand{\BoundFixedAll}{\kappa_{\textsc{fa}}(\ExpDemand, \NumAgents)}
\newcommand{\LowerboundFunExPostOne}{\kappa_{\textsc{p}}(\ExpDemand, 1)}
\newcommand{\PolicyName}{projected proportional allocation}
\newcommand{\PolicyNameAb}{PPA}

\newcommand{\TargetFillRate}{\tau}

\newcommand{\NumGoods}{m}
\newcommand{\Weightj}[1]{\lambda^{#1}}

\newcommand{\Budget}{B}
\newcommand{\Costj}[1]{c^{#1}}

\newcommand{\AppMinFill}{r}
\newcommand{\AppMinFilli}[1]{\AppMinFill_{#1}}
\newcommand{\AppAlloc}{y}
\newcommand{\AppAlloci}[1]{\AppAlloc_{#1}}
\newcommand{\SocialWelfare}{U}
\newcommand{\FairParam}{\alpha}

\newcommand{\MultiGoodSupply}{s}

\newcommand{\DetGuar}{\xoverline{W}}
\newcommand{\JointDistStateSpace}{\Delta\left({\NonNegRealsN}; \ExpDemand\right)}
\newcommand{\InvDistStateSpace}{\xoverline{\Delta}\left({\NonNegReals}; \ExpDemand\right)}

\newcommand{\fixedthresh}{TFR}
\newcommand{\fixedalloc}{fixed-allocation}

\newcommand{\invdemanddist}{G}
\newcommand{\invdemand}{v}
\newcommand{\tfrvar}{\tau}
\newcommand{\Tfrvar}{T}
\newcommand{\hfun}{h}
\newcommand{\coefvar}{c}
\newcommand{\xvar}{x}
\newcommand{\transmissparam}{\gamma}
\newcommand{\population}{p}
\newcommand{\neighborset}{N}
\newcommand{\exptoinfparam}{\delta}
\newcommand{\inftorecparam}{\lambda}
\newcommand{\adjacencyparam}{\alpha}

\newcommand{\randomwalk}{X}
\newcommand{\hypermean}{\xi_r}
\newcommand{\hypersd}{\sigma_r}
\newcommand{\revcolor}[1]{#1}

\newcommand{\dvar}{\ExpDemand}

\begin{document}

\TITLE{Fair  Dynamic  Rationing}

\RUNAUTHOR{Manshadi, Niazadeh, Rodilitz}

\RUNTITLE{Fair Dynamic Rationing}

\ARTICLEAUTHORS{
\AUTHOR{Vahideh Manshadi}
\AFF{Yale School of Management, New Haven, CT, \EMAIL{vahideh.manshadi@yale.edu}}
\AUTHOR{Rad Niazadeh}
\AFF{University of Chicago Booth School of Business, Chicago, IL, \EMAIL{rad.niazadeh@chicagobooth.edu}}
\AUTHOR{Scott Rodilitz}
\AFF{Stanford Graduate School of Business, Stanford, CA, \EMAIL{rodilitz@stanford.edu}}
}

\ABSTRACT{
We study the allocative challenges that governmental and nonprofit organizations face when tasked with equitable and efficient rationing of a social good among agents whose needs (demands) realize sequentially and are possibly correlated. As one example, early in the COVID-19 pandemic, the Federal Emergency Management Agency faced overwhelming, temporally scattered, a priori uncertain, and correlated demands for medical supplies from different states. In such contexts, social planners aim to maximize the minimum fill rate across sequentially arriving agents, where each agent's fill rate is determined by an irrevocable, one-time allocation. For an arbitrarily correlated sequence of demands, we establish upper bounds on the expected minimum fill rate (ex-post fairness) and the minimum expected fill rate (ex-ante fairness) achievable by any policy. Our upper bounds are parameterized by the number of agents and the expected demand-to-supply ratio, yet we design a simple adaptive policy called projected proportional allocation (PPA) that simultaneously achieves matching lower bounds for both objectives (ex-post and ex-ante fairness), for any set of parameters. Our PPA policy is transparent and easy to implement, as it does not rely on distributional information beyond the first conditional moments. Despite its simplicity, we demonstrate that the PPA policy provides significant improvement over the canonical class of non-adaptive target-fill-rate policies. We complement our theoretical developments with a numerical study motivated by the rationing of COVID-19 medical supplies based on a standard SEIR modeling approach that is commonly used to forecast pandemic trajectories. In such a setting, our PPA policy significantly outperforms its theoretical guarantee as well as the optimal target-fill-rate policy.}

\maketitle

\newpage

\section{Introduction}
\label{sec:intro}
In Spring 2020, with the COVID-19 pandemic surging across the US, states were relying on  the Federal Emergency Management Agency (FEMA) to provide urgently needed medical equipment from the Strategic National Stockpile. 
Unequipped for such a widespread emergency, FEMA aimed to ration its limited supplies in order to address states' current needs while also retaining some of the stockpile in anticipation of future needs. 
However, the allocation decisions made by FEMA were inconsistent and lacked transparency, which frustrated state officials \citep{WaPo2020}.\footnote{``{\em We don’t know how the federal government is making those decisions,}” said Casey Katims, the federal liaison for Washington state. \label{foot:untransparent}}
Because having access to medical equipment {can be} a matter of life or  death for a COVID-19 patient, making {allocation} decisions which are efficient and equitable is of paramount importance \citep{NJM2020}. 
Achieving efficiency alone is easy: a first-come, first-serve policy allocates all of the supply to meet early-arriving needs. However, such a policy 
can be unfair to patients in states where needs materialize later.

The above is just one example of a fundamental sequential allocation problem that social planners face when aiming to allocate divisible goods as \emph{efficiently} and \emph{equitably} as possible to demanding agents that arrive over time.

\subsection{Overview of Contributions}

In this paper, we take the first step toward theoretically studying the aforementioned class of problems.
We develop 
a framework for fair dynamic rationing where agents' 
one-time needs (demands) for a divisible good realize sequentially and can be {\em arbitrarily} correlated. In particular, upon arrival of each agent's demand, 
the planner makes an irrevocable decision about their 
{\em fill rate} (FR), \revcolor{i.e., the fraction of the agent's demand that is satisfied by a one-time allocation.}  
Toward jointly achieving efficiency and equity, the planner aims to {\em maximize the minimum} FR, either ex post or ex ante.  
To assess the performance of sequential allocation policies, we introduce measures of ex-post and ex-ante fairness guarantees. 
For this general setting:
\begin{enumerate}[label=(\roman*)]
\item We establish upper bounds on the ex-post and ex-ante fairness guarantees achievable by any policy. These bounds are parameterized by the \emph{supply scarcity} (i.e., the expected demand-to-supply ratio) and the number of agents.  
    
    \item Remarkably, we show that a simple, adaptive, and transparent policy called {\em \PolicyName} (\PolicyNameAb) simultaneously achieves our upper bounds on the ex-post and ex-ante fairness guarantees for any set of parameters. 
    \item We illustrate the power of adaptivity by characterizing the ex-post guarantee of the optimal  target-fill-rate policy and showing that such a non-adaptive policy cannot achieve our upper bounds.
    \item Finally, we demonstrate the effectiveness of our policy through an illustrative case study motivated by 
    the allocation of COVID-19 medical supplies based on a model of demand which was used by the White House. 
\end{enumerate}

\smallskip
\noindent{\bf Introducing a framework for fair dynamic rationing:}
We study the allocation of a divisible good  to agents arriving over time with varying levels of demand. We assume the demand sequence is drawn from an arbitrary but known joint  distribution across all agents.
To account for heterogeneity in the demand level of different agents, we set each agent's utility to be its FR. In our base model, we focus on the objective of maximizing the minimum FR across all agents. Such an objective---which is in the spirit of Rawlsian justice---maximizes the utility of the worst-off agent. As such, it takes fairness into consideration along with efficiency.\footnote{In \Cref{sec:discussions}, we generalize our objective function.} Due to the stochasticity of the demand sequence, we consider two versions of this objective function: {the expected minimum FR and the minimum expected FR (see \cref{eq:expost} and \cref{eq:exante}, respectively, as well as the subsequent discussion).}

Like other online stochastic optimization problems,  our sequential allocation problem can be formulated as a dynamic program (DP), and it similarly suffers from the curse of dimensionality as well as other practical limitations such as a lack of interpretability.
\revcolor{(We provide further discussion of the DP in Remark \ref{rem:online}, and in Appendix \ref{apx:DP} we formally present the DP, illustrate its exponential size, and discuss  other practical drawbacks.)}
Consequently, we aim to design sequential allocation policies that perform well while being practically appealing and computable in polynomial time.
We assess the performance of a policy by computing its ex-post and ex-ante fairness guarantees for any given supply scarcity and number of agents. In defining our notions of such guarantees,
we use the 
minimum FR achievable under deterministic demand as a \revcolor{normalization factor (see  Definitions \ref{def:scarcity}, \ref{def:normalization:factor}, and \ref{def:polperf} and their related discussion in Section \ref{sec:prelim})} to separate the impact of demand stochasticity from the impact of supply scarcity. 
The ex-post {(resp. ex-ante)} fairness guarantee of a policy serves as a lower bound on the expected minimum {(resp. minimum expected)} FR that the policy achieves relative to our \revcolor{normalization factor} under all possible joint demand distributions.

\smallskip
\noindent{\bf Establishing upper bounds:}
In order to gain insight into the difficulty of achieving equity and efficiency in sequential allocation, we develop upper bounds on the achievable fairness guarantees of any  policy, even policies which cannot be computed in polynomial time. For intuition, consider the following example with two agents. The first agent has demand of $B_1$, where $B_1$ is a Bernoulli random variable with success probability $2/3$. The second agent has demand $B_1 \times B_2$, where $B_2$ is an independent  Bernoulli random variable with success probability $1/2$.
In other words, the demand sequence is equally likely to be $(0, 0)$, $(1,0)$, or $(1,1)$.
For such an instance, no sequential  policy can distinguish between the latter two scenarios after observing the first demand, which leads to a sub-optimal decision. 
Building on the above intuition, in Sections \ref{subsec:expostupperbound} and \ref{subsec:ex-ante}, we establish upper bounds on the ex-post and ex-ante guarantees of any policy (see Theorems \ref{thm:hardness-ex-post} and \ref{thm:exante}).

As we later show, these bounds are indeed tight. Thus, conducting 
comparative statics with respect to the supply scarcity and the number of agents reveals several insights (see Figure \ref{fig:upperbound3d}): {when demand is small relative to supply, the bounds on} both fairness guarantees deteriorate with increased demand. 
However, in the over-demanded regime, the bounds are independent of the supply scarcity. Further, in both the under-demanded and over-demanded regimes, the ex-post fairness guarantee  worsens with more agents. On the other hand, the ex-ante fairness guarantee  is independent of the number of agents. This highlights the fundamental difference in our notions of fairness: the objective corresponding to ex-post fairness 
is concerned with fairness along all samples paths, whereas the objective corresponding to ex-ante fairness 
is only concerned with marginal fairness (see the related discussion in \Cref{sec:prelim}).

\smallskip
\noindent{\bf Achieving upper bounds:} 
Since our upper bounds apply to all sequential policies including the optimal online policy (namely, the exponential-sized DP), it {would be reasonable if no} policy could achieve these  upper bounds in polynomial-time. However, we show that not only are these upper bounds achievable, but they can be achieved by our \PolicyNameAb\ policy.
To motivate our policy, let us consider a hypothetical situation where the demand sequence is known a priori. In that case, the optimal allocation under both objectives is to equalize the FRs and then maximize that FR (see \cref{eq:det:fair} and its related discussion). Alternatively, this can be written as a deterministic DP with a simple solution: 
at each time period, proportionally allocate the remaining supply based on the current demand and the total future demand (see Section \ref{subsec:ourpolicy} and Appendix \ref{apx:offline:DP}).
When demand is stochastic, our PPA policy simply  replaces all the future random demands by their projected values, namely, their conditional expectations (see \cref{eq:PPAinvariant}).

In Sections \ref{subsec:ex-post-lower} and \ref{subsec:ex-ante}, we analyze the ex-post and ex-ante fairness guarantees of the PPA policy and show that it achieves the best of both worlds: our lower bounds on the PPA policy's guarantees match the corresponding upper bounds for any supply scarcity and any number of agents. {These two analyses rely on delicate inductive arguments. For ex-post fairness,} we establish a lower bound on the value-to-go function of our \PolicyNameAb\ policy by analyzing the evolution of the minimum FR and progressively constructing a worst-case joint distribution for demand (see \Cref{lem:minfilllowerbound} in \Cref{subsec:ex-post-lower}). {For ex-ante fairness, we demonstrate that the expected demand-to-supply ratio before the arrival of each agent is non-increasing when following the \PolicyNameAb\ policy, which enables us to bound the marginal expected FR for each agent (see Appendix \ref{apx:ex-ante}).}

We highlight that beyond enjoying the best possible guarantees, our \PolicyNameAb\ policy is practically appealing: it is computationally efficient, interpretable, and transparent. In addition, it does not require full distributional knowledge, as it only relies on the first conditional moments of the joint distribution for demand. Policies which rely on detailed distributional knowledge can be prone to errors or perturbations (see Remark \ref{rem:online} and Appendix \ref{apx:DPexample}).

\smallskip
\noindent{\bf Establishing sub-optimality of target-fill-rate policies:}
In addition to showing that our PPA policy achieves the best possible guarantees, we extend our work to studying the subclass of target-fill-rate (\fixedthresh) policies. A \fixedthresh\ policy commits upfront to a fill rate $\TargetFillRate$, and upon arrival of each agent, it allocates a fraction $\TargetFillRate$ of that agent's demand  until it exhausts the supply. 
Our study of \fixedthresh\ policies is motivated by two reasons: (i) since such policies are transparent and easy-to-communicate, they are frequently used in practice, including at the outset of the COVID-19 pandemic when an initial formula allocated a fixed percentage of states' estimated needs \citep{WaPo2020}, and (ii) \fixedthresh\ policies are a natural yet powerful class of non-adaptive policies (see \Cref{subsec:AdaptivityGap}). 
Consequently, comparing the performance of \fixedthresh\ policies with that of our adaptive PPA policy sheds light on the limitations of making non-adaptive decisions.

Intuitively, a \fixedthresh\ policy can perform poorly because it does not take advantage of information that reduces future uncertainty. For instance, consider a setting with two agents where the second agent's demand is perfectly correlated with the first agent's demand. A simple adaptive policy---such as our \PolicyNameAb\ policy---will perform optimally in such a setting because demand is deterministic upon the first agent's arrival. The \PolicyNameAb\ policy achieves such performance by crucially leveraging information about the second agent's demand when determining the first agent's fill rate. In contrast, a \fixedthresh\ policy targets the same fill rate regardless of the first agent's demand, and consequently cannot ensure that sufficient supply remains for the second agent. Based on this intuition, in \Cref{subsec:AdaptivityGap}, we provide a tight bound on the ex-post fairness guarantee of the optimal \fixedthresh\ policy (see \Cref{thm:NonAdaptive}), which can be considerably lower than the corresponding guarantee of our adaptive \PolicyNameAb\ policy (see \Cref{fig:AdaptivityGap}). 
\revcolor{On the other hand, we show that if the coefficient of variation of total demand is low, then the optimal TFR policy can provide a stronger guarantee than our \PolicyNameAb\ policy (see Proposition \ref{prop:coefvar} and \Cref{fig:coefvar}).}

To characterize the ex-post fairness guarantee of the optimal TFR policy, we construct the worst-case total demand distribution against such a policy.
In the proof, we establish a rather surprising connection to the literature on monopoly pricing and Bayesian mechanism design (see \citet{hartline2013mechanism} for more details on this literature). In particular, upon mapping the problem {of finding the worst-case instance} into the quantile space, our problem reduces to a constrained version of the (single-item) monopoly pricing problem (see \Cref{rem:monopolypricing}). We identify two key properties of the worst-case distribution in this constrained monopoly pricing problem, and by exploiting the connection to our original problem, we end up with the desired characterization of the worst-case total demand distribution against the optimal TFR policy. Due to this connection, our proof technique and corresponding results 
can be of independent interest (e.g., see \citet{alaei2019optimal} for proof techniques and results in the same spirit).

\smallskip
\noindent{\bf Illustrative case study:} 
To demonstrate the effectiveness of our policy, in \Cref{sec:numerics} we conduct a numerical case study motivated by the allocative challenges that FEMA faced at the beginning of the COVID-19 pandemic (as discussed at the beginning of this section). 
\revcolor{Borrowing from the epidemiology literature around the COVID-19 pandemic, we develop a simple compartmental SEIR model that governs the need for medical supplies in different inter-connected locations. Using such a model, we demonstrate that the demand is highly variable and has complex correlation structure across locations (see Figure \ref{fig:sims}).}
Our simulation results illustrate the superior performance of our \PolicyNameAb\ policy compared to both its ex-post fairness guarantee and the optimal TFR policy.  
Further, the results suggest that our \PolicyNameAb\ policy performs nearly as well as the DP solution (which, as we discuss, suffers from many practical limitations). \revcolor{Additionally, our simulation results demonstrate the efficiency  of our policy as well as its robustness to model mis-specification (see Table \ref{table:full_sims}).}

Allocating  medical supplies in a pandemic 
 is just one motivating example of the challenges that arise when a governmental or nonprofit organization aims to ration supply among agents whose (a priori uncertain and correlated) needs realize sequentially. 
Other examples include the allocation of emergency aid when a natural disaster such as a hurricane or wildfire impacts multiple locations over time \citep{wang2019measuring}, as well as the distribution of food donations by mobile pantries that sequentially visit agencies \citep{lien2014sequential}.\footnote{As explained in detail in \citet{lien2014sequential}, even though the daily demand for food donations from different agencies are not temporally scattered, they will only be observed by the operators upon their arrival at the sites.} 
Our proposed policy can effectively guide transparent allocation decisions in such contexts while also providing a guarantee on the fairness level of the process. 
Finally, as discussed in \Cref{sec:discussions}, our framework can be enriched to account for other practical considerations, such as (i) generalized objective functions that enable the social planner to balance equity and efficiency to varying degrees, and  (ii) rationing multiple types of resources (see Corollaries \ref{cor:generalfunctions}, \ref{cor:multiplegoods}, and \ref{cor:multiplegoodsupper} in \Cref{subsec:variants}).

\subsection{Related Work}
We conclude this section by discussing how our work relates to and contributes to several streams of literature.
\label{sec:lit:rev}

\smallskip
\noindent {\bf Fairness in static resource allocation:}
Considerations of fairness and its trade-off with efficiency have frequently arisen in the resource allocation literature in operations research and computer science.\footnote{Other recent papers have focused on fairness in the contexts of pricing \citep{cohen2019pricing}, information acquisition \citep{cai2020fair}, targeted interventions \citep{levi2019optimal}, service levels \citep{jiang2019achieving}, and online learning \citep{gupta2019individual}. See also the work of \citet{cayci2020group} that considers fair resource allocation with online learning.} 
We begin by discussing papers which study fairness in static (one-shot) allocation settings.
The seminal work of \citet{bertsimas2011price} considers a general setting where a central decision-maker allocates  $m$ divisible resources to $n$ agents, each with a different utility function. Focusing on two commonly used notions of fairness in allocation, max-min and proportional fairness, the authors characterize the efficiency loss due to maximizing fairness (see also \citealt{bertsimas2012efficiency} and  \citealt{bertsimas2013fairness}). 
{If demand was deterministic in our setting,} the optimal allocation {would coincide} with that of the max-min objective in \citet{bertsimas2011price}. Namely, for both objectives, the optimal allocation consists of maximized equal FRs.

Focusing on indivisible goods,  \citet{donahue2020fairness}
considers the trade-off between fairness and utilization when demand is distributed across different agents. A priori, only demand distributions are known. However, after a one-shot allocation decision, all demand values realize. 
The fairness notion considered in this line of work is in the same spirit of our notion of ex-ante fairness:  they require that an individual's chance of receiving the resource should not significantly depend on the group to which the individual belongs. Similarly, by maximizing the minimum expected FR, we aim to reduce the impact of an agent's place in the sequence of arrivals. 
Sharing similar motivation to our paper, \citet{pathak2020fair} and \citet{grigoryan2021effective} consider equitable COVID-19 vaccine allocation. 
However, the settings (e.g., offline and deterministic), models, and techniques in both papers differ drastically from those in this work. 

Also falling within the category of static allocation of indivisible goods, a stream of papers in computer science considers allocation problems when agents' valuations are deterministically known. For deterministic algorithms, recent research has centered on the existence of allocations which satisfy certain fairness properties, such as envy-freeness up to any good (see, e.g., \citet{chaudhury2020efx} and references therein). 
For randomized algorithms, the closest to our work is the recent work of \citet{nisargec20}, which uses notions of {ex-post and ex-ante} fairness and explores whether both can be achieved simultaneously. They develop a randomized algorithm that is {approximately fair ex post and precisely fair ex ante.} 
We ask a similar question, albeit in a dynamic divisible-good setting with random and correlated demand, and we affirmatively answer it: our \PolicyNameAb\ policy exactly achieves the best possible fairness guarantee ex post as well as ex ante  (see Theorems \ref{thm:expost} and \ref{thm:exante}).

\smallskip
\noindent {\bf Fairness in dynamic resource allocation:}
We now turn our attention to papers that consider fairness in dynamic (online) allocation settings. In terms of modeling, closest to our work are \citet{lien2014sequential} and \citet{sid2020}. Motivated by the distribution of food donations by mobile pantries,\footnote{For other examples of work in this application area, see \citet{solak2014stop, orgut2018robust}, and \citet{eisenhandler2019humanitarian}.} \citet{lien2014sequential} introduced the problem of  sequential resource allocation which coincides with our base model and the ex-post fairness objective function, in that it aims to maximize the expected minimum FR (although it only studies the special case of independent demands).
The recent work of \citet{sid2020} considers a similar model; however, it focuses on a multi-criteria objective which is based on an allocation's distance from the optimal
offline Nash Social Welfare solution. 
We note that their notion of fairness is also different in nature from ours.\footnote{{In \citet{sid2020}, their notion of fairness is with respect to the absolute allocation, i.e., if possible, agents' allocation should be equalized regardless of differences in their needs. In contrast, we aim for an allocation which is proportional to need.}}

The algorithmic aspects of both \citet{lien2014sequential} and \citet{sid2020} consist of designing novel heuristics and numerically evaluating them against a relevant benchmark (the intractable DP solution and the Nash Social Welfare solution, respectively). On the other hand, we take a theoretical approach and analyze fairness guarantees for the policies we design. Further, we provide upper bounds on the performance of any policy (including the DP solution), which serves a dual purpose: (i) it establishes that our policy is the best possible one if we aim to achieve both ex-ante and ex-post fairness guarantees, and (ii) it highlights the fundamental limits of achieving equity in a dynamic setting.

\revcolor{A related stream of papers in computer science study fair division problems in dynamic settings \citep{walsh2011online,  kash2014no, aleksandrov2015online}, albeit based on different motivating applications and with different objectives. Consequently, the dynamic aspects of these papers differ from our modeling approach. \citet{walsh2011online} studies a setting where agents arrive over time but allocation decisions are not required to be immediate. In a similar direction, \citet{kash2014no} and \citet{aleksandrov2015online} consider models where agents remain in the system and can receive multiple allocations. 
Beyond the model dynamics, these papers allow for an arbitrary sequence of arriving agents (e.g., an adversarial arrival model). Moreover, their results pertain to obtaining envy-free allocations, and hence do not rely on a direct comparison of agent's utilities. 
In contrast, in our setting agents' demands are drawn from a known (arbitrarily correlated) joint distribution, and our results center on direct comparisons based on agents' fill rates (which play the role of the utilities in our model).}

In settings with multiple types of resources, \citet{azar2010allocate} and \citet{bateni2016fair} study online versions of Fisher markets and develop policies with fairness guarantees under two different arrival models. The former assumes an adversarial  model whereas the latter considers demand that belongs to a general class of stochastic processes.\footnote{We remark that papers considering general convex objective functions, such as \citet{agrawal2014fast} and \citet{balseiro2020best}, admit many common fairness objective functions as special cases. See also \citet{mehta2013online} for more details.}
There are fundamental differences between our work and the aforementioned papers. 
Just to name one, 
the settings of \citet{azar2010allocate} and \citet{bateni2016fair} 
are motivated by online advertising, where demanding agents  (advertisers) are offline and items (impressions) arrive in an online fashion. 
Demanding agents have a large budget compared to the price of each arriving item, and they derive item-specific utilities. 
Consequently, the fairness notion is concerned with the total utility of each agent, which is a function of {\em all} items allocated to it during the horizon.
In contrast from such a setting, demanding agents in our work arrive in an online fashion while the supply side is offline, and each demanding agent receives a {\em single} allocation. The recent works of \citet{WillMa} and \citet{nanda2020balancing} are closer to our setting in that the demanding agents arrive online; however, they differ in several aspects: (i) the underlying arrival process is known i.i.d. where arriving demand belongs to various groups, (ii) they focus on group-level fairness, and (iii) {they consider} a matching setting, i.e., allocating indivisible goods. 

The objectives of ex-post and ex-ante fairness which we study in our problem bear some resemblance to the objective in the online contention resolution scheme (OCRS) problem, although the two problems are not directly comparable. The OCRS is basically a rounding algorithm that aims to uniformly preserve the marginals induced by a fractional solution while obtaining feasibility of the final allocation. This technique has found application in many settings such as Bayesian online selection, oblivious posted pricing mechanisms, and stochastic probing models (see, e.g., \citealt{alaei2014bayesian}, \citealt{feldman2016online}, and \citealt{lee2018optimal}). The OCRS problem diverges from ours because that setting focuses on designing randomized policies for allocating indivisible goods, while our focus is on divisible goods (consequently, restricting to deterministic policies is without loss).

\smallskip
\noindent {\bf Dynamic allocation of social goods:}
On a broader level, our paper is related to the literature on dynamic allocation of social goods and services, such as public housing, donated organs, and emergency care. 
Examples of centralized allocation policies include ~\citet{kaplan1984managing,ashlagi2013kidney, agarwal2019empirical}, and \citet{ashlagi2019matching}; examples of decentralized mechanisms are~\citet{dynamic-matching-waiting-lists-leshno, anunrojwong2020information}, and \citet{arnosti2017design}. For the most part, the aforementioned papers focus on the analysis of social welfare in steady-state models where both demand and supply dynamically arrive. 
We complement this literature by focusing on equitable allocation  in a non-stationary framework where a fixed amount of supply must be rationed across demand that arrives over time.

\revcolor{Further, our work is broadly related to the growing literature
on dynamic mechanism design without money  \citep{Peng-Santiago, gorokh2019remarkable,  Gorokh20}. 
These papers study settings with repeated interaction between a principal and agents, and they  assume that agents' valuations are drawn independently across individuals and across time. Our framework differs from such settings in several key aspects: (i) each agent only interacts once with the social planner,  (ii) agents' demands can be arbitrarily correlated, and (iii) as explained below, agents are non-strategic.}

\revcolor{In our work, we abstract away from strategic behavior and assume that agents do not control the timing of their demand (i.e., the order of their arrival) nor can they misrepresent their demand, either individually or as a coalition. 
{Several papers in the supply chain and social choice literature (see, e.g., \citealt{sprumont1991division,lee1997information,cachon1999capacity}) study incentive issues that arise when demand for a resource exceeds its capacity in various other contexts with no access to monetary mechanisms (as is the case in our setting).} In particular, \citet{lee1997information} shows that proportional allocation (which is the static version of our proposed PPA policy) can induce strategic behavior as agents may benefit from over-stating their demand.
However, such strategic considerations are often inapplicable in our motivating applications. In contexts such as a pandemic or a natural disaster, the sequence of realized demand is exogenous. Furthermore, demand is verifiable in these settings, and false reporting can be severely punished under the Disaster Fraud Act.\footnote{\revcolor{For one such example, see \href{https://www.oig.dhs.gov/news/press-releases/2019/02202019/manhattan-us-attorney-announces-53-million-proposed-settlement-lawsuit-against-new-york-city-fraudulently-obtaining}{https://www.oig.dhs.gov/news/press-releases/2019/02202019/manhattan-us-attorney-announces-53-million-proposed-settlement-lawsuit-against-new-york-city-fraudulently-obtaining}.}}}

\smallskip
\noindent {\bf Online resource allocation:}
From a technical point of view, our work is related to the rich literature on online resource allocation  and prophet
inequalities, which started from the seminal work of \citet{krengel1978semiamarts} and \citet{samuel1984comparison}. For an informative survey, we refer
the interested reader to \citet{lucier2017economic}. We highlight that in terms of modeling demand, our work departs from the prevailing approaches in this literature, namely adversarial, i.i.d., or random permutation arrival models. In our work, we assume that the sequence of demands can be arbitrarily correlated and the joint distribution is known in advance. 
In terms of modeling demand, our work is closest to a few papers that consider prophet inequalities with correlated demand \citep{10.2307/4355751, Corprophet, immorlica2020prophet}. However, the nature of the online decisions {is different; in our model, a fraction of a divisible good is allocated to each arriving demand, whereas in prophet inequality settings, an indivisible good is allocated to a single agent.}

Finally, our \PolicyNameAb\ policy relies on re-optimizing the FR by replacing all future random demands by their expected values. As such, it is related to the stream of papers in  revenue management and  dynamic programming that  theoretically analyze the performance of such heuristics. Great examples of work in this direction include \citet{ciocan2012model}, \citet{jasin2012re}, \citet{balseiro2019approximations}, and \citet{calmon2020revenue}.

\section{Model and Preliminaries}
\label{sec:prelim}

\smallskip
\noindent{\bf{Problem setup:}} Consider a planner that is using a sequential allocation policy---also referred to as an online policy---to allocate a divisible resource of supply $\Supply$ among $\NumAgents$ agents. Without loss of generality, we normalize the total supply so that $\Supply=1$. Agents arrive sequentially over time periods $1,2,\ldots,\NumAgents$, and we index agents according to the period in which they arrive. 
Once agent $i$ arrives, their demand $\Demandi{i}\in \NonNegReals$ is realized and observed by the planner. Based on the observed demand $\Demandi{i}$ and the history up to time period $i$, the sequential policy makes an irrevocable decision by allocating an amount $\Alloci{i}$ of the resource to this agent. The allocated amount $\Alloci{i}$ cannot exceed the agent's realized demand $\Demandi{i}$ nor can it exceed the remaining supply before agent $i$'s arrival, which we denote by $\Supplyi{i}$. Thus, $\Alloci{i}$ is a feasible allocation if $\Alloci{i}\in [0, \min\{s_i, d_i\}]$. 
Given the feasible allocation $\Alloci{i}$ and the demand $\Demandi{i}$, agent $i$'s \emph{fill rate} (FR) is defined as $ \frac{\Alloci{i}}{\Demandi{i}}$.\footnote{{If $\Alloci{i}=\Demandi{i}=0$, we set the FR to $1$ as a convention.}} 
After allocating $\Alloci{i}$ to agent $i$, the remaining supply before the arrival of agent $i+1$ is  $\Supplyi{i+1}=\Supplyi{i}-\Alloci{i}$.

To model the uncertainty about future demands, we consider a Bayesian setting where the $\Demandi{i}$'s are stochastic and arbitrarily correlated such that $\vec{\Demand}=(\Demandi{1},\Demandi{2},\ldots,\Demandi{\NumAgents}$) is drawn from a joint distribution $\DemandJoinDist\in\Delta\left({\NonNegRealsN}\right)$ \emph{known} by the planner. \revcolor{A key characteristic of a demand sequence is the expected demand-to-supply ratio, which we call the {\em supply scarcity} as defined formally below.
\begin{definition}[Supply Scarcity]
\label{def:scarcity}
The supply scarcity of a 
 demand sequence $\vec{\Demand}$ drawn from a joint distribution $\DemandJoinDist\in\Delta\left({\NonNegRealsN}\right)$ is given by:
 \begin{equation*}
     \ExpDemand\triangleq \frac{\expect[\vec{\Demand}\sim\DemandJoinDist]{\sum_{i\in[\NumAgents]}\Demandi{i}} }{s}.\footnote{For any $a \in \mathbb{N}$, we use $[a]$ to refer to the set $\{1, 2, \dots, a\}$.}
 \end{equation*}
Since we normalize the supply to be $1$, the supply scarcity is equal to the total expected demand.
 \end{definition}
}
For simplicity of presenting our results, we consider joint distributions that assign non-zero probability to at least one sample path of demands with $\Demandi{n}\neq 0$. Equivalently, we assume $\Demandi{n}$ is not deterministically equal to zero.\footnote{This assumption is without loss of generality, as one can alternatively re-define $n$ to be the smallest index such that $d_{n'}$ is deterministically equal to zero for $n'>n$.}

As detailed earlier, our setup is motivated by
the distributional operations of a governmental or nonprofit organization. Consequently, we focus on an egalitarian planer that intends to balance the equity and efficiency of the allocation. To this end, the planner's objective is to maximize the minimum achieved FR among the agents, i.e., $\min_{i\in[n]}\frac{\Alloci{i}}{\Demandi{i}}$, given the uncertainty in the demands. Maximizing such an objective has its roots in the classic literature on welfare economics (e.g.,~\citealt{arrow1963social}) and has been studied more recently in similar contexts in operations research~(e.g.,~\citealt{lien2014sequential}). It provides equity through its focus on the worst FR across all agents---in contrast to the sum of FRs---{and provides}  efficiency by aiming to maximize this FR---in contrast to  allocating  an equally minimal amount of the resource to all agents.\footnote{We consider a broader class of objectives that subsumes the minimum FR in \Cref{subsec:variants}.} 

\revcolor{Before introducing our objective function, we comment on our assumptions that the supply is fixed a priori and that we make a one-time allocation decision for each agent. In {some applications}, supply can get replenished over time, and there can be multiple allocations to {the same} agent. However, in our motivating applications there is a time urgency that we aim to incorporate into our model: as one example, during a pandemic it may take months to replenish the national stockpile of medical resources such as ventilators.\footnote{\revcolor{During the first peak of the COVID-19 pandemic in the US, it took a few months to {produce} ventilators, and there were concerns that those ventilators would be ready too late \citep{WaPo2020ventilators}.}} In the meantime, a ``pandemic wave'' may last only a few weeks, which makes any potential second shipment less valuable or even unnecessary. (For instance, the need for ventilators may greatly reduce a few weeks after the peak.) The same features exist in response to a natural disaster, as providing relief is most valuable in the immediate aftermath.}

\smallskip
\noindent{\bf Objectives \& fairness guarantees:} 
{Since demands are a priori uncertain in the setup described above, the planner should consider appropriate metrics to aggregate over uncertain outcomes. We now formally define the planner's objectives by considering two different metrics:}
the ex-post minimum FR and the ex-ante minimum FR. 
For any sequential allocation policy $\Policy$,
the ex-post minimum FR of policy $\Policy$ is its expected minimum FR, i.e.,
\begin{equation}
\label{eq:expost} \tag{ex-post}
\ExpostObj^{\DemandJoinDist}(\Policy)\triangleq\expect[\Vec{\Demandi{}}\sim\DemandJoinDist]{\underset{i\in[n]}{\min}\left\{\frac{\Alloci{i}}{\Demandi{i}}\right\}}~,
\end{equation}
where $\vec{\Alloci{}}=(\Alloci{1}, \Alloci{2}, \dots, \Alloci{\NumAgents})$ is the sequence of allocations generated by  $\Policy$. 
On the other hand, the ex-ante minimum FR of policy $\Policy$ is its minimum expected FR, i.e.,
\begin{equation}
\label{eq:exante}\tag{ex-ante}
\ExanteObj^{\DemandJoinDist}(\Policy)\triangleq\underset{i\in[n]}{\min}\left\{\expect[\Vec{\Demandi{}}\sim\DemandJoinDist]{\frac{\Alloci{i}}{\Demandi{i}}}\right\}~.
\end{equation}
For a randomized policy $\Policy$, we abuse notation and again use $\ExpostObj^{\DemandJoinDist}(\Policy)$ and $\ExanteObj^{\DemandJoinDist}(\Policy)$ to denote the expectation of the above two quantities over the policy $\Policy$'s internal randomness.\footnote{In principle, we allow randomization of our policies in this paper; however, as will be clear later, all of our proposed policies are deterministic and no randomization is needed to obtain our targeted performance guarantees.} 

These two objectives represent two different notions of fairness: \cref{eq:expost} aims for equity in outcomes, whereas \cref{eq:exante} aims for equity in \emph{expected} outcomes. We largely focus on the ex-post minimum FR for two main reasons. First,
when allocating supplies in response to a rare event like a pandemic or natural disaster, agents only observe one realized outcome. Because the ex-ante minimum FR is only concerned with marginal fairness, it can have unfair outcomes for \emph{every} {sample path}, i.e., every realized demand sequence. In contrast, the ex-post minimum FR considers each full sample path; every sample path with positive probability which results in an unfair outcome reduces $\ExpostObj^{\DemandJoinDist}(\Policy)$.
Second, by Jensen's inequality, the ex-post minimum FR serves as
a lower bound on the ex-ante minimum FR, i.e., for any policy $\Policy$,
$$\ExanteObj^{\DemandJoinDist}(\Policy) \geq \ExpostObj^{\DemandJoinDist}(\Policy).$$

However, for a fixed ex-post minimum FR, achieving a higher ex-ante minimum FR is desirable because it reduces systematic biases against a particular agent, e.g., the last-arriving agent. In the extreme case where $\ExanteObj^{\DemandJoinDist}(\Policy) = \ExpostObj^{\DemandJoinDist}(\Policy)$, one particular agent receives the smallest FR, {\em regardless} of the sample path. On the other hand, $\ExanteObj^{\DemandJoinDist}(\Policy) > \ExpostObj^{\DemandJoinDist}(\Policy)$ implies that 
the worst-off agent varies across different sample paths.

\revcolor{Having defined our notions of fairness, we first observe that 
if the sequence of demand is deterministic, then the policy that maximizes {\em both} the ex-post and the ex-ante minimum FR is simply equalizing all FRs. Namely, this policy achieves a FR equal to:
\begin{align}
    \max_{\pi} \Big\{\ExanteObj^{\DemandJoinDist}(\Policy)\Big\} = \max_{\pi} \Big\{ \ExpostObj^{\DemandJoinDist}(\Policy)\Big\} = \frac{\Alloci{1}}{\Demandi{1}} = \frac{\Alloci{2}}{\Demandi{2}} = \dots = \frac{\Alloci{n}}{\Demandi{n}} = \min\Big\{1, \frac{1}{\mu}\Big\}. \label{eq:det:fair}
\end{align}
Inspired by the above observation, we define a \emph{normalization factor} that will help us with providing more informative performance guarantees. 
\begin{definition}[Normalization Factor]
\label{def:normalization:factor}
Given a joint demand distribution $\DemandJoinDist\in\Delta\left({\NonNegRealsN}\right)$ with total expected demand of $\ExpDemand$, the normalization factor $\DetGuar\triangleq \min\Big\{1, \frac{1}{\mu}\Big\}$ represents the optimum
ex-post and  ex-ante minimum FR if $\DemandJoinDist$ is replaced by a deterministic distribution (over demand sequences) with an identical
total expected demand $\ExpDemand$.
\end{definition}

To see why we introduce the normalization factor $\DetGuar$, note that the above observation in \cref{eq:det:fair} highlights that \emph{even without stochasticity in the demand sequence}, when total demand exceeds supply we cannot guarantee a minimum FR better than $1/\ExpDemand$. This is simply due to the ``scarcity of supply.'' However, if the sequence of demands is stochastic (and possibly correlated), an ex-post or ex-ante minimum FR of $\DetGuar$ may not be always achievable by an online allocation decision maker due to the ``scarcity of information,'' i.e., the unknown realizations of future demands.
Consequently, we use $\DetGuar$ to enable us to decouple the effect of supply scarcity from the effect of stochasticity in the demand sequence on the quality of online allocation decisions. 

We emphasize that $\DetGuar$ is not a benchmark in the sense that it does not serve as an upper bound on the achievable ex-post (or ex-ante) minimum FR for all joint demand distributions with a total expected demand of $\mu$. Due to the stochastic nature of demand and convexity in our objectives, 
it is possible that 
when taking expectation over the realized demand sequence, we can achieve  ex-post and  ex-ante minimum FR's that are higher than $\DetGuar$. However, as explained above, $\DetGuar$ serves as an informative normalization factor to decouple the effects of supply scarcity and information scarcity.}\footnote{\revcolor{
In \Cref{rem:offline} and \Cref{prop:offlinecomparison} of \Cref{subsec:expostupperbound}, we explain that the offline solution is too powerful to allow for a constant-factor competitive ratio. Hence, we do not use that solution as a benchmark.}}

Consequently, we evaluate policies based on how they perform relative to $\DetGuar$. For a policy $\Policy$ and a joint demand distribution $\DemandJoinDist$, we say that the policy achieves \emph{ex-post fairness} (resp. \emph{ex-ante fairness}) of ${\ExpostObj^{\DemandJoinDist}(\pi)}/{\DetGuar}$ (resp. ${\ExanteObj^{\DemandJoinDist}(\pi)}/{\DetGuar}$). 
We aim to design a policy with guarantees on both ex-post and ex-ante fairness that hold universally for all joint demand distributions $\DemandJoinDist$ with $n$ agents and supply scarcity $\ExpDemand$. 
We refer to the universal lower bounds of a policy $\Policy$ as its \emph{fairness guarantees}, which we {formally} define below.

\begin{definition}[\textbf{Ex-post/Ex-ante Fairness Guarantee}]
\label{def:polperf}
A sequential allocation policy $\Policy$ achieves an \emph{ex-post fairness guarantee} (resp. \emph{ex-ante fairness guarantee}) of $\LowerboundFunP$ (resp. $\LowerboundFunA$), if for all $\NumAgents\in\mathbb{N}$ and $\ExpDemand\in \NonNegReals$,
\begin{align*}
\underset{\DemandJoinDist\in\JointDistStateSpace}{\inf} \frac{\ExpostObj^{\DemandJoinDist}(\pi)}{\DetGuar}~\geq ~\LowerboundFunP~~~~~\Bigg(\textrm{resp.}~~~~\underset{\DemandJoinDist\in\JointDistStateSpace}{\inf}\frac{\ExanteObj^{\DemandJoinDist}(\pi)}{\DetGuar}~\geq~\LowerboundFunA\Bigg)~,
\end{align*}
{where $\JointDistStateSpace$ denotes the domain of joint demand distributions with $n$ agents and {total expected demand of} $\mu$.}
\end{definition}

Our goals are (i) to understand the limits of achieving fairness in sequential allocation by computing upper bounds on the achievable guarantees, and (ii) to obtain tight lower bounds by designing policies with strong ex-post guarantees as well as ex-ante guarantees. {We show in \Cref{sec:model} that no gap exists between the achievable upper and lower bounds under both ex-post and ex-ante notions. More specifically, we show how to obtain {\em{exactly matching}} upper and lower bounds for both notions of fairness using a {\em single} adaptive policy.}

\section{Optimal Bounds on Fairness Guarantees}
\label{sec:model}
\label{sec:mainPPA}
In this section, we present our main results for the setting introduced in \Cref{sec:prelim}. First, we focus on ex-post fairness in \Cref{subsec:expostupperbound} and establish parameterized upper bounds on the ex-post fairness guarantee achievable by any sequential allocation policy---whether adaptive or non-adaptive, computationally efficient (i.e., with polynomial running time) or not.
Then, somewhat surprisingly, we show that such upper bounds can be achieved by our policy, which is introduced and analyzed in  Sections \ref{subsec:ourpolicy} and \ref{subsec:ex-post-lower}. 
Next, to illustrate the power of our simple adaptive algorithm, in \Cref{subsec:AdaptivityGap} we characterize the ex-post fairness guarantee of the best policy which non-adaptively aims for a particular target fill rate, and we show that our policy performs favorably compared to such a policy. Finally, in Section \ref{subsec:ex-ante} we turn our attention to the notion of ex-ante fairness, and we show that our policy also achieves the best possible ex-ante fairness guarantee.

\subsection{Upper Bound on Ex-post Fairness Guarantee }
\label{subsec:expostupperbound}
We begin this section by establishing a fundamental limit on ex-post fairness for any allocation policy when faced with stochastic and sequential demands. The main result of this subsection is the following theorem: 

\begin{theorem}[Upper Bound on Ex-post Fairness Guarantee]
\label{thm:hardness-ex-post}
Given a fixed number of agents $\NumAgents\in\mathbb{N}$ and supply scarcity $\ExpDemand\in \NonNegReals$, no sequential allocation policy obtains an ex-post fairness guarantee (see \Cref{def:polperf}) greater than $\LowerboundFunEx$, defined as
\begin{equation}
\label{eq:kappa-expost}
   \LowerboundFunEx \triangleq\begin{cases}
   1-(\frac{\NumAgents}{2(\NumAgents+1)})\ExpDemand, &\ExpDemand \in [0, 1) \\
\ExpDemand-(\frac{\NumAgents}{2(\NumAgents+1)})\ExpDemand^2, &\ExpDemand \in [1, \frac{\NumAgents+1}{\NumAgents})\\
\frac{\NumAgents+1}{2\NumAgents}, &\ExpDemand \in [\frac{\NumAgents+1}{\NumAgents}, +\infty)
\end{cases}.
\end{equation}
\end{theorem}

\begin{figure}[t]
    \centering
    \includegraphics{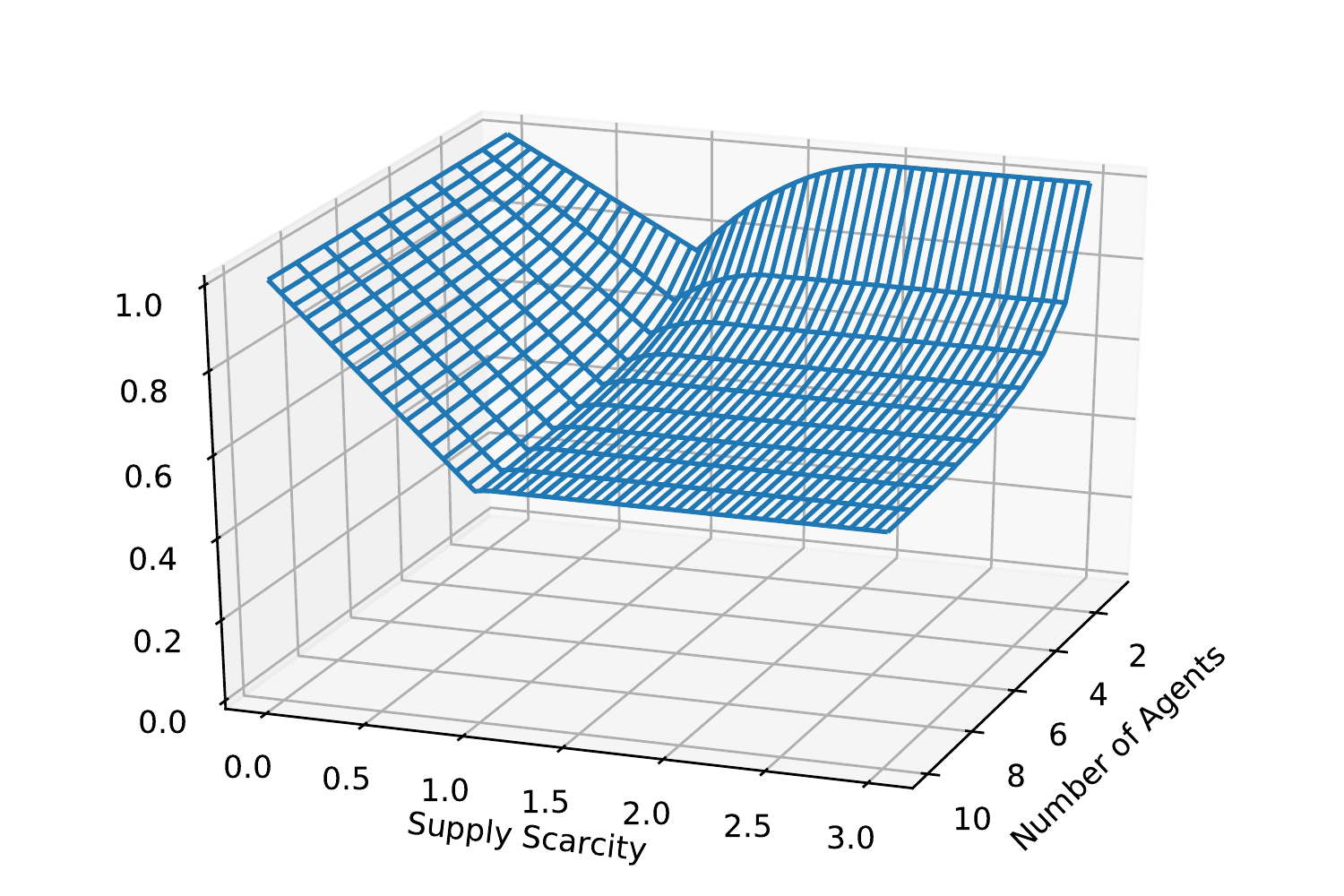}
    \caption{The upper bound on the ex-post fairness guarantee of any policy, as a function of the supply scarcity and the number of agents.}
    \label{fig:upperbound3d}
\end{figure}

See \Cref{fig:upperbound3d} for an illustration of this upper bound as a function of the supply scarcity $\ExpDemand$ and the number of agents $\NumAgents$. 
Per \Cref{def:polperf}, the ex-post fairness guarantee is relative to the achievable minimum FR when demands are deterministic, namely $\DetGuar = \min\{1, 1/\ExpDemand\}$. Consequently, this upper bound provides insight into the unavoidable loss in efficiency and equity when demands are a priori uncertain and realize sequentially.
In particular, we remark that the achievable fairness guarantee crucially depends on 
the supply scarcity. In the regime where {$\mu <1 + \frac{1}{\NumAgents}$, which we refer to as the {\em under-demanded regime}},   $\LowerboundFunEx$ {initially} worsens as $\mu$ increases {before hitting its minimum (for any fixed $\NumAgents$) when expected demand equals supply, i.e., at $\ExpDemand = 1$.} This suggests that the stochastic nature of demand is most harmful when expected demand exactly equals supply.
On the other hand, in the {\em over-demanded regime} where $\mu \geq 1+\frac{1}{\NumAgents}$, the achievable fairness guarantee is independent of $\mu$.
Given that we are usually in the over-demanded regime in our motivating applications, Theorem \ref{thm:hardness-ex-post} ensures that supply scarcity does not contribute to the loss in fairness due to uncertain, correlated, and sequential demand.
Fixing $\mu$, the upper bound always decreases with $n$, implying that achieving fairness can be more challenging for a larger population of agents with stochastic demands, even if the total expected demand of the population remains the same. 
Finally, we highlight  that the bound is always at least $1/2$ regardless of the supply scarcity and the number of agents, and it attains its minimum when $\ExpDemand =1$  and $n \rightarrow +\infty$.

The proof of Theorem \ref{thm:hardness-ex-post} relies on establishing two hard instances with similar structures, one for $\ExpDemand < 1+1/\NumAgents$ and one for $\ExpDemand \geq 1+1/\NumAgents$. The details of the proof are presented in Appendix \ref{apx:upper-bound}. 
Here, we present the instance for the over-demanded regime along with a sketch of our analysis.
In this instance, there are $\NumAgents$ possible equally-likely scenarios, i.e.,  scenario $\sigma$ happens with probability $1/\NumAgents$ for $\sigma \in [\NumAgents]$. 
In scenario $\sigma$, the first $\sigma$ agents have equal demand of $\frac{2\ExpDemand}{\NumAgents+1}$ and the rest have no demand.
We illustrate this instance in \Cref{fig:hardexample}.\footnote{We remark that similar settings can occur in practice. As one example, consider the challenge of allocating limited disaster-relief supplies to towns damaged by a hurricane which may continue on its destructive path or may veer back out to sea.}

First, note that the total supply scarcity for the above hard instance is $\ExpDemand$ (as shown in Appendix \ref{apx:upper-bound}). 
Next, consider any sequential policy that faces a non-zero demand from agent $i$. The policy cannot distinguish among possible scenarios ${i, i+1, \ldots, n}$. Consequently, its allocation decision for agent $i$ will be independent of the scenario. In light of this observation, 
any policy can be sufficiently described by a set of (possibly random) allocations with expected values $\vec{\AppAlloc} = (\AppAlloci{1}, \AppAlloci{2}, \dots, \AppAlloci{n})$, such that if agent $i$ has non-zero demand, then they receive an expected allocation $\AppAlloci{i}$. 
Given $\vec{\AppAlloc}$, the minimum FR for scenario $\sigma$ is 
\begin{equation}
    \label{eq:hard-instance-inequ}
    r_\sigma \triangleq \frac{(n+1)}{2 \mu}\expect[\Policy]{ \min\{\Alloci{1}, \Alloci{2}, \ldots,  \Alloci{\sigma}\}} \leq \frac{(n+1)}{2 \mu}\ \min\{\AppAlloci{1}, \AppAlloci{2}, \ldots,  \AppAlloci{\sigma}\}\leq \frac{(n+1)}{2 \mu} \AppAlloci{\sigma}~,
\end{equation}
where the first inequality is due to the expectation of a minimum being less than the minimum over expectations (Jensen's inequality).

In order to establish our upper bound, we set up a factor-revealing linear program as presented in \Cref{fig:hardnessLP}. The LP maximizes the expected minimum FR subject to three sets of natural constraints that must hold for any sequential policy:
\begin{itemize}
    \item The minimum FR in scenario $\sigma$ cannot exceed the FR for agent $\sigma$, as shown in \cref{eq:hard-instance-inequ}.
    \item The minimum FR in scenario $\sigma$ is at most the minimum FR in scenario $\sigma -1$.
    \item The total amount of expected allocations cannot exceed the available supply of $1$.
\end{itemize}
In Appendix \ref{apx:upper-bound}, we provide an upper bound on the optimal value of this LP by presenting a feasible solution to its dual. To complete the proof of \Cref{thm:hardness-ex-post}, we must scale by $\DetGuar = \min\{1, 1/\ExpDemand\}$ to translate this upper bound on the expected minimum FR into an upper bound on ex-post fairness (see \Cref{def:polperf}).

\begin{figure}[t]
 \centering
 \begin{subfigure}[b]{0.59\textwidth}
     \includegraphics[trim={11cm 12cm 16cm 8cm},clip,width=\textwidth]{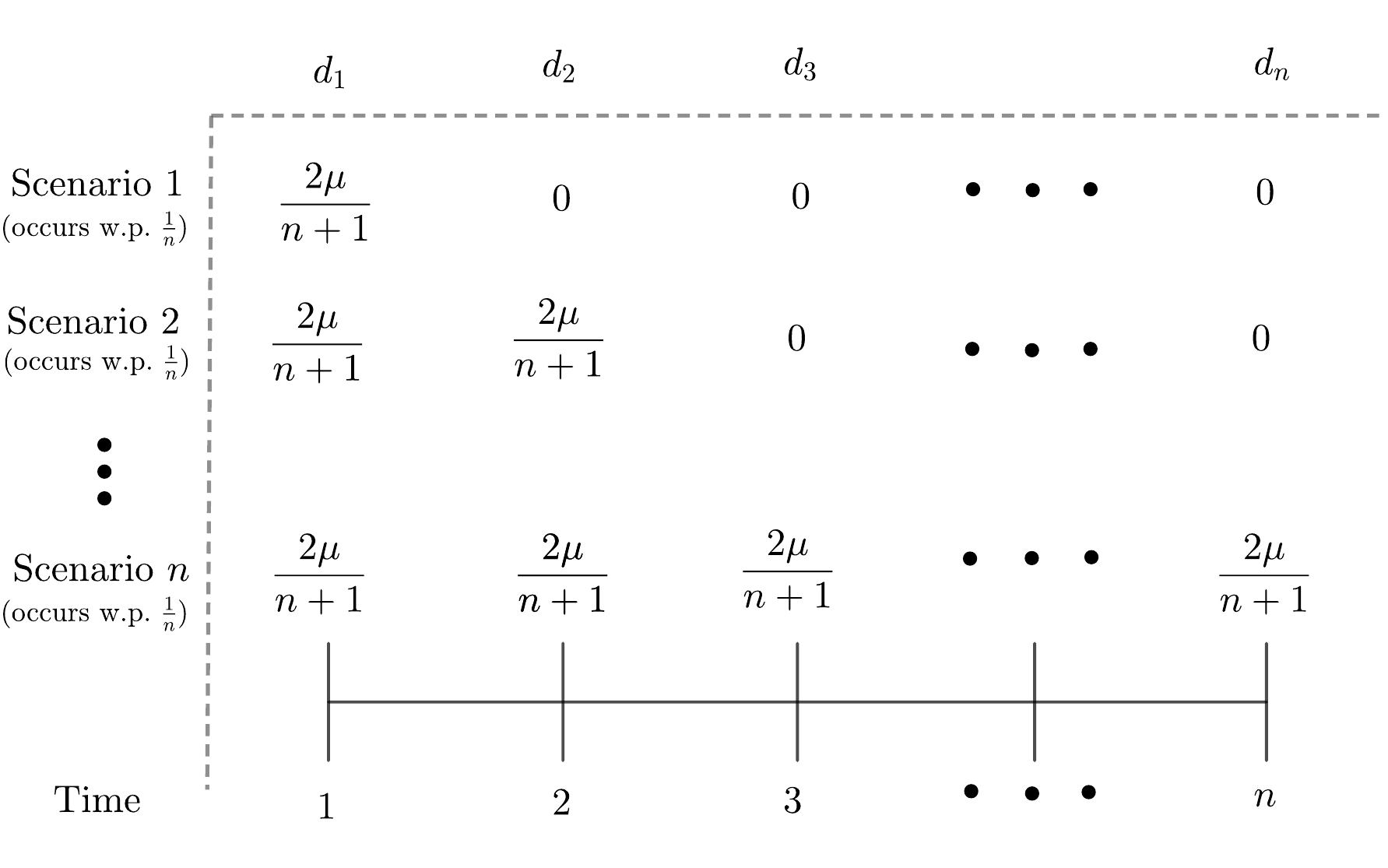}
    \caption{}
    \label{fig:hardexample}
 \end{subfigure}
 \begin{subfigure}[b]{0.39\textwidth}
 \includegraphics[trim={6.5cm 17.5cm 6.5cm 3cm},clip,width=.95\textwidth]{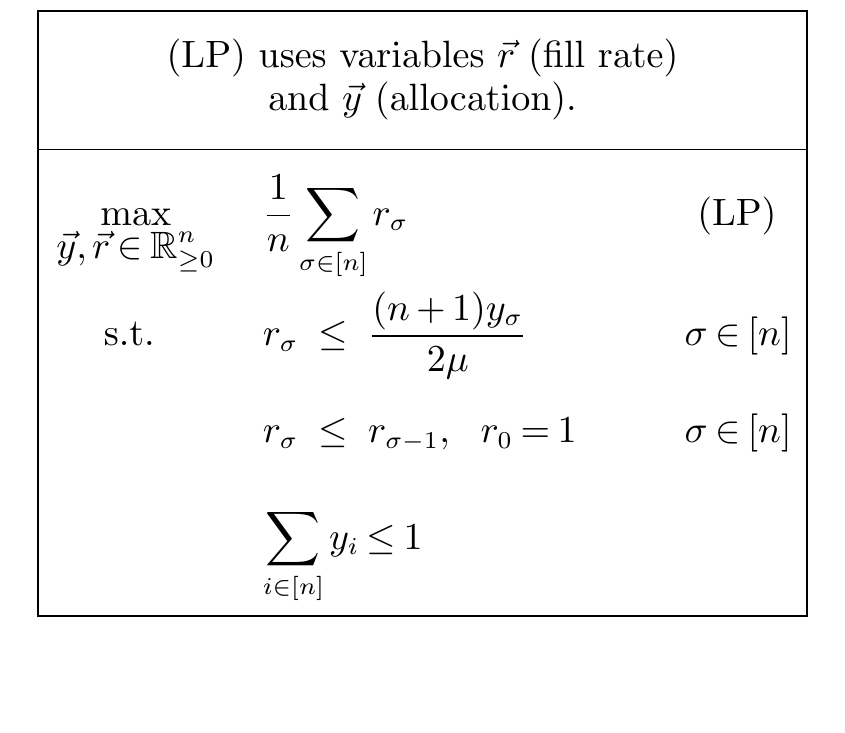}
    \caption{}
    \label{fig:hardnessLP}
 \end{subfigure}
 \caption{(a) The instance and (b) the factor-revealing LP which establish an upper bound of $\bm{\LowerboundFunEx}$ for the over-demanded regime (i.e., when $\bm{\ExpDemand \geq 1 +\frac{1}{\NumAgents}}$).}
\end{figure}

\revcolor{We finish this subsection with two important remarks regarding (i) the use of the offline solution as a benchmark and (ii) the shortcomings of the optimum online solution.
\begin{remark}[Comparison with Offline Solution]
\label{rem:offline}
In sequential decision making problems, it is common to evaluate the performance of a policy by comparing it to the {\em offline solution}, i.e., the optimum solution that observes the entire demand sequence before making any decisions (see, e.g., \citealt{mehta2013online} and references therein). However, in our setting, such an offline solution proves to be too powerful, in the sense that it is impossible to achieve a constant-factor guarantee compared to such a solution.  In the following proposition (proven in Appendix \ref{apx:offline}), we use the same example as shown in \Cref{fig:hardexample} to establish this impossibility result. 
\begin{proposition}[Comparison with Offline Solution]
\label{prop:offlinecomparison}
Given a fixed number of agents $\NumAgents\in\mathbb{N}$, there exists a supply scarcity $\ExpDemand\in \NonNegReals$ such that no sequential allocation policy can guarantee more than a $\frac{1}{\log(\NumAgents+1)}$ fraction of the expected minimum fill rate achieved by the offline solution.
\end{proposition}
\end{remark}
In light of the above proposition, we focus on establishing absolute guarantees on  the expected minimum fill rate. However, as explained in Section \ref{sec:prelim} (see Definition \ref{def:normalization:factor} and its related discussion), we use the normalization factor $\DetGuar$ to disentangle the effects of supply scarcity from the effects of sequential decision-making when faced with a stochastic demand sequence.} 
\revcolor{
\begin{remark}[Optimum Online Solution]
\label{rem:online}
A natural candidate for a policy that may achieve the upper bound on the ex-post fairness guarantee -- given by \cref{eq:kappa-expost} -- is the {\em optimal online policy} which can be found via a DP.
 In Appendix \ref{apx:DP:details}, we formally present the underlying DP. 
 However, we also illustrate that there are significant limitations and drawbacks to a DP approach for maximizing the expected minimum FR in this setting. First, 
 (i) as we show in Appendix \ref{apx:DP:hardness-input-model},
 the state space of such a DP is exponentially large for arbitrarily correlated demands, which makes the DP intractable (in particular, the state space is exponential in the number of agents $n$). Nevertheless, in Appendix \ref{apx:DP:independent} we  present an FPTAS for the special case of independent demand.
 Even beyond computational challenges,  (ii) solving the DP requires full distributional knowledge, 
 (iii) such a DP solution does not necessarily perform well for our second objective function, i.e., maximizing the ex-ante minimum FR, and (iv) the DP decisions may lack transparency and interpretability, which are highly desirable properties in our motivating applications. (For an illustration of points (iii) and (iv), see \Cref{ex:DPBad} in Appendix \ref{apx:DPexample}; for a summary of the drawbacks of the DP, see \Cref{table:DP:Sucks}).
 \end{remark}}

Remarkably, in the following subsection, we design a simple adaptive policy that not only achieves the best possible ex-post fairness guarantee of $\LowerboundFunEx$, but also offers several corresponding advantages over a DP solution:
(i) it can be computed efficiently, (ii) it only requires knowledge of the conditional first moments of agents' demands, and (iii) its decisions can be clearly explained. Additionally, as shown in \Cref{subsec:ex-ante}, 
it simultaneously attains the best-possible ex-ante fairness guarantee.

\subsection{Projected Proportional Allocation  Policy}
\label{subsec:ourpolicy}
We introduce our policy, referred to as the \emph{\PolicyName}~(\PolicyNameAb) policy,
through the following simple intuition. Consider a planner that (magically) has access to all the demand realizations 
$\DemandVec$.
As already discussed in \Cref{sec:prelim},
to maximize the minimum FR {when the demand realizations are known a priori}, the planner should equalize the FR of all agents by allocating $\Alloc^*_i=\min\left(\Demand_i,\frac{\Demandi{i}}{\sum_{j\in[n]}\Demand_j}\right)$ to each agent $i$.
If  $\sum_{j\in[\NumAgents]}\Demandi{j}$ is at most the initial supply (which we normalize to $1$), then each agent $i$ obtains a full allocation of $\Alloci{i}=\Demandi{i}$ in such a solution. This results in the maximum equal FR of $1$. Otherwise, all the agents will have an equal FR of $1/\sum_{j\in[\NumAgents]}\Demandi{j}$, which is $1/\ExpDemand$ when each demand is equal to its expected value.

This solution can alternatively be obtained by solving a DP that returns allocations $\Alloc^*_\NumAgents,\Alloc^*_{\NumAgents-1},\ldots,x^*_1$ maximizing the minimum FR. By a simple induction argument, given the remaining supply $\Supplyi{i}$ at period $i$, this DP maintains the following invariant at each period $i$ (refer to Appendix~\ref{apx:offline:DP} for details): 
\begin{equation}
\Alloc^*_i=\min\left\{\Demandi{i},\Supplyi{i}\frac{\Demand_i}{\Demand_i+\sum_{j \in [i+1:\NumAgents]}\Demandi{j}}\right\}=\min\left\{\Demandi{i},\Supplyi{i}\frac{\Demand_i}{\Demand_i+\left[\textrm{total future demand}\right]}\right\}~.\footnote{For any $a, b \in \mathbb{N}$ we use $[a:b]$ to refer to the set $\{a, a+1, \dots, b\}$ if $a \leq b$ (and the empty set otherwise).} \label{eq:PPAinvariant}
\end{equation}
Notably, the above invariant suggests a \emph{sequential} implementation of the optimal solution at each period $i$ that only uses the knowledge of $\Demandi{i}$ (i.e., the current demand at period $i$) and $\sum_{j\in [i+1:n]}\Demandi{j}$ (i.e., the total future demand from period $i+1$ to $n$). {Now consider a setting with incomplete information, namely, with only knowledge of the current sample path of the observed demands up to period $i$, which we denote by $\DemandVec_{[1:i]} \triangleq (\Demandi{1}, \Demandi{2}, \dots, \Demandi{i})$. Our \PolicyNameAb\ policy implements a version of the above policy by replacing the exact realization of total future demand with the conditional first moment of this random variable given the current sample path. More precisely:}

\mybox{
\begin{displayquote}
\begin{itemize}
    \item Given the remaining supply $\Supplyi{i}$, the \PolicyNameAb~policy allocates an amount
\begin{equation}
\tag{\textsc{\PolicyNameAb's update rule}}
  \displaystyle  x_i = \min\left\{\Demand_i, \Supply_i \frac{\Demand_i}{\Demand_i + \ExpDemand_{i+1}}\right\}
\end{equation}
of the (divisible) resource to agent $i$ upon their arrival, where 
\begin{equation*}
    \quad \quad ~\ExpDemand_{i+1} \triangleq \expect[\vec{\Demand}\sim\DemandJoinDist]{\ \sum_{j \in [i+1: \NumAgents]}\Demandi{j} \ \Big| \  \DemandVec_{[1:i]}} \qquad \qquad \qquad \qquad \qquad \qquad
\end{equation*}
\end{itemize}
\end{displayquote}
}

Note that the conditional expected future demand $\ExpDemand_{i+1}$ given all previously-realized demands $\DemandVec_{[1:i]}$ is a function of  $\DemandVec_{[1:i]}$; however, for ease of notation, we use $\ExpDemand_{i+1}$ without any input arguments.

We highlight that the \PolicyNameAb~policy is simple, computationally efficient, and solely uses first-moment knowledge about the future demands. \revcolor{Consequently, the \PolicyNameAb~policy does not need to know the order of future arrivals.}
Further, because the allocation decisions of the \PolicyNameAb\ policy depend smoothly on the first moment of future demand, these decisions are robust to small changes in the scale of any marginal distribution.
Yet, as we show in Sections \ref{subsec:ex-post-lower} and \ref{subsec:ex-ante}, this simple policy remarkably achieves the best possible guarantee for both notions of fairness (ex-post  and ex-ante), even though these two notions are quantitatively different whenever $\NumAgents > 1$.

We now point out an important technical property of the \PolicyNameAb\ policy that will help us in establishing these key results.

\begin{remark}
\label{rem:never-run-out}
The \PolicyNameAb~policy can only run out of supply at the end of period $i$ if $\ExpDemand_{i+1}=0$, or equivalently, only if all future demands $\DemandVec_{[i+1:n]}$ are deterministically equal to zero, conditional on the current realized sample path of demands $\DemandVec_{[1:i]}$. This property holds simply because 
\begin{equation*}
    s_{i+1}=s_i-x_i\geq s_i\frac{\ExpDemand_{i+1}}{\Demand_i+\ExpDemand_{i+1}}.
\end{equation*}
\end{remark}

\revcolor{To conclude this section, we note that one can consider a policy similar to our \PolicyNameAb\ policy but with a slightly different updating rule that is monotone non-increasing in the fill rate. In particular, suppose the allocation at any time $i$ is given by
$x_i = \min\left\{f_{i}  \Demand_i, \Supply_i \frac{\Demand_i}{\Demand_i + \ExpDemand_{i+1}}\right\}$, where $f_i \in [0,1]$ is the minimum FR before the arrival of agent $i$. In words, this alternative policy ensures that in each time period  $i$ we do not have a fill rate  larger than the current minimum FR. While such an alternative policy achieves weakly larger ex-post fairness than our \PolicyNameAb\ policy, the two policies provide an identical ex-post fairness \emph{guarantee} (see \Cref{subsec:ex-post-lower}). Furthermore, this alternative policy suffers from a significant drawback: by definition, it systematically disfavors late-arriving agents. In fact, the minimum FR under this policy is always attained by the last agent (agent $n$) as long as that agent's demand is non-zero, which leads to a sub-optimal ex-ante fairness guarantee. 
In sharp contrast, our \PolicyNameAb\ policy
avoids systematically disfavoring late-arriving agents; consequently, in \Cref{subsec:ex-ante} we show that our \PolicyNameAb\ policy achieves the best possible ex-ante fairness guarantee.
}

\subsection{Ex-post Fairness of \PolicyNameAb\ Policy}
\label{subsec:ex-post-lower}
In this section, we analyze the ex-post fairness guarantee of our \PolicyNameAb~policy. In the following theorem, we show that this simple policy indeed achieves the best possible ex-post fairness guarantee. 
\begin{theorem}[Ex-post Fairness Guarantee of \PolicyNameAb\ Policy]
\label{thm:expost}
Given a fixed number of agents $\NumAgents\in\mathbb{N}$ and supply scarcity $\ExpDemand\in \NonNegReals$, the \PolicyNameAb ~policy achieves an ex-post fairness guarantee (see \Cref{def:polperf}) of at least $ \LowerboundFunEx$ 
(defined in \cref{eq:kappa-expost}).
\end{theorem}

\subsubsection{\texorpdfstring{Proof of \Cref{thm:expost}}{}} In order to prove the above theorem, we would have liked to analyze the evolution of the minimum FR, which we denote with $f_i$ at the end of period $i-1$, i.e., $f_1=1,~f_{i}=\min\{f_{i-1},\frac{x_{i-1}}{d_{i-1}}\}$ for $i\in[2:n+1]$. Instead, we consider the evolution of a closely related stochastic process, which makes the analysis simpler. We define this surrogate stochastic process as follows:
\begin{align}
    \label{eq:beta:def}
    \beta_1 \triangleq \min\left\{1, \frac{n+1}{n\dvar}\right\} ~~~~
    \beta_{i} \triangleq \min\left\{\beta_{i-1}, \frac{x_{i-1}}{d_{i-1}}\right\}, ~~ i\in[2:n+1].
\end{align}
First, we note that $\beta_i = \min\{f_i,\frac{\NumAgents+1}{\NumAgents\dvar}\}$,  $i \in [n+1]$. Next, recall that $\Supplyi{i}$ denotes the remaining supply after agent $i-1$ arrives and receives an allocation. We observe that under the \PolicyNameAb ~policy, $\Supplyi{i}$ evolves according to 
\begin{equation}
    \label{eq:supply:def}
    \Supply_1 = 1 ~~~~ \Supplyi{i} = \Supplyi{i-1} - \min\left\{\Demand_{i-1}, \frac{\Demand_{i-1}}{\Demand_{i-1} + \ExpDemand_{i}}\Supplyi{i-1}\right\}, ~~  i \in [2:n].
\end{equation}
With the above observations, the main step of the proof is carefully analyzing the evolution of $(\beta_k, s_k)$ under the \PolicyNameAb\ policy, which enables us to lower bound the final expected minimum FR in the following lemma.

\begin{lemma}[Lower Bound on Expected Minimum FR]
\label{lem:minfilllowerbound}
Under the \PolicyNameAb ~policy,
for all $i \in [\NumAgents+1]$ and any subsequence of demand realizations $\DemandVec_{[1:i-1]}$,
\begin{equation}
\label{eq:invariant}
\expect[\vec{\Demand}\sim\DemandJoinDist]{f_{n+1}\condition \DemandVec_{[1:i-1]}}\geq  \beta_i\left(1 - \frac{n+1-i}{2(n+2-i)}\frac{\ExpDemandi{i}}{\Supplyi{i}} \beta_i \right)
\end{equation}
where $\beta_i$ is defined in \cref{eq:beta:def}.\footnote{If $\Supplyi{i} = 0$, then by Remark  \ref{rem:never-run-out}, we must also have $\ExpDemandi{i} = 0$. In such cases, we take the convention that $\frac{\ExpDemandi{i}}{\Supplyi{i}} = 0$.} 
\end{lemma}

 Since the objective of our dynamic decision-making problem has no per-stage rewards and consists only of a terminal reward (i.e., the minimum FR), \Cref{lem:minfilllowerbound} can be thought of as establishing a lower bound on the value-to-go function of the \PolicyNameAb\ policy. Before providing the proof for this key lemma, we lay out the two remaining steps that finish the proof of \Cref{thm:expost}: (i) plugging $i=1$ into inequality \eqref{eq:invariant} to obtain a lower bound on $\expect[\vec{\Demand}\sim\DemandJoinDist]{f_{n+1}}$, and (ii) scaling the obtained lower bound result by our normalization factor, namely $\DetGuar = \min\{1, 1/\ExpDemand\}$, which provides an ex-post fairness guarantee (see \Cref{def:polperf}). 

\begin{proof}{{\bf Proof of \Cref{lem:minfilllowerbound}:}} 
We will show that inequality \eqref{eq:invariant} holds via backwards induction. The base case of $i = \NumAgents+1$ is trivial 
as it follows from the observation we made earlier: $\beta_i = \min\{f_i,\frac{\NumAgents+1}{\NumAgents\dvar}\}$.

Now let us consider $i=k < n+1$. Instead of proving inequality \eqref{eq:invariant}, we prove a stronger result:
\begin{equation}
\label{eq:invariant2}
\expect[\vec{\Demand}\sim\DemandJoinDist]{f_{n+1}\condition \DemandVec_{[1:k]}}\geq  \beta_k\left(1 - \frac{n+1-k}{2(n+2-k)}\frac{\Demandi{k} + \ExpDemandi{k+1}}{\Supplyi{k}} \beta_k \right).
\end{equation}
Establishing inequality \eqref{eq:invariant2} means that the inequality in \eqref{eq:invariant} holds for any realization of agent $k$'s demand. Consequently, it will hold when we take an expectation over agent $k$'s demand.
In order to prove inequality \eqref{eq:invariant2}, we consider two different cases that can arise depending on the remaining supply  $\Supplyi{k}$, agent $k$'s demand $\Demandi{k}$,
and the future expected demand $\ExpDemandi{k+1}$.
In the following, we introduce and analyze these cases separately.

\begin{enumerate}[label=(\roman*)]
    \item {\em Sufficient supply ($\Supplyi{k} \geq \beta_k(\Demand_k + \ExpDemandi{k+1})$)}: Recall that according to the \PolicyNameAb\ policy, $\Alloci{k} = \min\{\Demand_k, \frac{\Demand_k}{\Demand_k + \ExpDemand_{k+1}}\Supplyi{k}\}$. Therefore, in this case, either $\Alloc_k = \Demand_k$, i.e., the \PolicyNameAb\ policy meets the entire demand, or $\Alloci{k} =  \frac{\Demand_k}{\Demand_k + \ExpDemand_{k+1}}\Supplyi{k} \geq \beta_k \Demandi{k}$, i.e.,  the \PolicyNameAb\ policy attains an FR of at least $\beta_k$. According to the dynamics specified in \eqref{eq:beta:def} and \eqref{eq:supply:def},
    this implies
    \begin{align*}
       \beta_{k+1} = \beta_k ~~~~ \text{and}~~~~ \frac{\ExpDemand_{k+1}}{s_{k+1}} = \frac{\ExpDemand_{k+1}}{s_k - \min\{\Demand_k, \frac{\Demand_k}{\Demand_k + \ExpDemand_{k+1}}\Supplyi{k}\}} \leq \frac{\ExpDemand_{k+1}}{s_k -  \frac{\Demand_k}{\Demand_k + \ExpDemand_{k+1}}\Supplyi{k}}=\frac{\Demandi{k}+\ExpDemand_{k+1}}{s_k}.
    \end{align*}
    Using our inductive hypothesis when $i = k+1$,
    \begin{equation}
    \label{eq:bound1}
    \expect[\vec{\Demand}\sim\DemandJoinDist]{f_{n+1}\condition \DemandVec_{[1:k]}}\geq \beta_k\left(1-\frac{n-k}{2(n+1-k)}\frac{ \Demand_k + \ExpDemand_{k+1}}{\Supply_k} \beta_k\right)\triangleq\textsc{RHS}^{(1)}.
    \end{equation}
The lower bound given by $\textsc{RHS}^{(1)}$ is a linear function of $\Demand_k + \ExpDemand_{k+1}$, as illustrated by the dotted red lines in all panels of \Cref{fig:expostproof} (in the regime where $\Demand_k + \ExpDemand_{k+1}\in[0, s_k/\beta_k]$). This linear function has a non-positive slope and an intercept of $\beta_k$. We can further lower bound this function for any $\Demand_k + \ExpDemand_{k+1}\in[0,s_k/\beta_k]$ by another linear function with the same intercept of $\beta_k$ and a smaller (more negative) slope. In particular, since $\frac{n-k}{n+1-k} \leq \frac{n+1-k}{n+2-k}$, we have: 
\begin{equation} 
\label{eq:bound1and2repeat}
\textsc{RHS}^{(1)} \geq \beta_k\left(1-\frac{n+1-k}{2(n+2-k)}\frac{ \Demand_k + \ExpDemand_{k+1}}{\Supply_k}\beta_k\right),
\end{equation}
which proves inequality \eqref{eq:invariant2} in the sufficient supply case (see the blue lines in all panels of \Cref{fig:expostproof}).

\item {\em Insufficient supply ($\Supplyi{k} < \beta_k(\Demand_k + \ExpDemandi{k+1})$)}: 
    In this case, the allocation of the \PolicyNameAb~policy is $\Alloci{k} = \frac{\Demand_k}{\Demand_k + \ExpDemand_{k+1}}\Supplyi{k}$, which results in an FR less than $\beta_k$, i.e., $\frac{\Alloci{k}}{\Demandi{k}} = \frac{\Supply_k}{\Demand_k + \ExpDemand_{k+1}} < \beta_k$. According to the dynamics specified in \eqref{eq:beta:def} and \eqref{eq:supply:def},
    this implies
    \begin{align*}
       \beta_{k+1} = \frac{\Supply_k}{\Demand_k + \ExpDemand_{k+1}} ~~~~ \text{and}~~~~ \frac{\ExpDemand_{k+1}}{s_{k+1}}=\frac{\ExpDemand_{k+1}}{s_k - \frac{\Demand_k}{\Demand_k + \ExpDemand_{k+1}}\Supplyi{k}} =  \frac{\Demandi{k}+\ExpDemand_{k+1}}{s_k}.
    \end{align*}
Using our inductive hypothesis when $i = k+1$,
\begin{equation}
    \label{eq:bound3}
    \expect[\vec{\Demand}\sim\DemandJoinDist]{f_{n+1}\condition \DemandVec_{[1:k]}} \geq \frac{s_k}{\Demandi{k}+\ExpDemand_{k+1}} \left(1-\frac{n-k}{2(n+1-k)}\right)\triangleq \textsc{RHS}^{(2)}.
    \end{equation}
\end{enumerate}
The lower bound given by $\textsc{RHS}^{(2)}$ is a convex homographic function of  $\Demand_k + \ExpDemand_{k+1}$, as illustrated by the dashed red lines in all panels of \Cref{fig:expostproof} (in the regime where $\Demand_k + \ExpDemand_{k+1}\in[s_k/\beta_k,+\infty)$). To further lower bound this function by a linear function, note that for any variable $z$ the following inequality holds:
\begin{align*}
\left(\frac{n+2-k}{2(n+1-k)} \right)\frac{s_k}{z} - \beta_k\left(1-\frac{n+1-k}{2(n+2-k)}\frac{\beta_k}{s_k}z\right) = \frac{(n+1-k)\beta_k^2}{2(n+2-k)s_k z}\left(\frac{(n+2-k)s_k}{(n+1-k)\beta_k} - z\right)^2\geq 0~.
\end{align*}
The proof of the above inequality is purely algebraic and we omit it for brevity. Substituting $z=\Demand_k + \ExpDemandi{k+1}$ in this inequality, we have:
\begin{equation} 
\label{eq:bound1and2repeat2}
\textsc{RHS}^{(2)} \geq \beta_k\left(1-\frac{n+1-k}{2(n+2-k)}\frac{ \Demand_k + \ExpDemand_{k+1}}{\Supply_k}\beta_k\right),
\end{equation}
which proves inequality \eqref{eq:invariant2} in the insufficient supply case (again, see the blue lines in all panels of \Cref{fig:expostproof}).

\begin{figure}[t]
    \centering
    \includegraphics[trim={4.7cm 12.5cm 1.0cm 5.2cm},clip,width = .88\textwidth]{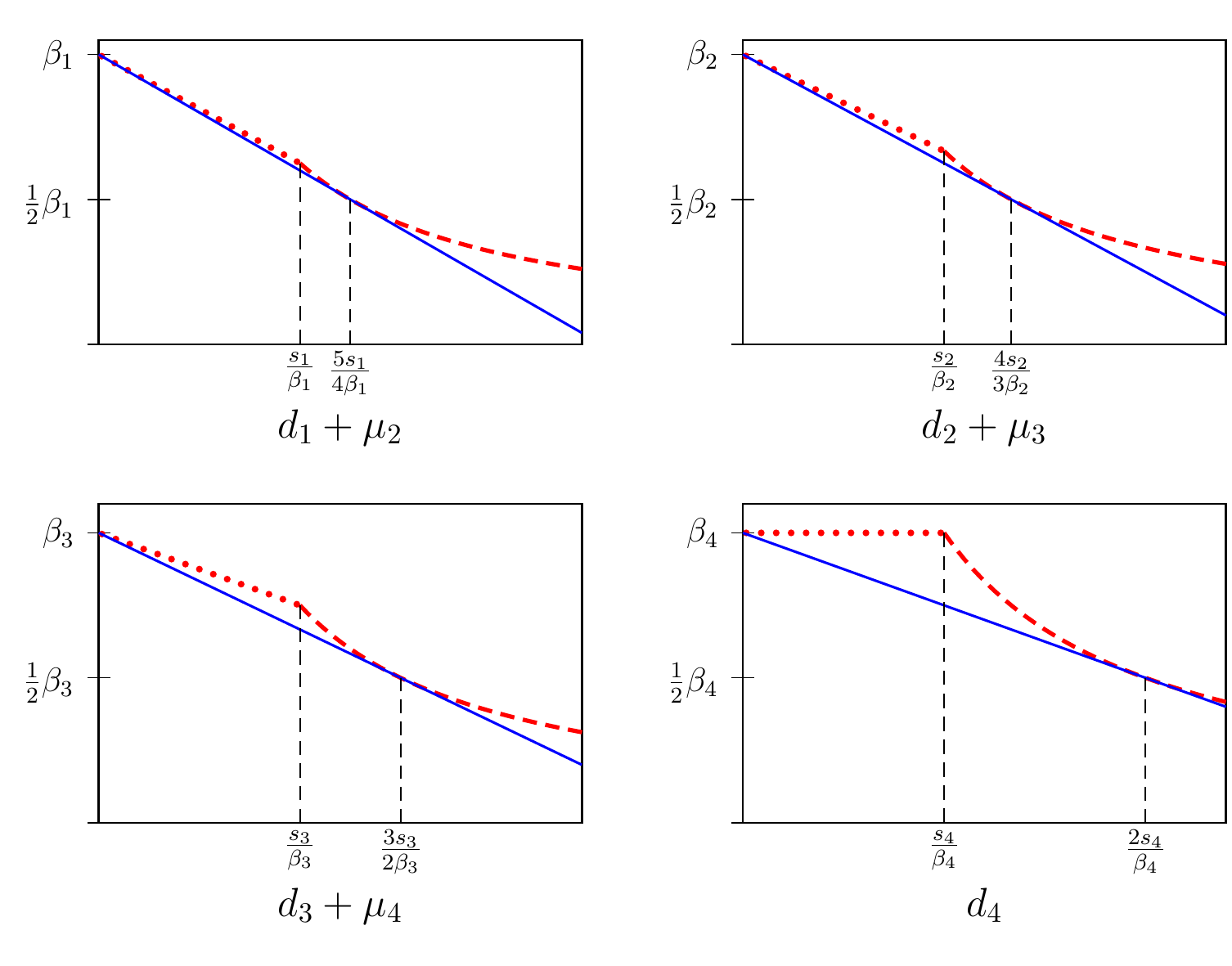}
    \caption{Lower bounds on the expected minimum FR given by \cref{eq:invariant2} (blue solid lines), \cref{eq:bound1} (red dotted lines), and \cref{eq:bound3} (red dashed lines)  when $\bm{n=4}$ for $\bm{k \in [4]}$.}
\label{fig:expostproof}
\end{figure}

Combining the above cases proves inequality \eqref{eq:invariant2} everywhere, which immediately implies the inductive hypothesis, i.e., inequality \eqref{eq:invariant},  for $i = k$, thus finishing the proof of the lemma.
\hfill\Halmos
\end{proof}

\subsection{Simple Non-adaptive Policies and Ex-post Fairness}
\label{subsec:AdaptivityGap}
As discussed in the previous sections,  our \PolicyNameAb\ policy is \emph{adaptive}, that is, the FR for agent $i$ (and its corresponding allocation decision) can depend not only on the observed demand $d_i$ but also on the exact sample path up to time $i$ as well as the remaining supply $s_i$. In contrast to an adaptive policy, a \emph{non-adaptive} policy commits to a sequence of feasible allocation maps $\{x_i(d_i)\}_{i\in[n]}$ upfront, where $x_i(d_i)\in[0,d_i]$ is the allocation decision for agent $i$ when agent $i$ has demand $d_i$.\footnote{For ease of presentation, we focus on deterministic non-adaptive policies. This is without loss of generality, as the ex-post fairness of any randomized non-adaptive policy must be weakly dominated by the ex-post fairness of one of the deterministic policies that it randomizes over.\label{foot:detpolicies}} If the non-adaptive policy's allocation decision $x_i(d_i)$ exceeds the remaining supply $s_i$, then agent $i$ instead receives the entire remaining supply.

For settings that we consider, adaptivity can indeed help with improving the expected minimum FR of a policy. As an example, compare running our \PolicyNameAb~policy versus the best non-adaptive policy on an instance with three agents. In this instance, the demands $\vec{d}=(d_1,d_2,d_3)$ follow one of the two possible sample paths $(\epsilon_1,1,1)$ or $(\epsilon_2,1,0)$ with equal probabilities $1/2$, where  $\epsilon_1,\epsilon_2\geq 0$ and $\epsilon_1\neq \epsilon_2$. After agent $1$'s demand is realized, the \PolicyNameAb~policy knows exactly which sample path is happening. By calculating the exact total demand of agents $2$ and $3$, it obtains the optimal expected minimum FR of $\frac{1}{2}\times 1+\frac{1}{2}\times \frac{1}{2}=3/4$ for small $\epsilon_1,\epsilon_2$. However, a non-adaptive policy  cannot distinguish between the two possible sample paths after agent $1$'s demand is realized. 
Therefore, without loss of generality, it targets a FR of $\tau$ for agent $2$ and obtains an expected minimum FR of  $\frac{1}{2}\tau+\frac{1}{2}\min\{\tau,1-\tau\}$ for small $\epsilon_1,\epsilon_2$, which attains its maximum equal to $\frac{1}{2}$ at any $\tau\in[\frac{1}{2},1]$.

In applications that allow for adaptivity, our \PolicyNameAb~policy obtains the optimal ex-post fairness guarantee while also having the desirable properties of transparency and interpretability. However, adaptivity is not admissible in some practical scenarios---e.g., when the social planner should commit to an allocation plan in advance for even more transparency or due to legal restrictions.\footnote{As discussed in the introduction, the initial strategy for allocating medical supplies at the beginning of COVID-19 pandemic had the form of a target-fill-rate policy, which is a canonical non-adaptive strategy as we will discuss soon.}  Motivated by such scenarios, we study two simple and natural canonical classes of non-adaptive policies:
those that fix the sequence of allocation decisions a priori, namely they specify one allocation vector $\vec{\Alloc}$,
and ``smarter'' policies which fix the sequence of fill rates $\vec{\TargetFillRate} \triangleq (\TargetFillRate_1, \TargetFillRate_2, \dots, \TargetFillRate_\NumAgents)$ a priori. 
In Appendix~\ref{apx:fixedallocation}, we show that the ex-post fairness guarantee 
for the former subclass is vanishing as $\NumAgents$ gets large.
Therefore, we focus on the latter subclass, which is formally defined as follows. 

\begin{definition}[\textbf{Target-fill-rate Policies}]
A \emph{target-fill-rate (\fixedthresh)} policy is any policy $\Policy$ which pre-determines a target fill rate $\TargetFillRate\in[0,1]$. Then, for every arriving agent $i$, the policy $\Policy$ must either allocate sufficient supply to meet the target or allocate all remaining supply, i.e., 
$$\forall i\in[n]:~~~\Alloci{i}(s_i,d_i) = \min\{\TargetFillRate \Demandi{i}, \Supplyi{i}\}.$$
\end{definition}

 In the following theorem, we provide a tight bound on the ex-post fairness guarantee (\Cref{def:polperf}) achievable by the {\em optimal} \fixedthresh\ policy---defined as the one that maximizes ex-post fairness for the given joint demand distribution. 
 We remark that setting one threshold is without loss of generality because the ex-post fairness guarantee of a policy which pre-determines a \emph{sequence} of target fill rates $\{\TargetFillRate_i\}_{i\in[n]}$ is upper bounded by that of a \fixedthresh\ policy with the same target fill rate $\TargetFillRate=\min_{i\in[n]}\{\TargetFillRate_i\}$ for all agents. We also highlight that in addition to achieving a lower ex-post fairness guarantee compared to our adaptive policy,  finding the best \fixedthresh\ policy requires full knowledge of the total demand distribution---in contrast to our \PolicyNameAb\ policy which only requires knowing the first conditional moments of the future total demand at each time.

\begin{theorem}[Ex-post Fairness Guarantee of Optimal \fixedthresh\ Policy]
\label{thm:NonAdaptive}
Given any number of agents $\NumAgents\in\mathbb{N} \setminus \{1\}$ and supply scarcity $\ExpDemand\in\NonNegReals$, the optimal \fixedthresh\ policy achieves an ex-post fairness guarantee (see \Cref{def:polperf}) of $\frac{\max\{1, \ExpDemand\}}{\ExpDemand + \sqrt{\ExpDemand^2+1}}$.
\end{theorem}

In \Cref{fig:AdaptivityGap}, we compare the guarantee 
of the optimal \fixedthresh\ policy against our \PolicyNameAb\ policy for different model primitives, $\mu$ and $n$.  First, we note that when $n$ is not too large, our \PolicyNameAb\ policy achieves a considerably higher guarantee. Next,
we highlight that the ex-post fairness guarantee for the optimal \fixedthresh\ policy does not depend on the number of agents $\NumAgents$.\footnote{We elaborate on the intuition behind this behavior when we present the hard instance for establishing the upper bound of Theorem \ref{thm:NonAdaptive}.} 
This is in contrast to the ex-post guarantee for the \PolicyNameAb\ policy $\LowerboundFunEx$, which worsens as the number of agents increases.
Furthermore, the guarantee in \Cref{thm:NonAdaptive} has a unique minimum of $\frac{1}{1+\sqrt{2}} \approx 0.41$ when $\ExpDemand = 1$. This once again suggests that the stochastic nature of demand is most harmful when expected demand exactly equals supply. \revcolor{Before presenting the proof of  \Cref{thm:NonAdaptive} (which we defer to \Cref{sec:non-adaptive-proof}), we first discuss a special case where \fixedthresh\ policies can provide a substantially stronger fairness guarantee.}

\begin{figure}[t]
 \centering
 \begin{subfigure}[b]{0.48\textwidth}
     \includegraphics[trim={5cm 17.3cm 5cm 4.0cm},clip,width=0.98\textwidth]{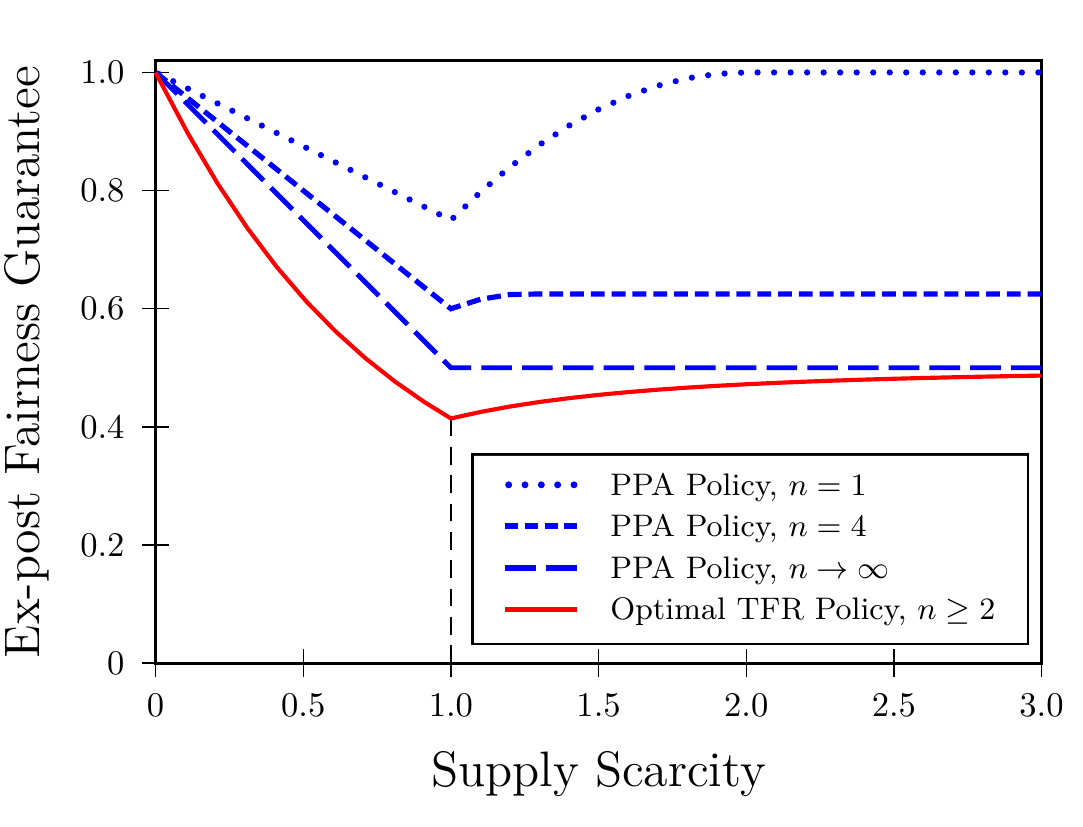}
    \caption{ }
    \label{fig:AdaptivityGap}
 \end{subfigure}
 \begin{subfigure}[b]{0.50\textwidth}
     \includegraphics[trim={7.8cm 17.3cm 1.7cm 4.5cm},clip,width=0.98\textwidth]{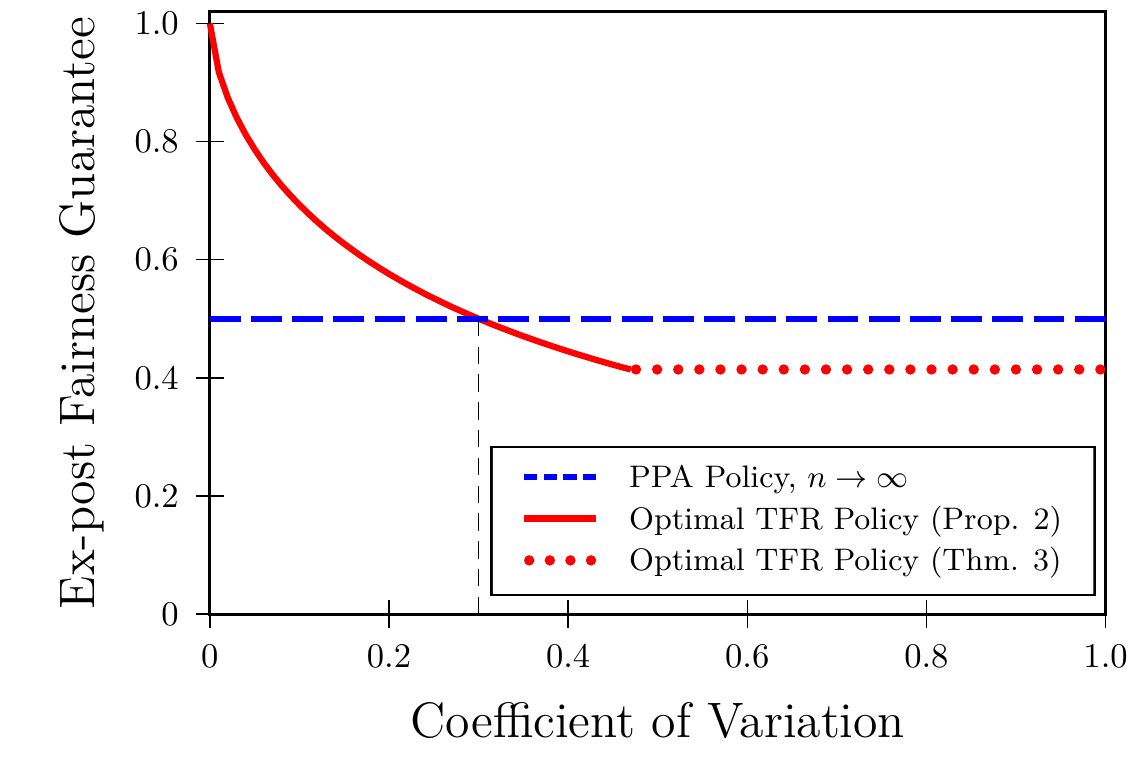}
     \caption{ }
    \label{fig:coefvar}
 \end{subfigure}
 \caption{(a) Tight ex-post fairness guarantees of our \PolicyNameAb\ policy and the optimal \fixedthresh\ policy.  \revcolor{(b) Ex-post fairness guarantee of the optimal \fixedthresh\ policy as a function of the coefficient of variation when $\ExpDemand =1$.}}
\end{figure}

\revcolor{\subsubsection{Optimal TFR Policy for Low-Variance Distributions} 
Having presented our general result for the optimal \fixedthresh\ policy,
we note that we make no assumption about the demand distribution in the statement of \Cref{thm:NonAdaptive}, aside from parameterizing the result by the expected total demand. Because non-adaptive policies such as the optimal \fixedthresh\ policy cannot update their decisions based on new information, we would expect their performance to suffer when the variance of total demand is high (fixing its expectation). For instance, the worst-case distributions that establish our parameterized tight bound for \fixedthresh\ policies (see \Cref{thm:NonAdaptive}) all have a coefficient of variation (CV) greater than $0.57$.
\revcolor{Inspired by this observation, in the following proposition we provide an alternative lower bound on the ex-post fairness guarantee of the optimal \fixedthresh\ policy that is parameterized by an upper bound on the CV.}
\begin{proposition}[\revcolor{Optimal TFR Policy for Bounded Coefficient of Variation}]
\label{prop:coefvar}
Suppose the coefficient of variation for total demand is at most $\coefvar$. (Equivalently, suppose the standard deviation of total demand is at most $\coefvar \ExpDemand$.) Then for any supply scarcity $\ExpDemand\in \NonNegReals$, the optimal \fixedthresh\ policy achieves an ex-post fairness guarantee (see \Cref{def:polperf}) of at least 
\begin{align}
    \max_{\xvar \in [0,\min\{1, \ExpDemand\}]} \frac{\max\{1, \ExpDemand\}}{\ExpDemand} \xvar \left(\frac{\left(\frac{1-\xvar}{\xvar} \right)^2}{\coefvar^2 + \left(\frac{1-\xvar}{\xvar} \right)^2 }\right). \label{eq:coefvar}
\end{align}
\end{proposition}
We provide the proof of \Cref{prop:coefvar} in Appendix \ref{apx:prop:coefvar}. We do not have a closed-form expression for this improved lower bound, but the single-variable concave maximization problem in \eqref{eq:coefvar} is easy to solve numerically. We note that the lower bound on the guarantee of the optimal \fixedthresh\ policy approaches $1$ as the CV approaches $0$, regardless of the supply scarcity~$\ExpDemand$. 

\revcolor{Even though the lower bound in \Cref{prop:coefvar} is not necessarily tight, as long as the CV is sufficiently small, this lower bound improves upon the bound in \Cref{thm:NonAdaptive} (which holds for an arbitrary coefficient of variation).}\footnote{\revcolor{The value for $\coefvar$ at which the two bounds cross depends on the supply scarcity $\ExpDemand$. Numerically, we find that if $\coefvar \leq 0.3$, the lower bound in \Cref{prop:coefvar} provides a better guarantee for any $\ExpDemand$.}} Furthermore, this bound establishes that in settings where the CV is below a threshold, the optimal \fixedthresh\ policy can provide an improved guarantee relative to the \PolicyNameAb\ policy.

In \Cref{fig:coefvar}, we illustrate the lower bound on the optimal \fixedthresh\ policy  when the supply scarcity $\ExpDemand = 1$, 
as a function of the CV (solid red line when the bound is given by \Cref{prop:coefvar}, dotted red line when the bound is given by \Cref{thm:NonAdaptive}). We compare this lower bound to the infimum of all lower bounds provided by the \PolicyNameAb\ policy. (This infimum, shown by the dashed blue line, occurs when $\NumAgents \rightarrow \infty$, and it is equal to $0.5$ as given by \Cref{thm:expost}.) In the case shown (where $\ExpDemand =1$), the lower bound on the guarantee of the optimal \fixedthresh\ policy given by \Cref{prop:coefvar} dominates the bound given by \Cref{thm:NonAdaptive}  when the CV is below $0.47$. Importantly, the optimal \fixedthresh\ policy can provide improvement compared to the \PolicyNameAb\ policy if the CV is less than $0.3$.\footnote{\revcolor{For the sake of brevity, we only show this comparison when the supply scarcity $\ExpDemand =1$, but similar results hold for other supply scarcity levels.}}} 

\revcolor{Thus, in cases where total demand is known to be well-concentrated, such as in instances with a large number of i.i.d. demands, \fixedthresh\ policies can perform quite well. In contrast, the \PolicyNameAb\ policy is particularly valuable when demand is correlated and highly variable. We emphasize that such demand sequences may arise in our motivating applications, such as when responding to a pandemic or natural disaster. In our case study presented in \Cref{sec:numerics}, we present a simple example for the COVID-19 pandemic based on commonly-used epidemiology models which exhibits correlation across demands and a high coefficient of variation for total demand.}

\subsubsection{\texorpdfstring{Proof of \Cref{thm:NonAdaptive}}{}}
\label{sec:non-adaptive-proof}
We now provide a proof of \Cref{thm:NonAdaptive}, with some details deferred until Appendix \ref{apx:adaptive-v2}.
We begin by placing a lower bound on the performance of the optimal \fixedthresh\ policy, and we then demonstrate the existence of a matching upper bound.

\begin{proof}{{\bf Proof of lower bound:}}
For any target fill rate $\TargetFillRate$, a \fixedthresh\ policy will achieve that fill rate if $\TargetFillRate\left(\sum_{i \in [\NumAgents]} \Demandi{i}\right) \leq 1$. Let us define $\invdemanddist$ as the cumulative distribution function (CDF) of the random variable $\invdemand \triangleq \frac{1}{\sum_{i \in [\NumAgents]} \Demandi{i}}$, where $\expect[\invdemand \sim \invdemanddist]{\frac{1}{\invdemand}}=\ExpDemand$. For ease of notation, we use $\InvDistStateSpace$ to denote the domain of all such CDFs. Given a CDF $\invdemanddist$, a \fixedthresh\ policy with target fill rate $\TargetFillRate$ achieves an expected minimum fill rate of at least $\TargetFillRate(1-\invdemanddist(\TargetFillRate))$, which implies that the optimal \fixedthresh\ policy
attains an expected minimum fill rate of at least $\max_{\TargetFillRate\in[0,1]}\TargetFillRate \left(1-\invdemanddist(\TargetFillRate)\right)$.
{In the following lemma, we establish a lower bound on $\max_{\TargetFillRate\in[0,1]}\TargetFillRate \left(1-\invdemanddist(\TargetFillRate)\right)$, which enables us to lower-bound the ex-post fairness guarantee that the optimal \fixedthresh\ policy achieves.}

\begin{lemma}[Tight Lower Bound for Optimal \fixedthresh\ Policy]
\label{lem:nonadaptlowbound}
Given a fixed number of agents $\NumAgents \in \mathbb{N}$ and 
supply scarcity $\ExpDemand \in \NonNegReals$, the following holds:
\begin{equation}
\label{eq:best-non-adaptive-v3}
    \underset{\invdemanddist \in \InvDistStateSpace}{\inf}\left\{\underset{\TargetFillRate\in[0,1]}{\max} \TargetFillRate \left(1-\invdemanddist(\TargetFillRate)\right)\right\} = \frac{1}{\ExpDemand +\sqrt{\ExpDemand^2 +1}}.
\end{equation}
This infimum is attained by the following CDF
\begin{equation}
\label{eq:worstthresholddist}
\hat{\invdemanddist}(\invdemand) =
\left\{
	\begin{array}{ll}
		0 &\quad\quad\mbox{if}~\invdemand\in[0,\hat{q}) \\
		1-{\hat{q}}/{\invdemand} &\quad\quad\mbox{if}~\invdemand\in[\hat{q},1)  \\
		1-\hat{q} &\quad\quad\mbox{if}~\invdemand\in[1,+\infty)\\
		1 &\quad\quad\mbox{if}~\invdemand=+\infty
	\end{array}
\right.
\end{equation}
where $\hat{q} = \frac{1}{\ExpDemand +\sqrt{\ExpDemand^2+1}}$.
\end{lemma}
Before presenting the proof of the above lemma in \Cref{sec:key-lemma-non-adaptive}, 
which is the key step of the proof of \Cref{thm:NonAdaptive}, 
we establish a matching upper bound and complete the proof of the theorem.

\end{proof}

\begin{proof}{{\bf Proof of upper bound:}}
We show a matching upper bound by considering a two-agent instance. In this instance, only the first agent has stochastic demand. In particular,
$\Demand_1=\frac{1-\epsilon}{\invdemand}$ for $\invdemand\sim \hat{\invdemanddist}$ (defined in \cref{eq:worstthresholddist}) and $\Demand_2 = \epsilon\ExpDemand$ deterministically. Note that $\expect{\Demand_1 +\Demand_2} = \ExpDemand$. For any target fill rate $\TargetFillRate' = (1-\epsilon)\TargetFillRate$ where $\TargetFillRate \in [0,1]$, supply will be exhausted before the arrival of agent $2$ with probability $\hat{\invdemanddist}(\TargetFillRate)$, in which case the minimum FR will be $0$.  Therefore, the expected minimum fill rate of the optimal \fixedthresh\ policy in this instance is at most
\begin{align*}
    \max_{\TargetFillRate\in[0,1]} (1-\epsilon)^{-1}\TargetFillRate \left(1-\hat{\invdemanddist}(\TargetFillRate)\right) = 
    (1-\epsilon)^{-1}\frac{1}{\ExpDemand +\sqrt{\ExpDemand^2 +1}},
\end{align*}
where the equality follows from \Cref{lem:nonadaptlowbound}.

By allowing $\epsilon\to 0$, we conclude that there exists an instance where the expected minimum fill rate of the optimal \fixedthresh\ policy is $\frac{1}{\ExpDemand +\sqrt{\ExpDemand^2 +1}}$, which matches the lower bound from above. 
We remark that the construction of the above two-agent example clarifies why our upper bound does not depend on the number of agents: we can modify the example to an $\NumAgents$-agent one where the total demand of the first $\NumAgents-1$ agents have correlated demand equal to $\frac{1-\epsilon}{\invdemand}$ for $\invdemand\sim \hat{\invdemanddist}$ and the last agent has a deterministic demand of $\epsilon \ExpDemand$.
\end{proof}

With the above (matching) bounds, we complete
the proof of \Cref{thm:NonAdaptive} by 
 scaling this tight bound by our benchmark for deterministic demand, namely $\DetGuar = \min\{1, 1/\ExpDemand\}$, to arrive at the guarantee stated in \Cref{thm:NonAdaptive}.
\hfill\Halmos

\subsubsection{Proof of \texorpdfstring{\Cref{lem:nonadaptlowbound}}{} and Connections to Monopoly Pricing}
\label{sec:key-lemma-non-adaptive}
Having laid out the proof steps of \Cref{thm:NonAdaptive},
we now provide a constructive proof of the key lemma, i.e., Lemma \ref{lem:nonadaptlowbound}. We do so by identifying properties of the worst-case distribution against the optimal \fixedthresh\ policy, which enables us to exactly characterize that distribution.

 To aid in this proof, we introduce a one-to-one mapping of each target fill rate $\tau$ into the quantile space, such that quantile $q$ corresponds to \fixedthresh\ $\TargetFillRate$ if and only if there is sufficient supply to meet a fraction $\tfrvar$ of demand with probability exactly $q$.   
 We start by describing notation for this transformation, along with some basic properties, in the following definition.
 For simplicity of exposition, we assume all the distributions playing the role of $\invdemanddist$ are non-atomic.\footnote{This assumption is without loss of generality, as one can always add an infinitesimal continuous perturbation to each distribution, which does not change any of the arguments in this proof.}

\begin{definition}[\textbf{TFR in Quantile Space}] 
\label{def:rev}
Given a (non-atomic) CDF $\invdemanddist:\NonNegReals\rightarrow[0,1]$ and inverse total demand $\invdemand \sim \invdemanddist$, we define the following mappings.
\begin{itemize}
\item {\underline{\em TFR-to-quantile map $Q^\invdemanddist$}:}
    ~The quantile corresponding to TFR $\tfrvar \in [0,1]$ is $Q^\invdemanddist(\tfrvar)\triangleq 1-\invdemanddist(\tfrvar)$. In words,
    the probability of being able to meet a fraction $\tfrvar$ of total demand is $Q^\invdemanddist(\tfrvar)$. This map is monotone non-increasing.
    \item {\underline{\em Quantile-to-TFR map $\Tfrvar^\invdemanddist$}:}
    ~The TFR corresponding to quantile $q\in[0,1]$ is $\Tfrvar^\invdemanddist(q)\triangleq \invdemanddist^{-1}(1-q)$.
    In words, $\Tfrvar^\invdemanddist(q)$ is the TFR for which the probability of being able to meet a fraction $\Tfrvar^\invdemanddist(q)$ of total demand is $q$.  This map is monotone non-increasing and is the inverse of the TFR-to-quantile map, i.e., $\Tfrvar^\invdemanddist=\left({Q^\invdemanddist}\right)^{-1}$. 
    
    \item {\underline{\em The expected achievable fill rate (EAFR) curve $R^\invdemanddist$}:} 
   ~For $q\in[0,1]$,
    $R^\invdemanddist(q)\triangleq q\cdot \invdemanddist^{-1}(1-q)$ is the EAFR when the probability of meeting demand (given the TFR) is exactly equal to $q\in[0,1]$, i.e., the EAFR obtained by targeting a fill rate $\Tfrvar^\invdemanddist(q)$. 
\end{itemize}
\end{definition}

\begin{figure}[t]
 \centering
 \begin{subfigure}[b]{0.48\textwidth}
     \includegraphics[
     trim = {2.4cm 2.2cm 13.15cm 2cm}
     ,clip,width=.95\textwidth]{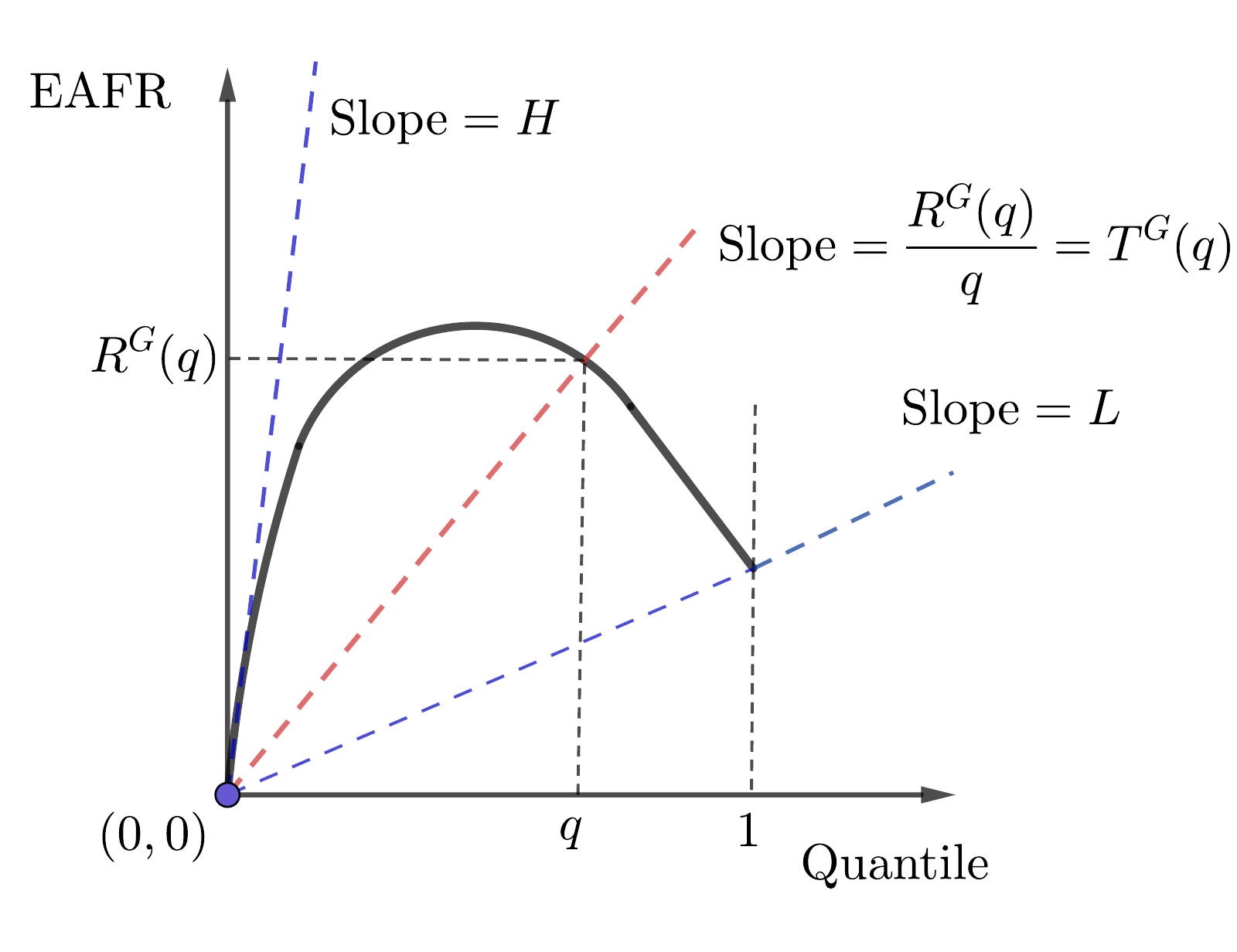}
    \caption{ }
    \label{fig:rev}
 \end{subfigure}
 \begin{subfigure}[b]{0.48\textwidth}
     \includegraphics[
     trim={5.4cm 1.5cm 9cm 2cm}
     ,clip,width=.95\textwidth]{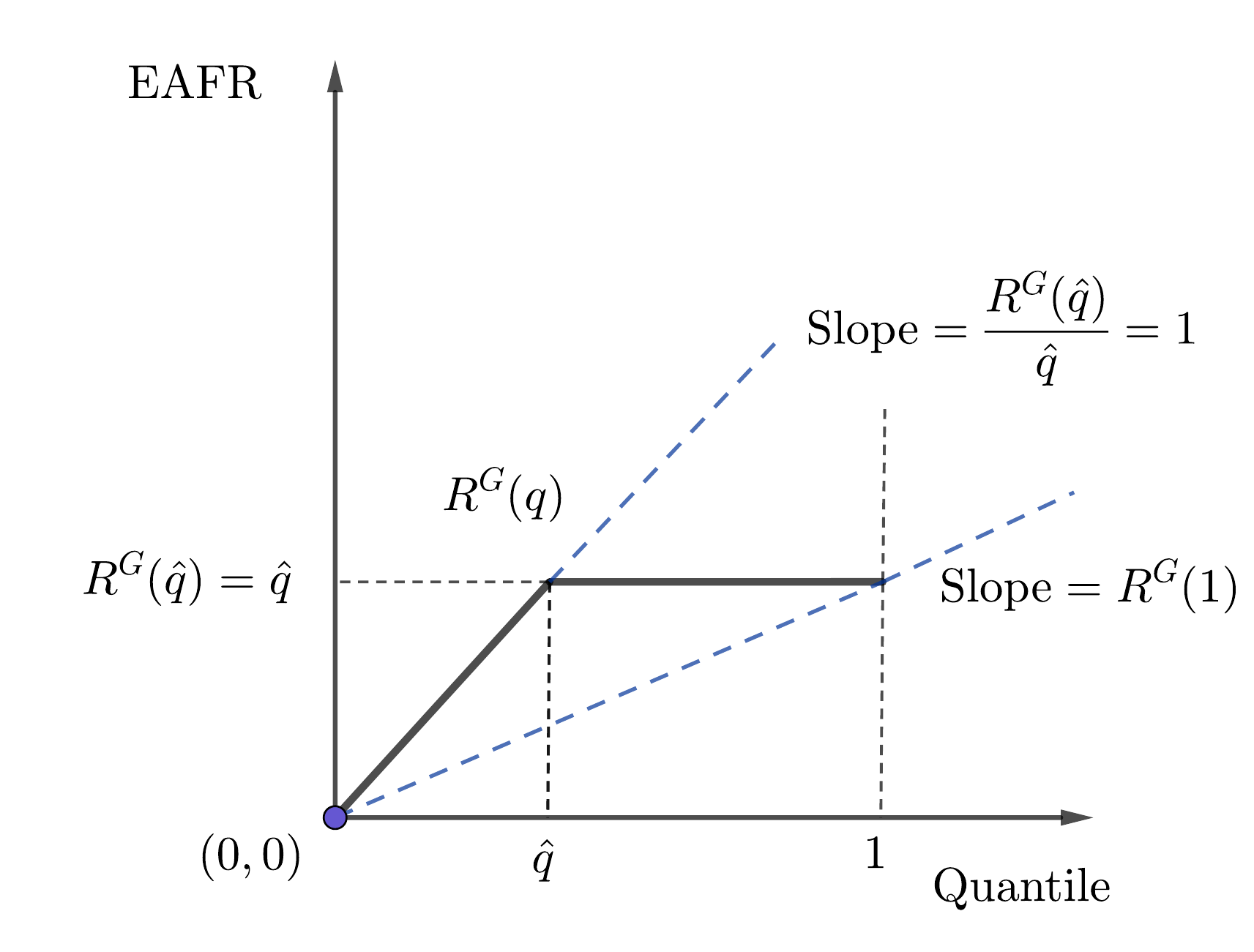}
     \caption{ }
    \label{fig:equal-rev}
 \end{subfigure}
 \caption{(a) The expected achievable fill rate (EAFR) curve, fill rates, and support in quantile space. (b) The EAFR curve of the worst-case distribution, as defined in \cref{eq:worstthresholddist}, in quantile space.}
\end{figure}
\begin{remark}
\label{rem:monopolypricing}
In light of the above transformation, we remark that there is a reduction from our setup to a single-parameter Bayesian mechanism design problem in which a monopolistic seller has an item to sell to a single buyer with private valuation $\invdemand \sim \invdemanddist$, where $\invdemanddist$ is the common prior valuation distribution. See \cite{alaei2019optimal} for an example of such a setting; also refer to \cite{hartline2013mechanism} for more details on monopoly pricing. In this reduction, target fill rates correspond to prices and the EAFR corresponds to the expected revenue in monopoly pricing (accordingly, the EAFR curve also corresponds to the revenue curve). The problem in this parallel monopoly pricing setting is identifying the worst-case distribution $\invdemanddist$  satisfying $\expect[\invdemand \sim \invdemanddist]{1/v}=\ExpDemand$, so that we minimize the maximum revenue obtained from selling the item at prices {constrained to be in} the interval $[0,1]$.  
\end{remark}

According to \Cref{def:rev}, $\tfrvar\left(1-\invdemanddist(\tfrvar)\right)$ is equivalent to $R^\invdemanddist\left(Q^\invdemanddist(\tfrvar)\right)$ for any TFR $\tfrvar \in [0,1]$. Based on this insight, to prove \Cref{lem:nonadaptlowbound}, it is sufficient to show that
\begin{equation}
    \label{eq:best-non-adaptive-lemma}
 \displaystyle\underset{\invdemanddist \in \InvDistStateSpace}{\inf}\left\{\underset{q\in[0,1]:\Tfrvar^\invdemanddist(q)\in[0,1]}{\max} R^\invdemanddist(q)\right\} = \frac{1}{\ExpDemand + \sqrt{\ExpDemand^2 + 1}}.
\end{equation}

Consider all cumulative distribution functions $\invdemanddist \in \InvDistStateSpace$. We first identify two additional constraints on $\invdemanddist$ that do not change the infimum in \cref{eq:best-non-adaptive-lemma}. These constraints enable us to
find the worst-case distribution that achieves the infimum value which establishes the desired result.

Before proceeding, we develop intuition using an illustrative example of the EAFR curve shown in \Cref{fig:rev}.
In general, 
if one draws $R^\invdemanddist(q)$ as a function of $q\in[0,1]$ (i.e., in the quantile space), then the slope of the line connecting the point $(0,0)$ to $\left(q,R^\invdemanddist(q)\right)$ is equal to $\Tfrvar^\invdemanddist(q)=R^\invdemanddist(q)/q$. This slope is monotone non-increasing in $q$ for any CDF $G$ according to Definition \ref{def:rev}.
Hence, given the EAFR curve $R^\invdemanddist(q)$, the support of the feasible fill rates 
is equal to  $[L,\min\{1,H\}]$, where $L\triangleq R^\invdemanddist(1)$ and $H\triangleq \underset{q\to 0}{\liminf}~{R^\invdemanddist(q)}/{q}$. 
The two constraints that we will add below, as stated in Claims \ref{clm:nonadaptive1} and \ref{clm:nonadaptive2},  imply that the outer optimization problem in \cref{eq:best-non-adaptive-lemma} will remain unchanged if we require the EAFR curve to be (i)  flat over quantiles corresponding to target fill rates in $[L,1]$, i.e., quantiles in the interval $[Q^\invdemanddist(1), 1]$, and (ii) a straight line with slope $1$ for quantiles in the interval $[0, Q^\invdemanddist(1)]$. 
With these two additional constraints, in Claim \ref{clm:nonadaptive3} we find the worst-case CDF, which has an EAFR curve as shown in \Cref{fig:equal-rev}.

\begin{claim}[Equal EAFR]
\label{clm:nonadaptive1}
Adding the constraint $R^\invdemanddist(q)=R^\invdemanddist(q'), \forall q,q'\in [Q^\invdemanddist(1),1]$ to the outer optimization in \cref{eq:best-non-adaptive-lemma} does not change its infimum value.
\end{claim}
We prove Claim \ref{clm:nonadaptive1} by contradiction:
 we show that for any CDF $G \in \InvDistStateSpace$, if the above condition does not hold, we can slightly modify $G$ to design a new distribution $\tilde{G} \in \InvDistStateSpace$ which has an EAFR curve with a lower maximum value. The details are presented in Appendix~\ref{apx:nonadaptclaim1}.
The above claim readily implies that we can focus on distributions for which the EAFR curve is flat in the interval $[Q^\invdemanddist(1),1]$. 

Next, 
we claim that we can restrict our attention to 
distributions where there is no probability mass for $\invdemand \in(1,+\infty)$. Said differently, the support of inverse demand is  $(0,1] \cup  \{+\infty\}$. 
\begin{claim}[Restricted Support for Inverse Demand]
\label{clm:nonadaptive2}
Adding the constraint $\invdemanddist(\invdemand) = \invdemanddist(1)$ for all $\invdemand \in [1, +\infty)$ and $\lim_{\invdemand \to +\infty}\invdemanddist(\invdemand)=1$ to the outer optimization in \cref{eq:best-non-adaptive-lemma} does not change its infimum value.
\end{claim}
We also prove Claim \ref{clm:nonadaptive2} by contradiction: we show that for any CDF $G \in \InvDistStateSpace$, if there is probability mass on $v \in (1, +\infty)$, we can construct a CDF $\tilde{G} \in \InvDistStateSpace$ which has an EAFR curve with a lower maximum value by shifting that mass to $+\infty$.
The details are presented in Appendix~\ref{apx:nonadaptclaim2}.
Again, note that this claim implies that we can focus on distributions for which the EAFR curve starts with a straight line up to quantile $Q^\invdemanddist(1)$.

Given the two claims above, the distribution that attains the infimum in \cref{eq:best-non-adaptive-lemma} must satisfy the two constraints introduced. \Cref{fig:equal-rev} summarizes the effect of these two restrictions on $R^\invdemanddist(q)$. 

\begin{claim}[Worst-case CDF]
\label{clm:nonadaptive3}
For any $\ExpDemand \in \NonNegReals$, the distribution $\hat{\invdemanddist}$ given in \cref{eq:worstthresholddist} is the unique distribution in $\InvDistStateSpace$
 satisfying the constraints introduced in Claims \ref{clm:nonadaptive1} and \ref{clm:nonadaptive2}. Therefore, this distribution attains the infimum in \cref{eq:best-non-adaptive-lemma}.
\end{claim}
We prove Claim \ref{clm:nonadaptive3} in Appendix~\ref{apx:nonadaptclaim3}. Since the EAFR curve $R^{\hat{\invdemanddist}}(q)$ has a maximum value of $\hat{q} = \frac{1}{\ExpDemand +\sqrt{\ExpDemand^2+1}}$, we have shown that the optimal \fixedthresh\ policy always achieves an EAFR of at least $\hat{q}$.
This completes the proof of \Cref{lem:nonadaptlowbound}.\hfill\halmos

\subsection{Ex-ante Fairness}
\label{subsec:ex-ante}
In this section,
 we study our second notion of fairness, namely, ex-ante fairness.
 As we did for ex-post fairness, we first establish an upper bound on the ex-ante fairness guarantee  achievable by any policy.  More importantly, we then show that our \PolicyNameAb\ policy achieves this worst-case ex-ante fairness bound.
The following theorem  establishes our matching upper and lower bounds on the ex-ante fairness guarantee.

\begin{theorem}[Ex-ante Fairness Guarantee of \PolicyNameAb\ Achieves Upper Bound]
\label{thm:exante}
Given a fixed number of agents $\NumAgents\in\mathbb{N}$ and supply scarcity $\ExpDemand\in \NonNegReals$, no sequential allocation policy obtains an ex-ante fairness guarantee (see \Cref{def:polperf}) greater than $\LowerboundFunExAnte$, defined as
\begin{equation}
\label{eq:kappa-exante}
   \LowerboundFunExAnte 
   =\begin{cases}
1-\frac{\ExpDemand}{4}, &\ExpDemand \in [0, 1) \\
\ExpDemand\left(1-\frac{\ExpDemand}{4}\right), & \ExpDemand \in [1, 2) \\
1, &\ExpDemand \in [2, +\infty)
\end{cases}.
\end{equation}
Further, the \PolicyNameAb ~policy achieves an ex-ante fairness guarantee of at least $\LowerboundFunExAnte$.
\end{theorem}

Like its counterpart for ex-post fairness, $\LowerboundFunExAnte$ depends on the supply scarcity, $\ExpDemand$, and is at its lowest when expected demand equals supply, which highlights the loss due to stochasticity when trying to achieve efficiency and equity ex ante. 
However, unlike the bound for ex-post fairness, the ex-ante fairness bound is independent of the number of agents. In fact, this bound is identical to the ex-post fairness bound in the single-agent case, i.e. $\LowerboundFunExAnte = \LowerboundFunExPostOne$ (which is shown by the dotted line in \Cref{fig:AdaptivityGap}).

For intuition about this relationship, note that one feasible policy is to allocate supply to each agent proportional to their expected demand. Since the ex-ante problem only depends on marginal FRs, this reduces ex-ante fairness to the minimum ex-ante fairness across $\NumAgents$ single-agent instances (where in each instance, the supply scarcity is $\ExpDemand$). In a single-agent instance, ex-ante fairness is equal to ex-post fairness, which implies that any lower bound on single-agent ex-post fairness also serves as a lower bound on ex-ante fairness with $\NumAgents$ agents.
Furthermore, since demands can be perfectly correlated, any single-agent instance can be expressed as an instance with $\NumAgents$ agents for any $\NumAgents \in \mathbb{N}$. This implies that any upper bound on single-agent ex-post fairness also serves as a upper bound on ex-ante fairness with $\NumAgents$ agents.

To prove the upper bound in \Cref{thm:exante}, we build on the hard instances from the proof of \Cref{thm:hardness-ex-post}. To prove the lower bound, we use ideas similar to the proof of Theorem \ref{thm:expost}. We
 show that when following the \PolicyNameAb\ policy, the expected FR for each agent is a decreasing and convex function of the ratio of expected remaining demand to remaining supply upon their arrival. We inductively place an upper bound on the ex-ante expected value of that ratio for each agent, which enables us to provide a lower bound on ex-ante fairness.
See Appendix \ref{apx:ex-ante} for a detailed proof.

\section{Numerical Results}
\label{sec:numerics}
\revcolor{
We complement our theoretical developments with an illustrative case study 
motivated by the allocative challenges of 
rationing medical supplies in the midst of 
a pandemic. 
First, we describe a simple compartmental model  that governs the need for medical supplies experienced by different locations, which is based on models used to forecast the COVID-19 pandemic. 
{We illustrate that in such a setting, the demand is sequential, highly variable, and has complex correlation structure due to network effects.}
Then, we study this dynamic rationing problem within our framework and illustrate the effectiveness of our \PolicyNameAb\ policy by comparing it to (i) its  theoretical guarantee (presented in \Cref{thm:expost}), (ii) the optimal \fixedthresh\ policy (defined in ~\Cref{subsec:AdaptivityGap}), and (iii) a DP approach (similar to the one described in Appendix \ref{apx:DP}). We conclude this section by 
demonstrating  the efficiency of our policy as well as its robustness to mis-specification of model parameters.

\subsection{Pandemic Demand Model}
We model the spread of a pandemic across inter-connected locations using a standard SEIR (susceptible, exposed, infectious, recovered) model, with different compartments for different locations. Such compartmental models are commonly-used and frequently influence practice. (See, e.g., \citealt{morozova2021one}.) Each location represents an agent in our allocation framework, and we use the peak number of infected individuals in a location as a proxy for that agent's need.

\subsubsection{SEIR Model Dynamics}
In the following, we briefly overview the SEIR model that we borrow from the epidemiology literature \citep{anderson1992infectious, diekmann2000mathematical}.
{In this model, there are $L$ locations where location $i$ has population $p_i$.}
Individuals interact with each other according to a time-varying rate $\transmissparam_t$ (which we will specify later). Individuals in location $i$ predominantly interact with members of their own location, but a small fraction $\adjacencyparam_i$ of their interactions occur with a member of an adjacent {location} $j$, where $j$ is equally likely to be any of location $i$'s neighbors (denoted by the set $N_i$). 
When a susceptible individual interacts with an infectious individual, the susceptible individual becomes exposed. Exposed individuals become infectious after a random time drawn from an exponential
distribution with mean $1/\exptoinfparam$, and infectious individuals become recovered after a random time drawn from an exponential
distribution with mean $1/\inftorecparam$.  
For large networks, this system approaches the deterministic
dynamics presented in \eqref{eq:rec}, with $S_i(t)$, $E_i(t)$, $I_i(t)$, $R_i(t)$ representing the fraction of the population that is susceptible, exposed, infected, and recovered (respectively) in location $i$ at time $t$. 
We will use these dynamics to simulate pandemic trajectories. 
\begin{align}
    \frac{\partial S_i(t)}{\partial t} \ &= \ -\transmissparam_t S_i(t) \left((1-\adjacencyparam_i)I_i(t) + \frac{\adjacencyparam_i}{|\neighborset_i|}\sum_{j \in \neighborset_i}I_j(t) \right), \qquad \qquad \quad  \frac{\partial I_i(t)}{\partial t} \ = \ \exptoinfparam E_i(t) - \inftorecparam I_i(t) \nonumber \\
    \frac{\partial E_i(t)}{\partial t} \ &= \  \transmissparam_t S_i(t) \left((1-\adjacencyparam_i)I_i(t) + \frac{\adjacencyparam_i}{|\neighborset_i|}\sum_{j \in \neighborset_i}I_j(t) \right) - \exptoinfparam E_i(t), \qquad
    \frac{\partial R_i(t)}{\partial t} \ = \ \inftorecparam I_i(t) \label{eq:rec}
\end{align}

\subsubsection{SEIR Model Primitives}
We calibrate the parameters of our numerical studies in accordance with epidemiological estimates specific to the COVID-19 pandemic \citep{aleta2020modelling, li2020substantial,  park2020reconciling, walsh2020duration}. In particular, the state transition rates $\exptoinfparam$ and $\inftorecparam$ are set to be deterministic, while the initial interaction rate $\transmissparam_0$ is drawn from a (truncated) Normal distribution. {This distribution} is parameterized such that the range for the basic reproductive number in our model ($R_0 = \transmissparam_0/\inftorecparam$) reflects the wide range of estimates that appeared in the early stages of the pandemic. The interaction rate $\transmissparam_t$ then varies over time, according to the dynamics 
\begin{equation}
\transmissparam_t = \transmissparam_0 \  \text{exp}\left(\sum_{\tau =1}^t\randomwalk_\tau\right),
\end{equation} 
where each $\randomwalk_\tau$ is an independently-drawn Normal random variable with mean $\hypermean$ and standard deviation $\hypersd$. The term $\sum_{\tau =1}^t\randomwalk_\tau$ represents a random walk. This modeling choice is based on approaches taken in \citet{bhatt2020semi, morozova2021one}, and  \citet{mishra2021comparing}, as well as in the first version of the widely-used Los Alamos National Laboratory forecasts.\footnote{\revcolor{See \url{https://covid-19.bsvgateway.org/}.}} Allowing for stochastic changes in the interaction rate is a natural way of capturing the unknown future impact of factors such as seasonality, disease mutation, government interventions, and changes in individual behavior. The mean and standard deviation of the random walk that governs the time-varying interaction rate are drawn from uniform distributions which allow for a wide range of plausible trajectories. 
Finally, we fix the population size of all locations to be $1000$. We summarize the instance parameters and their associated distributions in \Cref{table:variables}.

\begin{table}[t]
\footnotesize
    \caption{\revcolor{SEIR model parameters and distributions.}}
    \centering
    \revcolor{\begin{tabular}{cll}
    & &\\ \hline
    Notation & \multicolumn{1}{c}{Description} & \multicolumn{1}{c}{Distribution}\\
    \hline 
         $\population_i$ & Population of Location $i$ & $1,000$ \ (Fixed, equal for all $i$) \\
         $\adjacencyparam_i$ & Prop. of interaction from adjacent locations & 0.015 (Fixed, equal for all $i$) \\
         $\exptoinfparam$ & Transition rate from $E$ to $I$ (per day) & 0.25 \ (Fixed) \\
         $\inftorecparam$ & Transition rate from $I$ to $R$  (per day) & 0.10 \ (Fixed) \\
            $\transmissparam_0$ & Initial interaction rate  (per day)& Normal$(0.4, 0.15)$, truncated at $0$ and $1$ \\
         $\hypermean$ & Interaction rate random walk mean& Unif$(-0.008, 0.002)$ \\
         $\hypersd$& Interaction rate random walk std. dev.& Unif$(0.0, 0.1)$ \\
         \hline
    \end{tabular}}
    \label{table:variables}
\end{table}

\subsubsection{Instance and Simulated Demands}
Our simulation study focuses on a simple four-location setting with equal populations where adjacency is given by a line graph.\footnote{\revcolor{
Focusing on this simple setting is partly necessitated by the challenges of computing the DP solutions for larger problems.}} (Location 1 is adjacent only to Location 2, Location 2 is adjacent to Locations 1 and 3, etc.) 
Initially, $0.01\%$ of the population of Location 1 is assumed to be exposed, and all other individuals are susceptible. 
In each simulation, we first independently draw instance parameters  according to distributions in \Cref{table:variables}, and we then draw realizations for the random walk governing the time-varying interaction rate.
Using those parameters, we then compute the demand trajectory based on the dynamics of the SEIR model described by the system of equations in \eqref{eq:rec}. As mentioned before, we use the peak number of infected individuals in each location as a proxy for its need, i.e., $\Demand_i = \max_{t} I_i(t)$.

\subsubsection{Properties of the Demand Sequence} Having described our simulation setting, we next illustrate that even in this simple setting, the demand is highly variable and there is a complex correlation structure across locations. 
In the left panel of \Cref{fig:sims}, we present a histogram of total demand across 1,000 simulations of the model. We highlight that the standard deviation is large relative to expected demand: the coefficient of variation in the simulations is $0.662$. Furthermore, in the right three panels of \Cref{fig:sims}, we present a scatter plot which highlights the strong correlation between locations' demands that results from our SEIR model.  We remark that the complex correlation structure goes beyond adjacency-based network effects and further depends on the stochastic and time-varying nature of the interaction rate. Due to the linear nature of adjacency in this setting, the sequence of peak demands across the 1,000 simulations is consistent (first, location 1 reaches its peak demand, then location 2, etc.).\footnote{\revcolor{We remark that in general, it is possible for locations' demands to realize out-of-order. As discussed in \Cref{subsec:ourpolicy}, our \PolicyNameAb\ policy does not need to know the order of future arrivals; however, the consistent sequence avoids complications when computing the DP solution.}} The average temporal gap between peak demands is around three weeks in our simulations.

\begin{figure}[t]
    \centering
    \includegraphics[trim={0.0cm 0.5cm 0.0cm 0.2cm},clip,width = .98\textwidth]{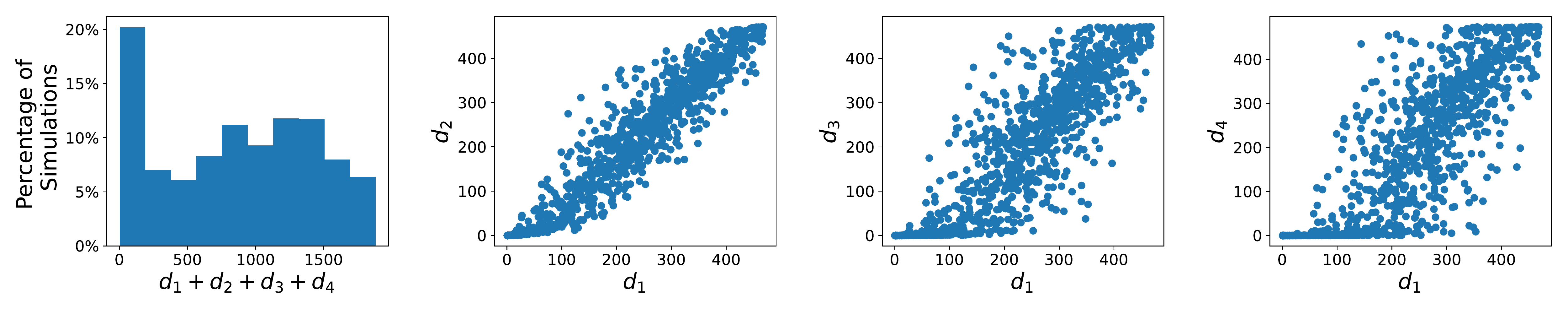}
    \caption{\revcolor{Left: Histogram of total demand. Middle and Right: Scatter plot of $(d_1, d_i)$ for $i \in \{2,3,4\}$.}}
\label{fig:sims}
\end{figure}

\subsection{Policy Performance}
For each of the 1,000 simulations of the setting described above, we numerically evaluate the performance of three policies: (i) our \PolicyNameAb\ policy, (ii) the optimal \fixedthresh\ policy, and (iii) a DP approach. We assume that initial supply is equal to the average demand across the simulations, which corresponds to a supply scarcity of $\ExpDemand = 1$. 

\subsubsection{Computing Policies}
In order to compute the allocation decision of our \PolicyNameAb\ policy we need to have access to conditional first moments of individual demand. Similarly, to compute the optimal \fixedthresh\ policy and the DP, we need access to the distribution of total demand and the full joint distribution, respectively.
Given that the joint distribution $\DemandJoinDist$ does not have an explicit representation, we approximate these quantities by their empirical averages. In particular, in order to estimate $\expect[\vec{\Demand}\sim\DemandJoinDist]{\ \sum_{j \in [i+1: \NumAgents]}\Demandi{j} \ \Big| \  \DemandVec_{[1:i]}}$, we first generate 1,000 sample paths for the setting described above. Then, for a given $\DemandVec_{[1:i]}$, we use a $k$-nearest neighbors approach (where $k=10$) to estimate conditional first moments of future demands.\footnote{\revcolor{Note that the $k$-nearest neighbors estimator serves as our oracle to generate samples for future demand: Fixing $\DemandVec_{[1:i]}$, it is unlikely that we have a sufficient number of sample paths with the exact same subsequence of demand. As such, we use the $k$ sample paths with the closest history to estimate the future demand. See Appendix \ref{apx:DP:hardness-input-model} for a more abstract discussion.}}
To compute the optimal \fixedthresh\ policy, we approximate $\sum_{i \in [\NumAgents]}\Demand_i$ by its empirical average, and we then numerically find the target threshold that performs best, which in this case is $\tau = 1$.

Since computing the DP requires the complete joint distribution of future demand, we increase the number of samples to 1,000,000 to provide a more accurate empirical average of the distribution.
We then follow the DP formulation presented in Appendix \ref{apx:DP:general} along with the discretization approach described in Appendix \ref{apx:DP:independent}. 
We stress that unlike the setting of Appendix \ref{apx:DP:independent}, here we have correlated demand. As such, we must augment the state space to include the entire demand history. 
Said differently, the Bellman equation we solve is a discrete version of \eqref{eq:bellman}.
However, given the small size of the problem, we are able to compute the DP solution when allowing for $50$ different discrete values for demand, which corresponds to $\epsilon \approx 0.01$ in the approach described in Appendix \ref{apx:DP:independent}.

\subsubsection{Policy Comparison}
In the first column of Table \ref{table:full_sims}, we present the average ex-post fairness for the aforementioned policies.\footnote{\revcolor{We remark that ex-post fairness is well-concentrated around its average value for all of the considered policies.}} In addition, we compute the average ex-post fairness achieved by the optimum offline policy. (For a given realization of demands $\DemandVec$, the optimum offline policy achieves a minimum FR of $\min\{1, s/\sum_{i \in [\NumAgents]}\Demand_i\}$.)
We make the following observations:

\begin{itemize} 
\item  {\bf PPA vs. Guarantee:} The \PolicyNameAb\ policy performs 30\% better than its theoretical guarantee of $0.60$ (from \Cref{thm:expost} when $\ExpDemand = 1$ and $\NumAgents = 4$) in this non-adversarial setting. 
\item {\bf PPA vs. Optimal TFR:} The \PolicyNameAb\ policy outperforms the optimal \fixedthresh\ policy by 44\%, and we highlight that the optimal \fixedthresh\ policy performs even worse than the \PolicyNameAb\ guarantee. Intuitively, in this pandemic-based setting where correlation is strong and the total demand has a relatively large coefficient of variation, adaptivity is particularly valuable.
\item {\bf PPA vs. DP:} The solution to the DP, despite being challenging to compute even in this simple setting, exhibits nearly identical performance to the \PolicyNameAb\ policy. This further implies that using additional information beyond first conditional moments does not provide significant benefit {in our numerical simulations}. 
\item {\bf PPA vs. Optimum Offline:}  The PPA policy even achieves $94\%$ of the optimal performance when the demand sequence is known in advance. (We remind that \Cref{prop:offlinecomparison} established that in general, no online policy can provide a constant-factor guarantee relative to this benchmark. However, in this practical instance, we are able to provide strong performance relative to the optimum offline.)
\end{itemize}

\begin{table}[h]
    \caption{\revcolor{Ex-post fairness and waste of three policies across 1,000 simulations when the sample paths are accurately drawn (columns 1-2) and for two different types of mis-specification (columns 3-6).}}
\footnotesize
    \centering
    \revcolor{\begin{tabular}{lcccccccc}
         & \multicolumn{2}{c}{Base Setting}&& \multicolumn{2}{c}{Mis-specified $\hypermean$}&&\multicolumn{2}{c}{Mis-specified $\inftorecparam$} \\\cmidrule{2-3} \cmidrule{5-6} \cmidrule{8-9}
        &Ex-post Fairness&Waste&&Ex-post Fairness&Waste&&Ex-post Fairness&Waste  \\\cmidrule{2-3} \cmidrule{5-6} \cmidrule{8-9}
         \PolicyNameAb\ Policy &0.782&0.007&&0.776&0.010&&0.778&0.008 \\
         Opt. \fixedthresh\ Policy \ \  &0.544&0&&0.469&0.224&&0.544&0  \\
        DP Solution &0.785&0.011&&0.783&0.016&&0.744&0.014 
        \\
        Optimum Offline &0.831&0&&0.831&0&&0.831&0 
    \end{tabular}}
    \vspace{4mm}

    \label{table:full_sims}
\end{table}

\subsubsection{Additional Practical Considerations} Moving beyond the comparison in terms of ex-post fairness, we also 
examine the allocative efficiency of our policy by reporting the average fraction of wasted supply in the second column of \Cref{table:full_sims}. Supply is \emph{wasted} whenever there is remaining supply that could have been allocated to an agent without exceeding their demand, which is undesirable in our motivating applications. Mathematically, the fraction of wasted supply is given by $\frac{\min\{s, \sum_{i \in [\NumAgents]}\Demandi{i}\} - \sum_{i \in [\NumAgents]}\Alloci{i}}{s}$. By definition, the optimum offline policy never wastes supply, and similarly, the optimal \fixedthresh\ policy in this setting -- where the optimal target fill rate is one -- also does not waste supply. The \PolicyNameAb\ policy does waste some supply in expectation, but this waste is minimal (less than $1\%$) and slightly less than the amount of waste under the DP solution.

Furthermore, in practice, various model parameters may be mis-specified. Therefore, it is practically important to rely on policies that are not too sensitive to such model mis-specifications. In the following, we examine the robustness of the three aforementioned policies and show that our \PolicyNameAb\ policy outperforms both the TFR policy and the DP in terms of robustness.
Specifically,
we consider two different scenarios of mis-specification in the parameters of the SEIR model. In scenario (i), when generating sample paths, the hyper-parameter $\hypermean$ (which governs the drift of the time-varying interaction rate) is sampled uniformly from the interval $[-0.05, 0.05]$ instead of the interval $[-0.08, 0.02]$, which leads to over-estimates of demand by an average of around $25\%$.\footnote{\revcolor{In a complementary scenario where the distribution of the hyper-parameter is sampled from $[-0.11, -0.01]$, the performance of the \PolicyNameAb\ policy is no worse.}} In scenario (ii),  the expected length of infection is underestimated by $20\%$ (i.e, $\inftorecparam$ is assumed to be $0.125$ instead of $0.1$), which leads to under-estimates of demand by an average of around $25\%$.\footnote{\revcolor{In a complementary scenario where the expected length of infection is over-estimated by $20\%$, the performance of the \PolicyNameAb\ policy is no worse.}}
In the middle (resp. right) two columns of \Cref{table:full_sims}, we present the ex-post fairness and the waste of all three policies on the same $1,000$ simulations, but when the policies are calibrated on sample paths generated by the mis-specified model from scenario (i) (resp. scenario (ii)).

We observe that the performance of the \PolicyNameAb\ policy remains relatively stable, despite the substantial levels of mis-specification. In contrast, the optimal \fixedthresh\ policy's performance changes somewhat dramatically in scenario (i). In this scenario, where demand is over-estimated, the optimal target fill rate drops precipitously from $1.0$ to $0.492$ due to a greater perceived risk of exhausting supply. This new target fill rate impacts both the ex-post fairness and the average fraction of supply that is wasted.  Similarly, the performance of the DP solution is substantially impacted in one of the two scenarios, this time in scenario (ii). In this scenario, when demand is under-estimated, the DP can be overly aggressive in its allocation decisions and end up with a near-$0$ minimum FR in certain simulations. The heavier left-tail for the distribution of the DP's minimum FR drives the average performance down, ultimately leading to a decrease in ex-post fairness by more than $5\%$.

We conclude this section by highlighting that we are only able to compute the DP solution due to limiting the setting to four agents. Yet despite this simplified setting, the solution to the DP also suffers from the additional shortcomings described in \Cref{rem:online}: it requires full distributional knowledge (which may be difficult to acquire and more prone to mis-specification), it achieves worse ex-ante fairness than the \PolicyNameAb\ policy (by nearly 5\% in these simulations), and its allocation is not transparent. Transparency is particularly important in this setting, as numerous states questioned the allocation procedures implemented by the federal government (see Footnote \ref{foot:untransparent} for one such example). In contrast, both the \PolicyNameAb\ policy and the optimal \fixedthresh\ policy follow strategies that can be easily explained to stakeholders. 
}

\section{Concluding Remarks, Extensions, and Future Directions}
We conclude the paper by summarizing our main findings, discussing a few extensions of our base framework, and listing a few future directions.
\label{sec:discussions}

\smallskip
\subsection{Summary} In this paper, we initiate the theoretical study of fair dynamic rationing by introducing a simple yet fundamental and well-motivated framework. In a nutshell, we design sequential policies for allocating limited supply to a sequence of arbitrarily correlated demands given an objective which encompasses the dual goals of efficiency and equity. Based on our formalized notions of ex-post and ex-ante fairness, we establish upper bounds on the fairness guarantees achievable by any sequential allocation policy which depend on the supply's scarcity level and the number of demanding agents.
More importantly, we show that our simple PPA policy achieves  the ``best of both worlds'' by attaining the upper bound on both the ex-post and ex-ante fairness guarantees.
In addition to enjoying optimal fairness guarantees, our PPA policy is practically appealing: it is interpretable as well as computationally efficient since it does not rely on
distributional knowledge beyond the conditional first moments.

Our framework 
lends itself to extensions such as considering generalized objectives and rationing multiple types of resources. More broadly, it serves as a base model for {theoretically} studying sequential allocation problems with an objective beyond utility maximization, which in turn opens several new research directions. 
In the rest of this section, we first discuss the aforementioned extensions of our base model and then finish the paper by discussing future directions.

\subsection{Extensions}
\label{subsec:variants}

\subsubsection{Generalized Social Welfare Objective Functions}
\label{subsec:differentobjectives}
Throughout the paper, we {have focused} on the minimum FR 
as a social welfare objective that combines elements of equity and efficiency (i.e., $\SocialWelfare (\vec{x})= \min_{i \in [\NumAgents]}\{ \frac{\Alloci{i}}{\Demandi{i}}\}$). 
However, this social welfare function---which is also known as the Rawlsian social welfare function thanks to the philosophical work of John Rawls \citeyearpar{rawls1973theory}---is only a special case of a more general 
class of social welfare functions that we call the \emph{weighted power mean (WPM) social welfare} {family of} functions. More precisely, this family is parameterized by $\FairParam \in [0, +\infty)$ and defined as 
\begin{equation}
\label{eq:social-welfare}
    \SocialWelfare_\FairParam(\vec{\Alloc}) \triangleq \begin{cases}\left( \frac{1}{\sum_{i \in [\NumAgents]} \Demandi{i}} \sum_{i \in [\NumAgents]} \Demandi{i} \left(\frac{\Alloci{i}} {\Demandi{i}}\right)^{1-\FairParam}\right)^{1/(1-\FairParam)}, &\FairParam \neq 1 \\
    \\
    \prod_{i \in [\NumAgents]} \left(\frac{\Alloci{i}} {\Demandi{i}}\right)^{\Demandi{i} / \sum_{i \in [\NumAgents]} \Demandi{i}}, &\FairParam = 1
    \end{cases}
\end{equation}
Note that the above is a weighted version of the celebrated power mean functions, introduced in \cite{atkinson1970measurement}, 
that provides a broad class of social welfare functions which balance equity and efficiency to varying degrees. 
Having weights proportional to the demands in \cref{eq:social-welfare} ensures that equity is measured {\em relative} to demand and not simply based on the absolute allocation.\footnote{Further, this family of functions has a one-to-one relationship with the $\alpha$-fairness social welfare functions introduced in \cite{mo2000fair}. In fact, the two families of functions are the same up to a transformation via a one-to-one, increasing function, which means that the maximizing vectors are identical for a given $\FairParam$.}
By varying the parameter $\alpha$ from $0$ to $+\infty$, the focus of the planner is shifted from extreme efficiency towards more equitable allocations. When $\FairParam = 0$, a utilitarian allocation (i.e., any allocation without waste) maximizes social welfare. In the limit as $\FairParam \rightarrow 1$, proportional fairness (i.e., a generalization of the Nash bargaining solution) maximizes social welfare. Finally, in the limit as $\FairParam \rightarrow +\infty$, maximizing the minimum FR maximizes social welfare. In fact, we highlight that the value of this social welfare function exactly approaches our objective in the base model (i.e., the minimum FR, or equivalently, the Rawlsian social welfare function).

For any parameter $\FairParam$ (including when $\FairParam \rightarrow +\infty$, which corresponds to our base model), the optimal policy when demands are deterministic is to allocate the supply proportionally. Such a policy achieves the optimal social welfare of $\DetGuar$ (defined in \cref{eq:det:fair}).\footnote{Note that {in a deterministic setting,} maximizing any social welfare objective function as in \cref{eq:social-welfare} is a concave maximization. By writing KKT conditions it is not hard to see that  any such function attains its maximum at a feasible proportional allocation, i.e., $x_i=\min\left\{d_i,\frac{d_i}{\sum_{j\in[n]}d_j}\right\}$, under deterministic demands. The maximum is then equal to $\DetGuar=\min\{1,1/\ExpDemand\}.$} However, the value of the parameter $\FairParam$ impacts the optimal policy when demands are stochastic. To study this impact, we naturally generalize our notion of ex-post fairness to WPM social welfare functions, i.e., ex-post fairness is given by  $\expect[\vec{\Demand}\sim\DemandJoinDist]{\SocialWelfare_{\FairParam}}/\DetGuar$. We remark that in the limit as $\FairParam \rightarrow +\infty$, this is equivalent to the notion of ex-post fairness introduced in \Cref{sec:prelim}.
In the following  corollary of  \Cref{thm:expost}, we establish that our \PolicyNameAb\ policy achieves an ex-post fairness guarantee of at least $\LowerboundFunEx$ {for any $\FairParam \in [0, +\infty)$.}

\begin{corollary}[PPA's Guarantee for WPM Objectives]
\label{cor:generalfunctions}
Given a fixed number of agents $\NumAgents\in\mathbb{N}$, supply scarcity $\ExpDemand\in\NonNegReals$, and any $\FairParam \in [0, +\infty)$, the \PolicyNameAb ~policy achieves an ex-post fairness guarantee of at least 
$\LowerboundFunEx$ (defined in \cref{eq:kappa-expost}) when social welfare is measured by a WPM function {(defined in \cref{eq:social-welfare})} with parameter $\FairParam$.
\end{corollary}
We prove Corollary \ref{cor:generalfunctions} in Appendix \ref{apx:cor:generalfunctions}.

\subsubsection{Rationing Multiple Types of Resources}
\label{subsec:multiplegoods}
In our base model, we assume that agents have demand for only a single type of resource. However, in many of the motivating applications that we consider, agents may have concurrent demand for multiple types of resources. For example, states may need many different types of medical supplies during the peak of a pandemic.

Our setup readily extends to the sequential allocation of $\NumGoods$ different resource types where an arriving agent {\em simultaneously} demands $\NumGoods$ types of supply. We allow the demands to be correlated across agents as well as resource types. For the sake of brevity, we refrain from repeating the setup and we simply augment our notation for various quantities (e.g., supply, demand, and allocation) by adding a superscript $j \in [\NumGoods]$. 
 In this generalized model, we define agent $i$'s utility to be their 
\emph{weighted} FR,  defined as:
\begin{align}
    \label{eq:wFR}
    \sum_{j \in [m]} \Weightj{j} \frac{\Alloci{i}^j}{\Demandi{i}^j},
\end{align}
where we normalize the weights $\Weightj{j}$ to satisfy $\sum_{j \in [m]} \Weightj{j} = 1$. 
A simple corollary of  \Cref{thm:expost}, as stated below, ensures that 
independently following our \PolicyNameAb\ policy for each resource achieves a lower bound on the expected minimum weighted FR which is a weighted sum of the expected minimum FR guaranteed by the \PolicyNameAb\ for one resource, i.e. $\LowerboundFunPj \max\{1, \ExpDemand^j\}$ for resource $j$.

\begin{corollary}[PPA's Guarantee on Expected Minimum Weighted FR]
\label{cor:multiplegoods}
{Consider any instance with $\NumAgents\in\mathbb{N}$ 
agents and $\NumGoods \in \mathbb{N}$ resources, where 
the initial supply for resource $j$ is $\Supply^{j} \in \NonNegReals$. For any  joint demand distribution over all agents and resources $\DemandJoinDist \in \Delta (\NonNegReals^{\NumAgents \times \NumGoods})$, independently following the \PolicyNameAb\ policy for each resource achieves an expected minimum weighted FR (as defined in \cref{eq:wFR}) of at least $\sum_{j \in [\NumGoods]} \Weightj{j} \LowerboundFunPnoarg\left(\frac{\ExpDemand^j}{\Supply^j}, \NumAgents\right) \max\left\{1, \frac{\ExpDemand^j}{\Supply^j}\right\}$, where $\ExpDemand^{j} \in \NonNegReals$ is the
expected total demand for resource $j$.}
\end{corollary}

In addition, 
{since demand can be correlated across agents, we can re-use the hard instances of \Cref{thm:hardness-ex-post} to construct a joint distribution which establishes an upper-bound on the performance of any policy matching the lower bound given in Corollary \ref{cor:multiplegoods}. We state this upper bound} as a corollary of \Cref{thm:hardness-ex-post} {below.}

\begin{corollary}[Upper Bound on Expected Minimum Weighted FR]
\label{cor:multiplegoodsupper}
{For any $\NumAgents\in\mathbb{N}$ 
agents, $\NumGoods \in \mathbb{N}$ resources, and any initial supply for resource $j$ of $\Supply^{j} \in \NonNegReals$, there exists a joint demand distribution over all agents and resources $\DemandJoinDist \in \Delta (\NonNegReals^{\NumAgents \times \NumGoods})$ for which no policy can achieve an expected minimum weighted FR greater than $\sum_{j \in [\NumGoods]} \Weightj{j} \LowerboundFunPnoarg\left(\frac{\ExpDemand^j}{\Supply^j}, \NumAgents\right) \max\left\{1, \frac{\ExpDemand^j}{\Supply^j}\right\}$, where $\ExpDemand^{j} \in \NonNegReals$ is the
expected total demand for resource $j$.}
\end{corollary}

Together, these two corollaries (which we prove in Appendix \ref{apx:discussionsection}) establish that independently following the \PolicyNameAb\ policy for each resource $j$ provides the best possible guarantee on the expected minimum weighted FR. Consequently, we can use our \PolicyNameAb\ policy to
shed light on how the social planner can prepare for demand across multiple types of resources.
If the initial endowment of different resource types is not exogenously set, then the social planner can solve an outer endowment optimization problem to maximize the guarantee on the expected minimum weighted FR subject to a budget constraint.

We remark that such an endowment optimization problem is a max-max-min problem where the social planner first optimizes the initial endowment across resource types subject to a budget constraint (the outer problem). Then, given the initial endowment, the social planner  maximizes over policies the minimum over demand distributions of our objective, i.e., the expected minimum weighted FR among agents.
We solve the outer maximization of the multiple resource-type problem by determining the optimal initial endowment across resource types when the social planner independently follows the \PolicyNameAb\ policy for each resource. To be concrete, suppose the social planner has a fixed budget $\Budget$ that can be used to procure an initial endowment $\vec{\Supply} = (\Supply^1, \Supply^2, \dots, \Supply^\NumGoods)$.
Further, suppose the per unit cost of resource $j \in [m]$ is $\Costj{j}$. Then, the outer endowment optimization problem can be formulated as follows: 
\begin{equation}
\label{eq:budget-opt}
\tag{\textsc{Endowment Optimization}}
\arraycolsep=1.4pt\def\arraystretch{1}
\begin{array}{lllr}
 \underset{\displaystyle\vec{\MultiGoodSupply}\in{\NonNegReals^\NumGoods}}{\textrm{max}}\quad\quad&
\displaystyle \sum_{j \in [m]} \Weightj{j} \LowerboundFunPnoarg\left(\frac{\ExpDemand^j}{\MultiGoodSupply^j}, \NumAgents\right) \max\left\{1, \frac{\ExpDemand^j}{\MultiGoodSupply^j}\right\}
\qquad\qquad&\text{s.t.}&
\\[1.4em]
& \displaystyle \Budget  ~\geq~ \sum_{j \in [m]}\Costj{j}\MultiGoodSupply^j & &
\end{array}
\end{equation}

We highlight that to formulate this optimization problem, we crucially use the parameterized characterization of 
the ex-post fairness guarantee in Theorem \ref{thm:expost}, as opposed to the worst-case guarantee for any set of parameters.  Further, we remark that the above maximization problem is linearly separable and concave in the decision variables $\vec{\MultiGoodSupply}$, meaning that it can be solved efficiently.\footnote{It is not difficult to check that the function $\LowerboundFunPnoarg\left(\frac{\ExpDemand^j}{\MultiGoodSupply^j}, \NumAgents\right) \max\left\{1, \frac{\ExpDemand^j}{\MultiGoodSupply^j}\right\}$ is concave in $\MultiGoodSupply^j$ for any choice of $\ExpDemand^j$ and $\NumAgents$; check \cref{eq:kappa-expost} for a definition of $\kappa_\textsc{p}(\cdot, \cdot)$. We omit this purely algebraic proof for brevity.} Based on Corollaries \ref{cor:multiplegoods} and \ref{cor:multiplegoodsupper}, the optimal solution $\vec{\MultiGoodSupply}^*$, combined with independently implementing the \PolicyNameAb\ policy for each resource, is indeed the optimal solution of the max-max-min multiple resource-type problem.

\subsection{Future Directions}
Our paper 
can be viewed as an analog of the {classic} prophet inequality problem~\citep{krengel1978semiamarts,samuel1984comparison} for equitably allocating divisible goods. As such, similar to prophet inequalities, many interesting variants of our setting arise. We discussed two such variants above, and for both, we established an achievable lower bound by employing our \PolicyNameAb\ policy. However, in the former variant, we do not establish a matching upper bound.  Establishing tight bounds on the achievable performance in such a setting---which may require the use of a different policy---is an interesting direction for future research.   
Further, understanding the inefficiency (unused supply) which may occur in sequential allocation due to our focus on an egalitarian objective is a fruitful research direction that we plan to pursue. 
Finally, here we made no assumption about the correlation structure underlying the demand sequence. It would be compelling to investigate whether including a (well-motivated) correlation structure can result in improved fairness guarantees. 

\section*{Acknowledgment}
The authors would like to thank Itai Ashlagi, Amin Saberi, and Ed Kaplan for helpful comments and insights at early stages of this work.

\setlength{\bibsep}{0.0pt}
\bibliographystyle{plainnat}
\OneAndAHalfSpacedXI
{\footnotesize
\bibliography{refs}}
 
 \newpage
 \renewcommand{\theHsection}{A\arabic{section}}
 \begin{APPENDIX}{}

\section{\texorpdfstring{Missing Proofs of \Cref{subsec:expostupperbound}}{}}
\label{apx:upperboundproofs}
\subsection{Proof of \texorpdfstring{\Cref{thm:hardness-ex-post}}{}  (\texorpdfstring{\Cref{subsec:expostupperbound}}{})}
\label{apx:upper-bound}
We prove the theorem by considering two separate cases corresponding to the over-demanded regime ($\ExpDemand\geq 1+\tfrac{1}{\NumAgents}$) and the under-demanded regime ($\ExpDemand < 1+\tfrac{1}{\NumAgents}$). For each regime, we provide an instance of the problem under which no sequential allocation policy obtains ex-post fairness larger than $\LowerboundFunEx$ restricted to that regime.

{\bf Over-demanded regime ($\bm{\ExpDemand \geq 1 +\frac{1}{\NumAgents}}$):} Consider an instance with $\NumAgents$ equally likely scenarios, where in scenario $\sigma$ all agents $i \in [\sigma]$ have demand $\Demandi{i} = \frac{2\ExpDemand}{\NumAgents+1}$. This instance is depicted in \Cref{fig:hardexample}. We remark that the total expected demand is equal to $\ExpDemand$, simply because
$$\sum_{\sigma \in [\NumAgents]} \frac{1}{\NumAgents} \cdot \frac{2\ExpDemand}{\NumAgents+1} \cdot \sigma = \ExpDemand.$$

In such a setting, whenever agent $i$ has non-zero demand, every agent $j$ where $j < i$ must also have had non-zero demand. Since the policy cannot distinguish among scenarios $i, i+1, \dots, \NumAgents$, its allocation decision must be independent of the scenario. Therefore, any policy can be described by a set of allocations $\vec{\AppAlloc} = (\AppAlloci{1}, \AppAlloci{2}, \dots, \AppAlloci{\NumAgents})$, such that if agent $i$ has non-zero demand, then they receive an expected allocation $\AppAlloci{i}$.  Furthermore, when making the allocation decision for agent $\NumAgents$, there is only one possible history: every other agent also had non-zero demand. Thus, any feasible sequential allocation policy must respect the constraint $\sum_{i \in [\NumAgents]} \AppAlloci{i} \leq 1$. 

Let us define $\AppMinFilli{\sigma}$ as the expected minimum FR of the given policy in scenario $\sigma$ (i.e., if only the first $\sigma$ agents have non-zero demand). By convention, we set $\AppMinFilli{0} = 1$ and we must have $\AppMinFilli{\sigma} \leq \AppMinFilli{\sigma-1}$ by definition. In addition, $\AppMinFilli{\sigma}$ must be less than the expected FR of agent $\sigma$, so $\AppMinFilli{\sigma} \leq \frac{(\NumAgents+1)}{2\ExpDemand}\AppAlloci{\sigma}$. Given $\vec{\AppMinFill}$, the expected minimum FR in this instance is equal to $ \frac{1}{\NumAgents}\sum_{\sigma \in [\NumAgents]} \AppMinFilli{\sigma}$. Based on these constraints and the objective, we can formulate a linear program whose optimal solution is an upper bound on the expected minimum FR achievable by any feasible sequential allocation policy. This linear program was originally presented in \Cref{subsec:expostupperbound}, but we replicate it here (\textsc{Primal-LP1}) along with its dual program (\textsc{Dual-LP1}).

\begin{equation*}
\arraycolsep=1.4pt\def\arraystretch{1}
\begin{array}{llll|lll}
(\textsc{Primal-LP1})& &&&~~(\textsc{Dual-LP1})  && \\[1.4em]
\underset{\displaystyle\vec{\AppAlloc},\vec{\AppMinFill}\in \NonNegReals^\NumAgents}{\textrm{max}}\quad\quad&
\displaystyle\frac{1}{\NumAgents}\sum_{\sigma\in[\NumAgents]}\AppMinFilli{\sigma}
\quad\qquad\qquad\qquad\qquad&\text{s.t.}&
&\quad\quad\underset{\substack{\displaystyle\vec{\gamma},\vec{\delta}\in\NonNegReals^\NumAgents\\\displaystyle \omega\in \NonNegReals}}{\textrm{min}}\qquad\quad &\displaystyle\omega + \gamma_1 
\quad\quad\qquad\qquad&\text{s.t.}\qquad\qquad \\[1.4em]
& \displaystyle \AppMinFilli{\sigma}~\leq~\frac{(\NumAgents+1) \AppAlloci{\sigma}}{2\ExpDemand} & \sigma\in[\NumAgents]& & &\displaystyle\omega~\geq~\frac{(\NumAgents+1)}{2\ExpDemand}\delta_i&i\in[\NumAgents]\\[1.4em]
& \displaystyle \AppMinFilli{\sigma}~\leq~\AppMinFilli{\sigma-1},~~\AppMinFilli{0}=1 & \sigma\in[\NumAgents]& & &\displaystyle\delta_i~\geq~\frac{1}{\NumAgents}+\gamma_{i+1}-\gamma_i\qquad\qquad&i\in[\NumAgents-1]\\[1.4em]
& \displaystyle\sum_{i\in[\NumAgents]}\AppAlloci{i}\leq 1 & & & &\displaystyle\delta_\NumAgents~\geq~\frac{1}{\NumAgents}-\gamma_\NumAgents&\\[1.4em]
\end{array}
\end{equation*}

To upper-bound the value of the program \textsc{Primal-LP1}, we find a feasible assignment for the dual program \textsc{Dual-LP1}. Consider the assignment where $\delta_i = \frac{1}{\NumAgents}$ and $\gamma_i = 0$ for all $i \in [\NumAgents]$, and where $\omega = \frac{\NumAgents+1}{2 \NumAgents \ExpDemand}$.  Under this assignment, all dual variables are non-negative and all constraints are satisfied (in fact, all are tight). Thus, this assignment is feasible in \textsc{Dual-LP1}. It also attains an objective value of $\frac{\NumAgents+1}{2\NumAgents \ExpDemand}$. By weak duality, this represents an upper bound on the optimal value of \textsc{Primal-LP1}, and hence an upper bound on the expected minimum FR of any policy in the over-demanded regime.

\begin{figure}
    \centering
     \includegraphics[trim={5cm 2cm 5cm 2cm},clip,width=.65\textwidth]{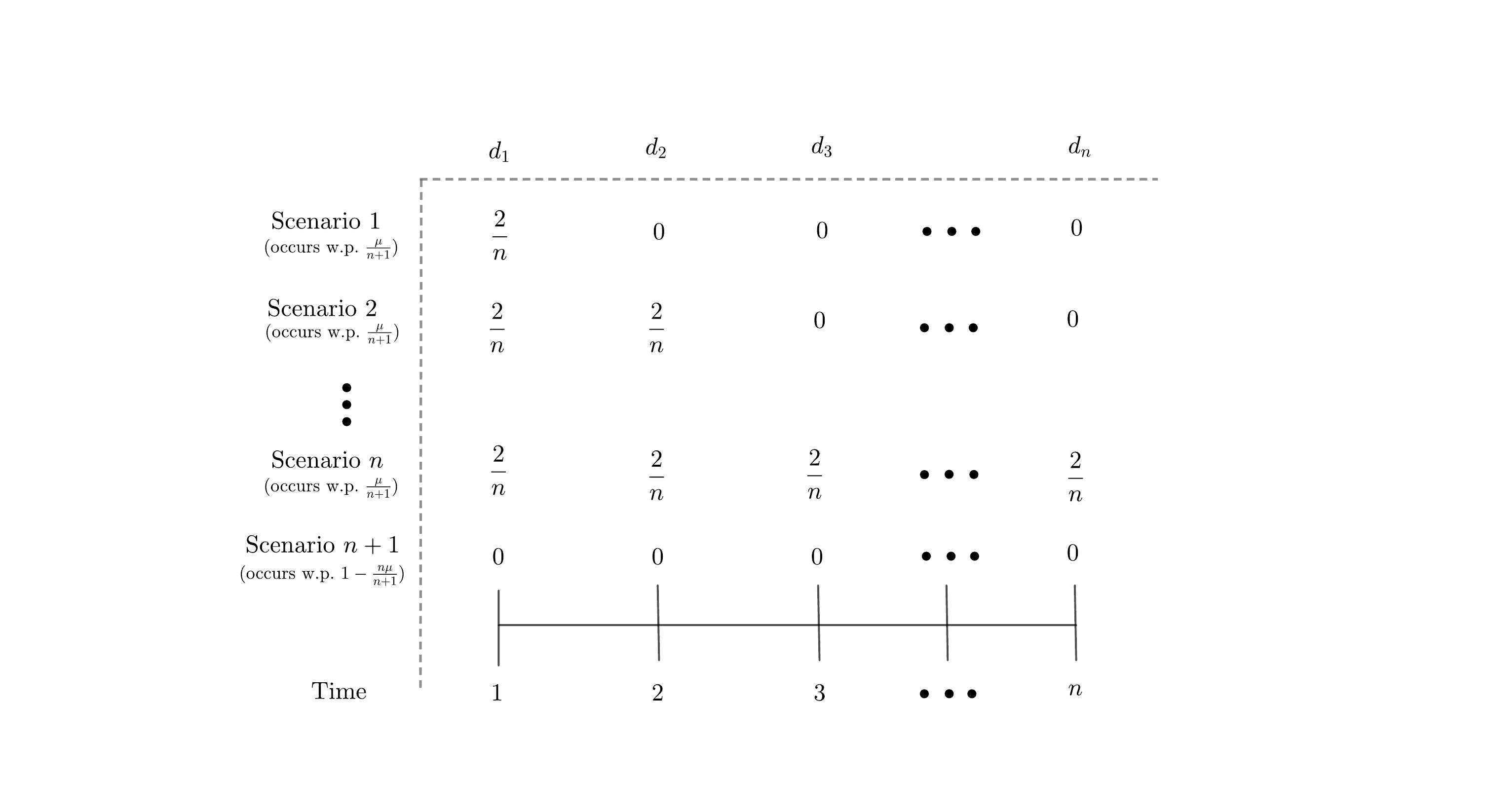}    \caption{The instance which establishes an upper bound of $\bm{\LowerboundFunEx}$ when $\bm{\ExpDemand < 1 +\frac{1}{\NumAgents}}$ (which we refer to as the under-demanded regime).}
    \label{fig:hardexample_underdemand}
\end{figure}

{\bf Under-demanded regime ($\bm{\ExpDemand < 1 +\frac{1}{\NumAgents}}$):} Consider an instance with $\NumAgents+1$ scenarios, where the first $\NumAgents$ scenarios each occur with equal probability of $\frac{\ExpDemand}{\NumAgents+1}$ and scenario $\NumAgents+1$ occurs with probability $1- \frac{\NumAgents\ExpDemand}{\NumAgents+1} $. In scenario $\NumAgents+1$, there is no demand. In scenario $\sigma$ for $\sigma \in [\NumAgents]$, all agents $i \in [\sigma]$ have demand $\Demandi{i} = \frac{2}{\NumAgents}$. This instance is depicted in \Cref{fig:hardexample_underdemand}. We remark the total expected demand is equal to $\ExpDemand$, simply because
$$
\sum_{\sigma \in [\NumAgents]} \frac{\ExpDemand}{\NumAgents+1} \cdot\frac{2}{\NumAgents}\cdot \sigma = \ExpDemand.
$$

As was the case in the over-demanded regime, any sequential allocation policy can be  described by a set of allocation decisions such that if agent $i$ has non-zero demand, then they receive an expected allocation $\AppAlloci{i}$. We again define $\AppMinFilli{\sigma}$ as the expected minimum FR in scenario $\sigma$, and we note that $\AppMinFilli{\NumAgents+1} = 1$. Thus, the expected minimum FR in this instance is equal to $\frac{\ExpDemand}{\NumAgents+1}\sum_{\sigma\in[\NumAgents]}\AppMinFilli{\sigma} + 1-\frac{\NumAgents\ExpDemand}{\NumAgents+1}$. By imposing the constraints described for \textsc{Primal-LP1}, we can formulate a slightly different linear program whose optimal solution is an upper bound on the expected minimum FR of any feasible sequential allocation policy in the under-demanded regime. This linear program (\textsc{Primal-LP2}), along with its dual program (\textsc{Dual-LP2}), is presented below.

\begin{equation*}
\arraycolsep=1.4pt\def\arraystretch{1}
\begin{array}{llll|lll}
(\textsc{Primal-LP2})& &&&~~(\textsc{Dual-LP2})  && \\[1.4em]
\underset{\displaystyle\vec{\AppAlloc},\vec{\AppMinFill}\in \NonNegReals^\NumAgents}{\textrm{max}}\quad\quad&
\displaystyle\frac{\ExpDemand}{\NumAgents+1}\sum_{\sigma\in[\NumAgents]}\AppMinFilli{\sigma} + 1-\frac{\NumAgents\ExpDemand}{\NumAgents+1}
\quad\quad\quad\qquad&\text{s.t.}&
&\quad\quad\underset{\substack{\displaystyle\vec{\gamma},\vec{\delta}\in\NonNegReals^\NumAgents\\\displaystyle \omega\in \NonNegReals}}{\textrm{min}}\quad\quad &\displaystyle\omega + \gamma_1  + 1-\frac{\NumAgents\ExpDemand}{\NumAgents+1}
\quad\quad\qquad\qquad&\text{s.t.} \\[1.4em]
& \displaystyle \AppMinFilli{\sigma} ~\leq~\frac{\NumAgents \AppAlloci{\sigma}}{2} & \sigma\in[n]& & &\displaystyle\omega~\geq~\frac{\NumAgents}{2}\delta_i&i\in[\NumAgents]\\[1.4em]
& \displaystyle \AppMinFilli{\sigma}~\leq~\AppMinFilli{\sigma-1},~~\AppMinFilli{0}=1 & \sigma\in[\NumAgents]& & &\displaystyle\delta_i~\geq~\frac{\ExpDemand}{\NumAgents+1}+\gamma_{i+1}-\gamma_i&i\in[\NumAgents-1]\\[1.4em]
& \displaystyle\sum_{i\in[\NumAgents]}\AppAlloci{i}\leq 1 & & & &\displaystyle\delta_\NumAgents~\geq~\frac{\ExpDemand}{\NumAgents+1}-\gamma_\NumAgents&\\[1.4em]
\end{array}
\end{equation*}

To upper-bound the value of the program \textsc{Primal-LP2}, we find a feasible assignment for its dual. Consider $\delta_i = \frac{\ExpDemand}{\NumAgents+1}$ and $\gamma_i = 0$ for all $\sigma \in [\NumAgents]$, and $\omega = \frac{\NumAgents\ExpDemand}{2(\NumAgents+1)} $. Under this dual assignment, all dual variables are non-negative and all constraints are satisfied (and again tight). Thus, this assignment is feasible in \textsc{Dual-LP2}. It also attains an objective value of $ \frac{\NumAgents\ExpDemand}{2(\NumAgents+1)} 1-\frac{\NumAgents\ExpDemand}{\NumAgents+1}= 1- \frac{\NumAgents\ExpDemand}{2(\NumAgents+1)}$. By weak duality, this represents an upper bound on the optimal value of $\textsc{Primal-LP2}$, and hence an upper bound on the expected minimum FR attainable by any sequential allocation policy when $\ExpDemand < 1+ \frac{1}{\NumAgents}$.

We conclude the proof by scaling the obtained upper bounds on the expected minimum FR by our normalization factor, namely $\DetGuar = \min\{1, 1/\ExpDemand\}$. This establishes an upper bound of $\LowerboundFunEx$ on the ex-post fairness guarantee (see \Cref{def:polperf}) achievable by any sequential allocation policy. \hfill\Halmos

\revcolor{
\subsection{\texorpdfstring{Proof of \Cref{prop:offlinecomparison} (\Cref{subsec:expostupperbound})}{}}
\label{apx:offline}
To establish this result, we make use of the hard instance when $\ExpDemand \geq 1+1/\NumAgents$ which was presented in the proof of \Cref{thm:hardness-ex-post} and illustrated in \Cref{fig:hardexample}. For ease of reference, we describe the instance below.

In this instance, there are $\NumAgents$ possible equally-likely scenarios, i.e.,  scenario $\sigma$ happens with probability $1/\NumAgents$ for $\sigma \in [\NumAgents]$. 
In scenario $\sigma$, the first $\sigma$ agents have equal demand of $\frac{2\ExpDemand}{\NumAgents+1}$ and the rest have no demand. As established in the proof of \Cref{thm:hardness-ex-post}, the expected demand in this instance is $\ExpDemand$, and no online policy can achieve an expected minimum FR greater than $\frac{\NumAgents+1}{2\NumAgents \ExpDemand}$.

The optimal offline solution is to divide supply proportionally across all \emph{realized} demand. This solution achieves a minimum FR of either $1$ or $\frac{1}{\sum_{i\in[\NumAgents]}\Demandi{i}}$. For this example, we will assume that $\frac{2\ExpDemand}{\NumAgents+1} \geq 1$, which means that each demand exceeds supply. This ensures that the minimum FR achieved by the optimal offline solution is exactly $\frac{1}{\sum_{i\in[\NumAgents]}\Demandi{i}}$.

Hence, the expected minimum FR achieved by the optimal offline solution is simply equal to 
\begin{align}
\expect[\vec{\Demand}\sim\DemandJoinDist]{ \frac{1}{\sum_{i\in[\NumAgents]}\Demandi{i}}} \quad = \quad \frac{1}{\NumAgents}\left(\sum_{\sigma \in [\NumAgents]} \frac{\NumAgents+1}{2\ExpDemand \sigma}\right) \nonumber 
\quad \geq  \quad \frac{\NumAgents+1}{2\NumAgents\ExpDemand}\int_{1}^{\NumAgents+1} \frac{1}{x}\ \partial x \label{eq:offline1} \quad = \quad \frac{\NumAgents+1}{2\NumAgents\ExpDemand} \log(\NumAgents+1) \nonumber
\end{align}
The inequality holds due to the summation representing a left Riemann sum of a decreasing function. As noted above, the upper bound on the 
expected minimum FR of any online policy is $\frac{\NumAgents+1}{2\NumAgents\ExpDemand}$. Thus, the ratio between the two is at least $\log(\NumAgents+1)$, which completes the proof of \Cref{prop:offlinecomparison}. \hfill\Halmos
}

\revcolor{
\section{Dynamic Programming for Optimum Online (\texorpdfstring{\Cref{subsec:expostupperbound}}{})}
\label{apx:DP}
In this section we describe the dynamic programming for computing the optimum online policy for maximizing the expected minimum FR. Given a joint demand distribution $\DemandJoinDist$, the optimum online policy $\Policy^*$ maximizes $\ExpostObj^{\DemandJoinDist}(\Policy^*)=\expect[\Vec{\Demandi{}}\sim\DemandJoinDist]{\underset{i\in[n]}{\min}\left\{\frac{\Alloci{i}^*}{\Demandi{i}}\right\}}$, where $\{\Alloci{i}^*\}_{i\in[n]}$ denote the allocation decisions of policy $\Policy^*$ when is run on the stochastic demand vector $\vec{d}\sim\DemandJoinDist$. In Appendix \ref{apx:DP:general} we describe this DP in its most general form, including for instances with a continuous state space. However, even in the case of a discrete state space, such a DP suffers from the curse of dimensionality when demand is arbitrarily correlated. We elaborate more on this issue in Appendix \ref{apx:DP:hardness-input-model} by providing additional details on the computational model (as well as the input model) that we consider. In Appendix \ref{apx:DP:independent}, we focus on the special case of independent demands, where the curse of dimensionality does not occur; we show that by appropriate discretization, we can obtain a fully polynomial-time {(near-optimal)} approximation scheme (FPTAS) for this DP. 
Finally, in Appendix \ref{apx:DPexample}, we present examples that illustrate the drawbacks of the DP approach beyond computational challenges. By convention, in the remaining of this section we assume that demands and supply are normalized such that the total supply $s = 1$.
\label{apx:DP:details}
\subsection{Exponential-sized DP for the General Case}
\label{apx:DP:general}

 We start describing our dynamic programming by its state space. Upon the discrete arrival time of agent $i$, the state of our dynamic programming captures the index of the current agent $i \in [n]$, the sample path of all previous agents' realized demands $\vec{d}_{[1:i-1]}$, the remaining supply $s_i$, and the minimum FR $f_i$ up to (not including) the arrival of agent $i$. Note that it is necessary to include the entire sample path of previous demands in our state, as the demand is arbitrarily correlated. Thus, the conditional distribution of the vector of future demands, and hence the performance of the optimum online policy, depends on the realized sample path so far.
 As we discuss in more detail in Appendix \ref{apx:DP:hardness-input-model}, the size of such a state space can be exponential in $n$ (even after discretization), as there might be exponentially many sample paths $\vec{d}$ that have non-zero probability of happening.

Now, we define $\mathcal{V}^{\DemandJoinDist}(i,\vec{d}_{[1:i-1]},s_i,f_i)$ to be the expected minimum FR of the optimum online policy starting from the state $\left(i,\vec{d}_{[1:i-1]},s_i,f_i\right)$. The goal of the dynamic programming is filling the table $\mathcal{V}^{\DemandJoinDist}\left(\cdot\right)$ for all possible state entries. 
For ease of exposition, we add a dummy agent $i=n+1$ at the end of the time horizon, who deterministically has $0$ demand. As the base of the dynamic programming, we set $\mathcal{V}^{\DemandJoinDist}(n+1,\vec{d}_{[1:n]},s_{n+1},f_{n+1})=f_{n+1}$ for all possible values of $\vec{d}_{[1:n]}\in\textrm{support}(\DemandJoinDist)$, $s_{n+1}\in[0,1]$ and $f_{n+1}\in[0,1]$. Now the DP table can be updated in a backward fashion (i.e., from time $i=n$ to $i=1$) following this simple Bellman equation:
\begin{equation}
\label{eq:bellman} 
    \mathcal{V}^{\DemandJoinDist}(i,\vec{d}_{[1:i-1]},s_i,f_i)=\expect[]{\underset{x_i\in\left[0, \min(s_i,d_i)\right]}{\max}\left\{ \mathcal{V}^{\DemandJoinDist}\left(i+1,\vec{d}_{[1:i]},s_i-x_i,\min(f_i,\frac{x_i}{d_i})\right)\right\}\Big|\vec{d}_{[1:i-1]}}
\end{equation}
Given the filled DP table for all possible values of the state space and given a state $\vec{d}_{[1:i-1]}$, $s_i$ and $f_i$ ahead of the arrival of agent $i$, upon realization of $d_i$ the optimum online policy picks the allocation decision $x^*_i(\vec{d}_{[1:i]},s_i,f_i)$ such that:
$$
x^*_i(\vec{d}_{[1:i]},s_i,f_i)=\underset{x_i\in\left[0, \min(s_i,d_i)\right]}{\argmax}\left\{ \mathcal{V}^{\DemandJoinDist}\left(i+1,\vec{d}_{[1:i]},s_i-x_i,\min(f_i,\frac{x_i}{d_i})\right)\right\}
$$
\subsection{Discussion on the Input Model and Computing DP}
\label{apx:DP:hardness-input-model} 

Having described the Bellman equation for the DP, we next discuss its computation. For the sake of this discussion, let us assume that all demands are discrete and can take $\Delta$ different values. Even in this special case, the support of the joint distribution $\DemandJoinDist$ can be $\Delta^n$ because we do not impose any restrictions on the structure of the joint distribution. 
Given that the DP table scales with the cardinality of the support of  $\DemandJoinDist$, it cannot be computed in polynomial time {(except for special cases, e.g., when the demands are mutually independent, as we discuss in Appendix \ref{apx:DP:independent})}. Said differently, in general even if we fix $i$, $s_i$, and $f_i$, the function $\mathcal{V}^{\DemandJoinDist}(i,\vec{d}_{[1:i-1]},s_i,f_i)$ can take $\Delta^{i-1}$ different values, corresponding to the number of possible realized
demand histories $\vec{d}_{[1:i-1]}$.\footnote{\revcolor{We would like to highlight that if the support of $\DemandJoinDist$ has size equal to $L\in\mathbb{N}$, i.e., if this distribution only assigns non-zero probabilities to $L$ different demand sequences, then the size of the DP table will be polynomial in $L$. Hence, under a naive unary  computational model where any algorithm should first read the input distribution as a $L$-dimensional vector in $[0,1]^n$, DP's running time can be thought of as a polynomial in its input size $L$. However, in principle, $L$ can be as large  as $\Delta^n$ (hence exponentially large in $n$), which makes reading the input distribution computationally inefficient and highly non-practical. As we explain in the next paragraph, we consider a different computational model to circumvent the issue of reading the input distribution. (We emphasize that this computational model does not reduce the exponential size of the DP table's state space.)}}

At the same time, computing first-order moments for such distributions can also be challenging: we cannot read the distribution as an input unless $\DemandJoinDist$ has a succinct representation. However, we argue that the aforementioned challenge 
can be circumvented when computing our \PolicyNameAb\ policy, but not when computing the DP. 
In particular, \revcolor{instead of considering a computational model in which any algorithm first reads the input joint distribution $\DemandJoinDist$ explicitly, we consider a more practical computational model with \emph{sample access} to this distribution, i.e.,} we assume that we have access to an oracle that can generate sample paths from distribution $\DemandJoinDist$. More precisely, given $\vec{d}_{[1:i-1]}$, the oracle can generate independent samples $\vec{d}_{[i:n]}$ according to the joint distribution conditioned on $\vec{d}_{[1:i-1]}$. Equipped with such an oracle, 
 at any time during the run of our \PolicyNameAb\ policy 
we can accurately estimate $\expect[\vec{\Demand}\sim\DemandJoinDist]{\ \sum_{j \in [i: \NumAgents]}\Demandi{j} \ \Big| \  \DemandVec_{[1:i-1]}}$ 
to make the next allocation decision. 
However, computing the optimum online policy according to  \eqref{eq:bellman} would still require filling an exponential-sized DP table, even when assuming access to such an oracle.
This further highlights the computational challenges of finding the DP solution. 

Finally, we remark that having access to such an oracle is practical and well-motivated. In fact, demand forecasting models in contexts such as a pandemic rely on parametric models for which the set of parameters are drawn from appropriate distributions (or uncertainty sets). As such, even though such distributions do not have succinct representations, it is easy to simulate them. We provide one such example in our numerical study in  \Cref{sec:numerics}, where we present a simple SEIR model to forecast the number of infections for the COVID-19 pandemic. As is common in the epidemiology literature, we assume that some of the parameters of the model are uncertain and drawn from given distributions. The resulting continuous joint demand distribution does not have a compact representation and induces demands that are highly correlated; however, it is easy to generate independent sample paths by drawing parameters and then following the dynamics of the SEIR model. 

\subsection{Special Case of Independent Demands: a Fully Polynomial-time Approximation Scheme}
\label{apx:DP:independent} 
We start this subsection by considering a special case of the dynamic programming in Appendix \ref{apx:DP:general} when demands are independent. However, we still allow for a continuous state space (i.e., the demand, the remaining supply, the minimum FR, and the allocation decision at each period $i$ can be any non-negative real numbers). Later we show how to discretize this DP to obtain a fully polynomial-time (near-optimal) approximation scheme. 

\vspace{2mm}
\smallskip
\noindent\textbf{Continuous DP:} When the demands $d_1,d_2,\ldots,d_n$ are independent, the dynamic programming for the optimum online policy (discussed in Appendix \ref{apx:DP:general}) \emph{does not} need to consider the sample path of all previous agents' realized demands $\vec{d}_{[1:i-1]}$ as part of the state, simply because future demand does not depend on this demand history. As a result, we can define $\mathcal{V}_{\textsc{ind}}^{\DemandJoinDist}(i,s_i,f_i)$ to be the expected minimum FR of the optimum online policy starting from time $i$, when it is initialized with a remaining supply of $s_i$ and a current minimum FR of $f_i$. For this special case, the Bellman update equation in \cref{eq:bellman} can be simplified as follows: 
\begin{equation}
    \mathcal{V}_{\textsc{ind}}^{\DemandJoinDist}(i,s_i,f_i)=\expect[]{\underset{x_i\in\left[0, \min(s_i,d_i)\right]}{\max}\left\{ \mathcal{V}_{\textsc{ind}}^{\DemandJoinDist}\left(i+1,s_i-x_i,\min(f_i,\frac{x_i}{d_i})\right)\right\}}
\end{equation}
Similar to the general setting, we add a dummy agent $i=n+1$ at the end of the time horizon, who deterministically has $0$ demand, and as the base of the dynamic programming, we set $\mathcal{V}_{\textsc{ind}}^{\DemandJoinDist}(n+1,s_{n+1},f_{n+1})=f_{n+1}$ for all possible values of $s_{n+1}, f_{n+1}\in[0,1]$. Given the filled DP table for all possible values of the state space and given a state $s_i$ and $f_i$ ahead of the arrival of agent $i$, upon realization of $d_i$ the optimum online policy picks the allocation decision $x^*_i(d_i, s_i, f_i)$ such that:
$$
x^*_i(d_i, s_i, f_i)=\underset{x_i\in\left[0, \min(s_i,d_i)\right]}{\argmax}\left\{ \mathcal{V}_{\textsc{ind}}^{\DemandJoinDist}\left(i+1,s_i-x_i,\min(f_i,\frac{x_i}{d_i})\right)\right\}
$$

Note that the above DP does not suffer from the curse of dimensionality. 
However, its state space is still continuous. In the following, we show that through proper discretization of the state space and the space of allocation decisions, we can develop a fully polynomial-time (near-optimal) approximation scheme for the continuous DP (denoted $\textsc{DP}^{\texttt{CONT}}$) by solving the discretized DP (denoted $\textsc{DP}^{\texttt{DISC}}$).

\vspace{2mm}
\smallskip
\noindent\textbf{Discretizing the DP:} Suppose $\DemandJoinDist=\mathcal{F}_1\times \mathcal{F}_2\times\ldots\times\mathcal{F}_n$. We reiterate our assumption (without loss of generality) that supply and demands are normalized such that $s=1$, and we further assume that the (normalized) demands are bounded by $d_H$. Fixing a parameter $\epsilon>0$ such that $\epsilon^{-1} \in \mathbb{N}$, 
we then create an $\epsilon$-grid bounded between $0$ and $\lceil\max\{1, d_H\}\rceil$, i.e., $G^{\epsilon}
\triangleq\{k\epsilon: k=0,1,\ldots, \lceil \frac{\max\{1,d_H\}}{\epsilon}\rceil\}$, which we will use to discretize our state space.

Given $d_i\sim\mathcal{F}_i$, we define the random variable $\tilde{d}_i$ to be $d_i$ rounded up to the closest multiple of $\epsilon$, i.e., $\tilde{d}_i$ is the smallest member of the set $G^{\epsilon}$ such that $\tilde{d}_i \geq d_i$. Let $\tilde{\mathcal{F}}_i$ denote the CDF of $\tilde{d}_i$. We also discretize the space for the allowable values of remaining supply using the same $\epsilon$-grid $G^{\epsilon}$. Based on this restriction along with the assumption that $\epsilon$ evenly divides $1$, the allowable allocations are likewise contained within the set $G^{\epsilon}$. (Of course, the allowable allocation remains bounded by the current remaining supply and current demand.) Furthermore, the set of possible minimum FRs is bounded by $|G^{\epsilon}|^2$, since an FR is determined by the allocated supply and the realized demand. Thus, the state space of $\textsc{DP}^{\texttt{DISC}}$ is at most $n|G^{\epsilon}|^3$. For each possible realized demand (of which there are at most $|G^{\epsilon}|$), determining the optimal allocation requires a search over at most $|G^{\epsilon}|$ feasible allocations. As a consequence, filling the entire table requires on the order of $n|G^{\epsilon}|^5$ operations.

\vspace{2mm}
\smallskip
\noindent\textbf{Providing a near-optimal approximation:}
To implement $\textsc{DP}^{\texttt{DISC}}$, whenever a demand $d_i$ is realized, we round it up to $\tilde{d}_i$. We use the solution of $\textsc{DP}^{\texttt{DISC}}$, denoted $x^{\texttt{DISC}}_i$, to guide our allocation decision. 
We also update the state of  $\textsc{DP}^{\texttt{DISC}}$, in particular the remaining supply, based on $x^{\texttt{DISC}}_i$.
However, to ensure consistency with the discretization scheme, we only allocate the portion of $x^{\texttt{DISC}}_i$ necessary to ensure that the actual minimum FR precisely matches the discretized minimum FR, and we discard the remaining supply.\footnote{\revcolor{It is easy to verify that discarding supply cannot improve performance.}} To be precise, we allocate $\bar{x}^{\texttt{DISC}}_i \triangleq x^{\texttt{DISC}}_i \frac{d_i}{\tilde{d}_i}$. This convention ensures that the expected minimum FR of $\textsc{DP}^{\texttt{DISC}}$ on the actual (continuous) instance is identical to its expected minimum FR on the discretized instance. 
We now show that this expected minimum FR provides a near-optimal approximation of the expected minimum FR obtained by $\textsc{DP}^{\texttt{CONT}}$.

\begin{proposition}
\label{prop:fptas}
For any $\epsilon>0$ such that $\epsilon^{-1} \in \mathbb{N}$, the expected minimum FR obtained by $\textsc{DP}^{\texttt{DISC}}$ is at least $1-2n\epsilon$ times the expected minimum FR obtained by $\textsc{DP}^{\texttt{CONT}}$.
\end{proposition}

\begin{proof}{Proof of \Cref{prop:fptas}:}
The proof relies on two key claims. For ease of presentation, we use the name of each DP to also denote the expected minimum FR it obtains. Furthermore, we let this expected minimum FR depend on the initial supply (which we remind is normalized to $1$). First, we will establish a connection between the expected minimum FR obtained by the discretized DP, i.e., $\textsc{DP}^{\texttt{DISC}}(1)$, and the expected minimum FR obtained by the continuous DP \emph{when initialized with less supply}.
\begin{claim}
\label{clm:fptas1}
For any $\epsilon>0$ such that $\epsilon^{-1} \in \mathbb{N}$, 
    $\textsc{DP}^{\texttt{DISC}}(1) \geq \textsc{DP}^{\texttt{CONT}}(1-2n\epsilon)$
\end{claim}
\begin{proof}{Proof of Claim \ref{clm:fptas1}:}
We proceed by finding a feasible solution in $\textsc{DP}^{\texttt{DISC}}(1)$ that always obtains a minimum FR at least as large as $\textsc{DP}^{\texttt{CONT}}(1-2n\epsilon)$. {Fix a demand sample path $\{\Demandi{i}\}_{i\in[1:n]}$ along with the corresponding rounded up demands  $\{\tilde{\Demandi{i}}\}_{i\in[1:n]}$.} Further, let $x^{\texttt{CONT}}_i\in[0,\Demandi{i}]$ denote the allocation decisions of $\textsc{DP}^{\texttt{CONT}}(1-2n\epsilon)$. Consider the solution $\hat{x}^{\texttt{DISC}}_i = \epsilon \lceil\frac{x^{\texttt{CONT}}_i}{\epsilon}+1\rceil$.
We will show that allocating $\min\{\tilde{\Demandi{i}},\hat{x}^{\texttt{DISC}}_i,\}$ at each time $i$ is a feasible allocation in the discretized DP, i.e., it never allocates more than initial supply.
We first note that $\hat{x}^{\texttt{DISC}}_i \leq \epsilon (\frac{x^{\texttt{CONT}}_i}{\epsilon}+2) = x^{\texttt{CONT}}_i + 2\epsilon$. Furthermore, in order to ensure feasibility, the total allocation $\sum_{i \in [\NumAgents]} x^{\texttt{CONT}}_i$ can never exceed $1-2n\epsilon$ regardless of the sample path. As a consequence of these two observations, we must always have that $\sum_{i \in [\NumAgents]} \min\{{\Demandi{i}},\hat{x}^{\texttt{DISC}}_i\}\leq \sum_{i \in [\NumAgents]} \hat{x}^{\texttt{DISC}}_i \leq 1$, which implies that our proposed solution  is feasible in $\textsc{DP}^{\texttt{DISC}}(1)$.

We now compare the FRs under the two policies. Recall our allocation convention for the discretized DP, which ensures that agent $i$ attains an FR of $\min\left\{1,\frac{\hat{x}^{\texttt{DISC}}_i}{\tilde{d}_i} \right\}$. To bound the second term in that minimum, we note: $$\frac{\hat{x}^{\texttt{DISC}}_i}{\tilde{d}_i} \quad \geq \quad \frac{\hat{x}^{\texttt{CONT}}_i + \epsilon}{d_i + \epsilon} \quad \geq \quad \frac{\hat{x}^{\texttt{CONT}}_i}{d_i}.$$
The first inequality comes from the facts that under $\hat{x}^{\texttt{DISC}}_i$, (i) we add at least an amount $\epsilon$ to the allocation decision of the continuous DP, and (ii) we round realized demand up by at most $\epsilon$. Thus, we are lower-bounding the numerator and upper-bounding the denominator. The second inequality comes from the fact that for any $a,b, c$ such that $b \geq a \geq 0$ and $c \geq 0$, $\frac{a+c}{b+c} \geq \frac{a}{b}$. As a result, the FR for agent $i$ under the proposed feasible policy in $\textsc{DP}^{\texttt{DISC}}(1)$ is lower-bounded by the FR for agent $i$ in $\textsc{DP}^{\texttt{CONT}}(1-2n\epsilon)$, simply because $\hat{x}^{\texttt{CONT}}_i\leq d_i$ and hence
$$
\min\left\{1,\frac{\hat{x}^{\texttt{DISC}}_i}{\tilde{d}_i} \right\}\quad \geq \quad \frac{\hat{x}^{\texttt{CONT}}_i}{d_i}.
$$

This comparison holds for the FR of each agent, {which means that the minimum FR under the proposed feasible solution in $\textsc{DP}^{\texttt{DISC}}(1)$ is at least the minimum FR in $\textsc{DP}^{\texttt{CONT}}(1-2n\epsilon)$ under any realized sample path $\{\Demandi{i}\}_{i\in[1:n]}$. This is a sufficient condition to prove Claim~\ref{clm:fptas1}, as the same comparison holds after taking an expectation over the stochasticity in the demand sequence.} 
\hfill \halmos
\end{proof}
Next, we establish a lower bound on the value of the continuous DP when initialized with less supply. 
\begin{claim}
\label{clm:fptas2}
For any $\delta > 0$, $\textsc{DP}^{\texttt{CONT}}(1-\delta) \geq (1-\delta)\textsc{DP}^{\texttt{CONT}}(1)$
\end{claim}
\begin{proof}{Proof of Claim \ref{clm:fptas2}:}
We proceed by finding a feasible allocation policy in $\textsc{DP}^{\texttt{CONT}}(1-\delta)$ that attains a value of at least $(1-\delta)\textsc{DP}^{\texttt{CONT}}(1)$. Let $x^{\texttt{CONT}}_i$ denote the allocation decisions of $\textsc{DP}^{\texttt{CONT}}(1)$. Now consider an alternative policy which follows a different allocation decision, 
$\hat{x}^{\texttt{CONT}}_i = (1-\delta)x^{\texttt{CONT}}_i$ at any given state and realized demand. By feasibility, the total allocation under $\textsc{DP}^{\texttt{CONT}}(1)$ can never exceed $1$. Since each allocation decision under this alternate policy is reduced by a factor $1-\delta$, the total allocation under this alternative policy will never exceed $1-\delta$. Thus, this alternative policy is feasible for $\textsc{DP}^{\texttt{CONT}}(1-\delta)$.

Furthermore, the minimum FR under this alternative policy will always be within a factor $1-\delta$ of the minimum FR under $\textsc{DP}^{\texttt{CONT}}(1)$, since the demands are the same and the allocation decisions are off by at most a factor of $1-\delta$. As a result, there is a feasible solution to $\textsc{DP}^{\texttt{CONT}}(1-\delta)$ that obtains an expected minimum FR of at least $(1-\delta)\textsc{DP}^{\texttt{CONT}}(1)$. This is a sufficient condition to prove Claim \ref{clm:fptas2}. \hfill \halmos
\end{proof}
Together, Claims \ref{clm:fptas1} and \ref{clm:fptas2} establish that $\textsc{DP}^{\texttt{DISC}}(1) \geq (1-2n\epsilon)\textsc{DP}^{\texttt{CONT}}(1)$, which completes the proof of \Cref{prop:fptas}. \hfill \halmos
\end{proof}
}

\revcolor{We end this subsection by noting that if we set $\epsilon=\frac{\epsilon'}{2n}$ when designing our discrete grid, then  $\textsc{DP}^{\texttt{DISC}}(1)\geq (1-\epsilon')\textsc{DP}^{\texttt{CONT}}(1)$. Moreover, filling the DP table for $\textsc{DP}^{\texttt{DISC}}(1)$ requires on the order of $n|G^{\frac{\epsilon'}{2n}}|^5$ operations, which results in a running time of $\mathcal{O}\left(n^6\left(\frac{ \max\{1,d_{H}\}}{\epsilon'}\right)^5\right)$.}\footnote{\revcolor{We would like to emphasize that although the running time of this FPTAS is polynomial in $n$ (for constant $\epsilon'$), it is much larger than the running time of the PPA policy, which is linear in $n$. Note also that the running time of the PPA policy remains linear in $n$ even when the demands are arbitrarily correlated.}}

\revcolor{
\subsection{Illustrating Limitations of Dynamic Programming (\texorpdfstring{\Cref{subsec:expostupperbound}}{})}
\label{apx:DPexample}
In Appendix \ref{apx:DP:hardness-input-model}, we illustrated the challenges of computing the DP solution when faced with an arbitrarily correlated demand sequence. However, even in the special case of independent demands where the DP solution can be efficiently approximated (see Appendix \ref{apx:DP:independent}), the DP solution suffers from a number of additional drawbacks. We highlighted these limitations in Remark \ref{rem:online} of \Cref{subsec:expostupperbound}, and we summarize these points in \Cref{table:DP:Sucks}.   

While we have already discussed the first two rows of \Cref{table:DP:Sucks}, in the following, we illustrate the third limitation using  \Cref{ex:DPBad}. Finally, note that because the DP allocation decisions is based on solving a complex stochastic optimization problem, it inherently lacks transparency and interpretability. However, to highlight lack of interpretability even further, we use  \Cref{ex:DPBad} to show how a small change in the demand of one agent can drastically change its DP allocation decision. Unlike the DP solution, the allocation decisions of the \PolicyNameAb\ policy are minimally impacted by small changes in demand.\footnote{\revcolor{Here we assume that the small change in demand does not impact the future demand distribution.}}

\begin{table}[t]

    \caption{\revcolor{Optimum Online (DP) vs. \PolicyNameAb}}
    \centering
    \footnotesize
    \revcolor{\begin{tabular}{| l | c | c |}
   \hline 
    \multicolumn{1}{|c|}{\bf Desired Property} & \multicolumn{1}{|c|}{\bf \ Optmium Online (DP) \  }& \multicolumn{1}{|c|}{\bf \ Our PPA Policy \  } \\ \hline
    Computationally  Efficient  & \emph{No} & \emph{Yes} \\ \hline
    Based on Minimal Distributional Knowledge & \emph{No} & \emph{Yes} \\ \hline
    Best Ex-Ante Fairness Guarantee & \emph{No} & \emph{Yes} \\ \hline
    Transparent/Interpretable Decisions &\emph{No} & \emph{Yes} \\ \hline
    \end{tabular}}
    \label{table:DP:Sucks}
\end{table}

\begin{example}
\label{ex:DPBad}
Consider an instance with two agents ($\NumAgents = 2$) 
with independent demand
where the total expected demand is almost twice the amount of supply ($\ExpDemand = 2+\epsilon$ for small $\epsilon>0$). There are two possible demand sequences that occur with equal probability, either $\DemandVec = (\frac{4}{3}+\epsilon, \frac{4}{3})$ or $\DemandVec = (\frac{4}{3}+\epsilon, 0)$.
\end{example}
 In this example, the first agent has deterministic demand $\Demandi{1} = \frac{4}{3}+\epsilon$. The second agent either has no demand (with probability $0.5$) or has demand $\Demandi{2} = \frac{4}{3}$ (also with probability $0.5$). Suppose the first allocation is deterministically given by $x_1$. (We note that neither the optimum online policy nor the \PolicyNameAb\ policy are randomized policies, and thus the first agent will receive a deterministic allocation under each policy.) Then, the expected minimum FR is equal to $$0.5\left(\frac{3x_1}{4+3\epsilon}\right) + 0.5 \left(\min\left\{\frac{3x_1}{4+3\epsilon}, \frac{3(1-x_1)}{4}\right\}\right).$$ It is easy to verify numerically that this is maximized at $x_1^*(d_1) = \frac{4+3\epsilon}{8+3\epsilon}$, which implies that if the second demand is realized, $x_2^*(d_2) = \frac{4}{8+3\epsilon}$. This achieves an expected minimum FR of $\frac{3}{8+3\epsilon}$. 
 
 In contrast, the PPA policy allocates $x_1^{\PolicyNameAb}(d_1) = \frac{4/3 + \epsilon}{4/3 + \epsilon + (0.5)4/3} = \frac{4+3\epsilon}{6+3\epsilon}$, which implies that if the second demand is realized, $x_2^{\PolicyNameAb}(d_2) = \frac{2}{6+3\epsilon}$. Based on our formula above, this achieves an expected minimum FR of $0.5\left(\frac{3}{6+3\epsilon}\right) + 0.5\left(\frac{1.5}{6+3\epsilon}\right) = \frac{3}{8+4\epsilon}$, 
 which is quite close to the expected minimum FR achieved by the optimal online policy.

We now make the following two observations based on this example:

\begin{enumerate}[label=(\roman*)]
\item {\it The DP solution achieves sub-optimal ex-ante fairness.} The FR of the first agent under the optimum online solution is deterministically equal to $\frac{x_1^*(d_1)}{d_1} = \frac{3}{8+3\epsilon}$, while the expected FR of the second agent is equal to $\expect{\frac{x_2^*(d_2)}{d_2}} = 0.5(1) + 0.5\left(\frac{3}{8+3\epsilon}\right)$, which is larger. Thus, the ex-ante minimum FR  that results from the DP solution is $\frac{3}{8+3\epsilon}$. 

In contrast, under the \PolicyNameAb\ policy, the FR of the first agent is deterministically equal to $\frac{x_1^{\PolicyNameAb}(d_1)}{d_1}=\frac{1}{2+\epsilon}$, while the expected FR of the second agent is equal to $\expect{\frac{x_2^{\PolicyNameAb}(d_2)}{d_2}} = 0.5(1) + 0.5\left(\frac{1}{4+2\epsilon}\right)$, which is larger. Thus, the ex-ante minimum FR that results from the \PolicyNameAb\ policy is $\frac{1}{2+\epsilon}$. We highlight that not only is the minimum ex-ante minimum FR of the DP solution less than that of our \PolicyNameAb\ policy, the corresponding ex-ante fairness for the DP solution is less than the \emph{guarantee} on ex-ante fairness provided our policy (as given in \Cref{thm:exante}).

\item {\it The DP solution lacks interpretability and is sensitive to small perturbations.} To see this, we show that the above DP solution can vary significantly if the demand distribution is slightly different. Suppose that we perturb this instance by decreasing the first agent's demand by $2\epsilon$. In this case, the demand sequence is equally likely to be $\DemandVec = (\frac{4}{3}-\epsilon, \frac{4}{3})$ or $\DemandVec = (\frac{4}{3}-\epsilon, 0)$, and the ex-post minimum FR for an initial allocation decision $x_1$ is given by $$0.5\left(\frac{3x_1}{4-3\epsilon}\right) + 0.5 \left(\min\left\{\frac{3x_1}{4-3\epsilon}, \frac{3(1-x_1)}{4}\right\}\right).$$
Despite the small change in the demand sequence, the DP solution changes dramatically, as this expression is maximized at $x_1^*(d_1) = 1$. 
This new expected minimum FR of $\frac{3}{8-6\epsilon}$ is essentially unchanged from the original setting, but the first agent (deterministically) receives almost twice as much supply as before.
Not only does this demonstrate that the DP solution is highly sensitive, but it also highlights that the DP solution can suffer from a lack of interpretability: the first agent receives more supply in this case, even though the only change was a deterministic and negligible decrease in that agent's demand. In sharp contrast, the allocation decision of our \PolicyNameAb\ policy barely changes, as the first agent deterministically receives $x_1^{\PolicyNameAb}(d_1) = \frac{4/3 - \epsilon}{4/3 - \epsilon + (0.5)4/3} = \frac{4-3\epsilon}{6-3\epsilon}$.
\end{enumerate}
}

\section{\texorpdfstring{Simple Backward Dynamic Programming for Optimum Offline (\Cref{subsec:ourpolicy})}{}}
\label{apx:offline:DP}
Suppose the planner has access to all the demand realizations $\DemandVec$, and
let $f_i$ be the minimum FR of the policy at the end of period $i-1$, i.e.,  $f_1=1,~f_{i}=\min\{f_{i-1},\frac{\Alloci{i-1}}{\Demandi{i-1}}\}$ for $i\in[2:n+1]$. We will show via backward induction that for any remaining supply $\Supplyi{i}$ and minimum FR $f_i$, the maximum-achievable minimum FR is $\min\{f_i, \frac{s_i}{\sum_{j \in [i:n]} d_j}\}$, which is achieved by a policy $\Alloc^*_i=\min\left\{\Demandi{i},\Supplyi{i}\frac{\Demand_i}{\sum_{j \in [i:n]}\Demandi{j}}\right\}$.

Clearly, this is true when $i=n$, as the optimal policy is to allocate as much supply as possible, i.e. $\Alloci{n}^* = \min\{\Demandi{n}, \Supplyi{n}\}$. This policy achieves a minimum FR of $$f_{n+1} = \min\left\{f_{n}, \frac{\Alloci{n}}{\Demandi{n}}\right\} = \min\left\{f_{n}, \frac{\min\{\Demandi{n}, \Supplyi{n}\}}{\Demandi{n}}\right\} = \min\left\{f_{n}, \frac{\Supplyi{n}}{\Demandi{n}}\right\}.$$ 

We now assume this is true for $i > k$. In that case, given an allocation to agent $k$ of $\Alloci{k}$, the minimum FR at the end of period $k$ is given by $f_{k+1} = \min\left\{f_k, \frac{\Alloc{k}}{\Demandi{k}}\right\}$ and the remaining supply is $\Supplyi{k+1} = \Supplyi{k} - \Alloci{k}$. Based on our inductive hypothesis, the maximum-achievable minimum FR is thus $$\min\left\{f_{k+1}, \frac{\Supplyi{k+1}}{\sum_{j \in [k+1:n]} d_j}\right\} = \min\left\{f_k, \frac{\Alloc{k}}{\Demandi{k}}, \frac{\Supplyi{k} - \Alloci{k}}{\sum_{j \in [k+1:n]} d_j}\right\}.$$
This is maximized when the second and third terms are equal, which occurs when $\Alloci{k} = \Supplyi{k}\frac{\Demand_k}{\sum_{j \in [k:\NumAgents]}\Demandi{j}}$. If this allocation is infeasible $\left(\text{i.e., if } \Supplyi{k}\frac{\Demand_k}{\sum_{j \in [k:\NumAgents]}\Demandi{j}} \geq \Demandi{k}\right)$, then an allocation of $\Alloci{k} = \Demandi{k}$ is optimal because this allocation ensures that $f_k$ must be the minimum of the three terms. This completes the proof by backward induction that the maximum-achievable minimum FR is $\min\left\{f_i, \frac{s_i}{\sum_{j \in [i:n]} d_j}\right\}$, which is achieved by a policy $\Alloc^*_i=\min\left\{\Demandi{i},\Supplyi{i}\frac{\Demand_i}{\sum_{j \in [i:n]}\Demandi{j}}\right\}$.

\section{\texorpdfstring{Analysis of Non-adaptive Fixed-allocation Policies (\Cref{subsec:AdaptivityGap})}{}}
\label{apx:fixedallocation}
In this appendix section, we consider another class of non-adaptive policies which we call \fixedalloc\ policies. A \fixedalloc\ policy is one which pre-determines an allocation $\Alloc_i$ for each agent $i \in [\NumAgents]$.\footnote{As explained in Footnote \ref{foot:detpolicies}, we focus on deterministic policies without loss of generality.} The optimal \fixedalloc\ policy is the policy which, given a joint demand distribution $\DemandJoinDist$, pre-determines a vector of allocations $\vec{\Alloc} = (\Alloci{1}, \Alloci{2}, \dots, \Alloci{\NumAgents})$ which maximizes $\expect[\Vec{\Demand}\sim\DemandJoinDist]{\min_{i\in[\NumAgents]}\left\{\frac{\Alloci{i}}{\Demandi{i}}\right\}}$.

\begin{proposition}[Ex-post Fairness Guarantee of the Optimal Fixed-allocation Policy]
\label{prop:fixedalloc}
Given a fixed number of agents $\NumAgents\in\mathbb{N}$ and supply scarcity $\ExpDemand\in \NonNegReals$, the optimal fixed-allocation policy achieves an ex-post fairness guarantee of
\begin{equation}
   \BoundFixedAll = \max\{1, \ExpDemand\}\begin{cases}
   1-\frac{\NumAgents\ExpDemand}{4}, &\ExpDemand\NumAgents \in [0, 2) \\
   \frac{1}{\NumAgents \ExpDemand}, &\ExpDemand\NumAgents \in [2, +\infty)
   \end{cases}.
\end{equation}
\end{proposition}
We remark that the ex-post fairness guarantee $\BoundFixedAll$ tends to $0$ as the number of agents $\NumAgents$ gets large. This is in stark contrast to the guarantees provided by the \PolicyNameAb\ policy and the optimal \fixedthresh\ policy, which are lower-bounded by a constant regardless of the number of agents.
\subsection{Proof of Proposition \ref{prop:fixedalloc}}
We prove this proposition by first showing that there exists a distribution where no \fixedalloc\ policy achieves ex-post fairness greater than $\BoundFixedAll$, which thus serves as an upper bound on the ex-post fairness guarantee of the optimal \fixedalloc\ policy. We then show that there exists a \fixedalloc\ policy that achieves ex-post fairness of at least $\BoundFixedAll$ for any demand distribution, which means the bound $\BoundFixedAll$ is tight.

{\bf Upper bound:} We prove the hardness result by considering two separate cases corresponding to $\NumAgents \ExpDemand < 2$ and $\NumAgents \ExpDemand \geq 2$. For each case, we provide an instance of the problem under which any \fixedalloc\ policy obtains ex-post fairness no larger than $\BoundFixedAll$.

\begin{enumerate}[label=(\roman*)]
    \item If $\NumAgents\ExpDemand < 2$, consider a joint demand distribution such that with probability $1-\tfrac{\NumAgents\ExpDemand}{2}$ there is no demand, and with probability $\frac{\NumAgents\ExpDemand}{2}$ one agent chosen uniformly at random has demand $\frac{2}{\NumAgents}$. In this case, with probability $1-\tfrac{\NumAgents\ExpDemand}{2}$, the minimum FR is $1$, and with probability $\frac{\NumAgents\ExpDemand}{2}$, the minimum FR is equal to the allocation of a randomly selected agent (which is at most $1/\NumAgents$) divided by the total demand $2/\NumAgents$. Therefore, the minimum expected FR for this instance is upper-bounded by $1-\tfrac{\NumAgents\ExpDemand}{4}$.
    \item If $\NumAgents\ExpDemand \geq 2$, consider a joint demand distribution where one agent chosen uniformly at random has demand equal to the expected total demand $\ExpDemand$. In this instance, the minimum expected FR is upper-bounded by the allocation of a randomly selected agent (which is at most $1/\NumAgents$) divided by the total demand $\ExpDemand$.
\end{enumerate}

Taken together, these instances provide an upper bound on the expected minimum FR one can hope to achieve with a \fixedalloc\ policy. We then scale each instance by our normalization factor, namely $\DetGuar = \min\{1, 1/\ExpDemand\}$, which provides an upper bound of $\BoundFixedAll$ on the ex-post fairness guarantee (see \Cref{def:polperf}) of any \fixedalloc\ policy. 

{\bf Lower bound:} Consider a policy which allocates an equal amount of supply to each agent, i.e., $\Alloci{i} = \frac{1}{\NumAgents}$ for all $i \in [\NumAgents]$. In that case, the minimum FR is lower-bounded by $\min\left\{1, \frac{1}{n\sum_{ i \in [\NumAgents]}\Demand_i}\right\}$. Further,
\begin{equation*}
    \min\left\{1, \frac{1}{n\sum_{ i \in [\NumAgents]}\Demand_i}\right\} \geq \begin{cases}
   1-\frac{\NumAgents\sum_{ i \in [\NumAgents]}\Demand_i}{4}, &\NumAgents\sum_{ i \in [\NumAgents]}\Demand_i \in [0, 2) \\
   \frac{1}{\NumAgents \sum_{ i \in [\NumAgents]}\Demand_i}, &\NumAgents\sum_{ i \in [\NumAgents]}\Demand_i \in [2, +\infty)
   \end{cases}
\end{equation*}
We note that the right hand side of the above inequality is convex in $\sum_{i \in [\NumAgents]}\Demand_i$. Therefore, using Jensen's inequality, the expected minimum FR must be at least
\begin{equation*}
\begin{cases}
   1-\frac{\NumAgents\ExpDemand}{4}, &\NumAgents\ExpDemand\in [0, 2) \\
   \frac{1}{\NumAgents\ExpDemand}, &\NumAgents\ExpDemand \in [2, +\infty)
   \end{cases}
\end{equation*}
We then scale each this lower bound on the expected minimum FR by our normalization factor, namely $\DetGuar = \min\{1, 1/\ExpDemand\}$. This provides a lower bound on the ex-post fairness guarantee (see \Cref{def:polperf}) that is equal to $\BoundFixedAll$. Thus, we have shown that $\BoundFixedAll$ is a tight bound on the ex-post fairness guarantee of the optimal \fixedalloc\ policy.
\hfill\Halmos

\section{Missing Proofs of \texorpdfstring{\Cref{subsec:AdaptivityGap}}{}}
\subsection{Proofs of Claims in \texorpdfstring{\Cref{lem:nonadaptlowbound}}{} (\texorpdfstring{\Cref{subsec:AdaptivityGap}}{})}
\label{apx:adaptive-v2}

In this subsection, we present the proofs of the three claims which appear in the proof of \Cref{lem:nonadaptlowbound}.

\subsubsection{Proof of Claim \ref{clm:nonadaptive1} (\texorpdfstring{\Cref{subsec:AdaptivityGap}}{})}
\label{apx:nonadaptclaim1}
 Suppose we have a distribution $\invdemanddist$ such that $R^\invdemanddist(q)$ is not flat in $(Q^\invdemanddist(1),1]${, that is, $\exists~q_1,q_2\in (Q^\invdemanddist(1),1]$ such that $ R^\invdemanddist(q_1)\neq R^\invdemanddist(q_2)$.} Let $R^\invdemanddist(q)$ attain its maximum in $[Q^\invdemanddist(1),1]$ at the quantile $\bar{q}$. Now consider a quantile $q'\in (Q^\invdemanddist(1),1]$ so that $R^\invdemanddist(q')<R^\invdemanddist(\bar{q})$. Let $\delta\triangleq q'-Q^\invdemanddist\left(\Tfrvar^\invdemanddist(q')+\epsilon\right)$ be the total probability mass in the TFR interval $(\Tfrvar^\invdemanddist(q'),\Tfrvar^\invdemanddist(q')+\epsilon]$. Then pick a small enough $\epsilon>0$  such that:
\begin{enumerate}[label=(\roman*)]
    \item $\epsilon+\delta+\epsilon\cdot\delta <R^\invdemanddist(\bar q)-R^\invdemanddist(q')$,
    \item $\epsilon\leq 1-\Tfrvar^\invdemanddist(q')$,
    \item $\bar{q}\notin [Q^\invdemanddist\left(\Tfrvar^\invdemanddist(q')+\epsilon\right),q']$.
\end{enumerate}
Now consider a distribution $\bar{\invdemanddist}$ that is generated from $\invdemanddist$ by moving all the $\delta$ probability mass in $(\Tfrvar^\invdemanddist(q'),\Tfrvar^\invdemanddist(q')+\epsilon]$ to the point $\Tfrvar^\invdemanddist(q')+\epsilon$. With this modification in the distribution, the EAFR of every TFR $\tfrvar\in[0,\Tfrvar^\invdemanddist(q')]\cup(\Tfrvar^\invdemanddist(q')+\epsilon,1]$ remains the same. Moreover, the maximum EAFR in the interval $(\Tfrvar^\invdemanddist(q'),\Tfrvar^\invdemanddist(q')+\epsilon]$ is also achieved at $\Tfrvar^\invdemanddist(q')+\epsilon$. {Given this target fill rate, the EAFR of the distribution $\bar{\invdemanddist}$ is equal to}
\begin{align*}
R^{\bar{\invdemanddist}}(\Tfrvar^\invdemanddist(q')+\epsilon) &= (\Tfrvar^\invdemanddist(q')+\epsilon)\left(Q^{\invdemanddist}\left(\Tfrvar^\invdemanddist(q')+\epsilon\right)+\delta\right)\\
&<\Tfrvar^\invdemanddist(q')\cdot Q^{\invdemanddist}\left(\Tfrvar^\invdemanddist(q')+\epsilon\right)+\epsilon+\delta+\epsilon\cdot\delta \\
&<R^\invdemanddist(q')+(R^\invdemanddist(\bar{q})-R^\invdemanddist(q'))=R^\invdemanddist(\bar{q})~.
\end{align*}
Therefore, the maximum EAFR over all possible TFRs in $[0,1]$ is the same for $\invdemanddist$ and $\bar{\invdemanddist}$, i.e., 
$$
\underset{q\in[0,1]:\Tfrvar^\invdemanddist(q)\in[0,1]}{\max} R^\invdemanddist(q)=\underset{q\in[0,1]:\Tfrvar^{\bar{\invdemanddist}}(q)\in[0,1]}{\max} R^{\bar{\invdemanddist}}(q) \quad = \quad R^\invdemanddist(\bar{q}).
$$
However, $\bar{\ExpDemand}=\expect[\invdemand\sim\bar{\invdemanddist}]{\frac{1}{\invdemand}}<\expect[\invdemand\sim{\invdemanddist}]{\frac{1}{\invdemand}}=\ExpDemand$. Now let $\tilde{\invdemanddist}$ be the distribution of the random variable $\left(\bar{\ExpDemand}/\ExpDemand\right)\cdot \invdemand$ where $\invdemand\sim \bar{\invdemanddist}$. We have:
$$
\underset{\tfrvar\in[0,1]}{\max} \tfrvar (1-\tilde{\invdemanddist}(\tfrvar))=\underset{\tfrvar\in[0,1]}{\max} \tfrvar \left(1-\bar{\invdemanddist}(\tfrvar\cdot\ExpDemand/\bar{\ExpDemand})\right) \leq \underset{\tfrvar\in[0,1]}{\max} \tfrvar \left(1-\bar{\invdemanddist}(\tfrvar)\right)=\underset{\tfrvar\in[0,1]}{\max} \tfrvar \left(1-{\invdemanddist}(\tfrvar)\right)~.
$$
Also, $\expect[\invdemand\sim\tilde{\invdemanddist}]{\frac{1}{\invdemand}}=\expect[\invdemand\sim{\invdemanddist}]{\frac{1}{\invdemand}}=\ExpDemand$. Therefore, dropping such a distribution $\invdemanddist$ from the feasible set in the outer optimization of \cref{eq:best-non-adaptive-lemma} does not change the infimum value, which proves Claim \ref{clm:nonadaptive1}.\hfill\Halmos 

\subsubsection{Proof of Claim \ref{clm:nonadaptive2} (\texorpdfstring{\Cref{subsec:AdaptivityGap}}{})}
\label{apx:nonadaptclaim2}
Suppose we have a distribution $\invdemanddist$ that has non-zero total mass in the interval $(1,+\infty)$. Now shift all the probability mass in $(1,+\infty)$ to $+\infty$. Let $\bar{\invdemanddist}$ be the resulting distribution. Note that the maximum EAFR among targets in $[0,1]$ is the same for $\invdemanddist$ and $\bar{\invdemanddist}$, as the EAFR for any target in $[0,1]$ remains the same. However,  $\bar{\ExpDemand}=\expect[\invdemand \sim \bar{\invdemanddist}]{\frac{1}{\invdemand}}<\expect[\invdemand\sim{\invdemanddist}]{\frac{1}{\invdemand}}=\ExpDemand$, as we have moved the probability mass of $\invdemand$ towards larger values (equivalently, the probability mass of demand to lower values). By using the same trick as above in Appendix \ref{apx:nonadaptclaim1}, we conclude that dropping such a distribution $\invdemanddist$ from the feasible set in the outer optimization of \cref{eq:best-non-adaptive-lemma} does not change the infimum value, which proves Claim \ref{clm:nonadaptive2}. \hfill\Halmos

\subsubsection{Proof of Claim \ref{clm:nonadaptive3} (\texorpdfstring{\Cref{subsec:AdaptivityGap}}{})}
\label{apx:nonadaptclaim3}
Consider an inverse demand distribution $\bar{\invdemanddist}$ that satisfies the two constraints given by Claim \ref{clm:nonadaptive1} and Claim \ref{clm:nonadaptive2}, namely (i) $R^\invdemanddist(q)=R^\invdemanddist(q'), \forall q,q'\in [Q^\invdemanddist(1),1]$, and (ii) $\invdemanddist(\invdemand) = \invdemanddist(1)$ for all $\invdemand \in [1, +\infty)$. Let us define $\bar{q}$ such that $Q^{\bar{\invdemanddist}}(1) = \bar{q}$, or equivalently, $\Tfrvar^{\bar{\invdemanddist}}(\bar{q}) = 1$. 

The EAFR curve for $\bar{\invdemanddist}$ by definition attains a value of $R^{\bar{\invdemanddist}}(\bar{q}) = \bar{q}$. Since $\bar{\invdemanddist}$ has a constant EAFR curve in the interval $[\bar{q},1]$, its EAFR curve must be constantly equal to $\bar{q}$ over that interval. Futher, since the EAFR curve can also be expressed as $q \cdot \bar{\invdemanddist}^{-1}(1-q)$, we must have $\bar{\invdemanddist}(\bar{q}/q)=1-q$ for all  $q \in [\bar{q},1]$. Equivalently, using a change of variable $\invdemand = \bar{q}/q$, we must have $ \bar{\invdemanddist}(\invdemand)=1-\bar{q}/\invdemand$ for all  $\invdemand \in [\bar{q},1]$ (which implies $\bar{\invdemanddist}(\bar{q}) = 0$).

Further, since the CDF $\bar{\invdemanddist}$ pushes all the probability mass in the interval $(1, +\infty)$ to $+\infty$, $\bar{\invdemanddist}$ must be constant in the interval $[1,+\infty)$. Thus, we have uniquely described $\bar{\invdemanddist}$, up to a constant $\bar{q}$:

\begin{equation*}
\bar{\invdemanddist}(\invdemand) =
\left\{
	\begin{array}{ll}
		0 &\quad\quad\mbox{if }~\invdemand\in[0,\bar{q}) \\
		1-{\bar{q}}/{\invdemand} &\quad\quad\mbox{if}~\invdemand\in[\bar{q},1)  \\
		1-\bar{q} &\quad\quad\mbox{if}~\invdemand\in[1,+\infty)\\
		1 &\quad\quad\mbox{if}~\invdemand=+\infty
	\end{array}
\right. 
\end{equation*}

For this distribution to have an expected demand of $\ExpDemand$, consider the corresponding CDF for demand $\bar{F}: \NonNegReals\rightarrow [0,1]$ for the random variable $x\triangleq\frac{1}{\invdemand}$ where $\invdemand\sim \bar{\invdemanddist}$. We have:
\begin{equation*}
\bar{F}(x) = 1-\bar{\invdemanddist}(1/x)=
\left\{
	\begin{array}{ll}
		0 &\quad\quad\mbox{if }~x=0 \\
		{\bar{q}} &\quad\quad\mbox{if}~x\in(0,1]  \\
		\bar{q} x&\quad\quad\mbox{if}~x\in(1,\frac{1}{\bar{q}}]\\
		1 &\quad\quad\mbox{if}~x\in(\frac{1}{\bar{q}},+\infty)
	\end{array}
\right.
\end{equation*}
As a result, 
\begin{equation*}
    \expect[\invdemand\sim \invdemanddist]{\frac{1}{\invdemand}}=\int_0^{\infty}\left(1-\bar{F}(x)\right)dx=1-\bar{q}+\int_{1}^{\frac{1}{\bar{q}}}(1-\bar{q}x)dx=\frac{1}{2}\left(\frac{1}{\bar{q}}-
\bar{q}\right).
\end{equation*}
The unique solution to $\frac{1}{2}\left(\frac{1}{\bar{q}}-
\bar{q}\right) = \ExpDemand$ satisfying $\bar{q} \geq 0$ is $\bar{q} = \frac{1}{\ExpDemand +\sqrt{\ExpDemand^2+1}}$. 

We highlight that $\bar{\invdemanddist}$ with $\bar{q} = \frac{1}{\ExpDemand +\sqrt{\ExpDemand^2+1}}$ is identical to the distribution $\hat{\invdemanddist}$ defined in \cref{eq:worstthresholddist}. This shows that $\hat{\invdemanddist}$ is the unique worst-case distribution. \hfill\Halmos

\revcolor{
\subsection{\texorpdfstring{Proof of \Cref{prop:coefvar} (\Cref{subsec:AdaptivityGap})}{}}
\label{apx:prop:coefvar}
If we target a fill rate of $\TargetFillRate \in [0,1]$, the expected minimum fill rate is at least $$\TargetFillRate \prob{\TargetFillRate \sum_{i \in [\NumAgents]}\Demand_i < 1} = \TargetFillRate\left(1- \prob{ \sum_{i \in [\NumAgents]}\Demand_i - \ExpDemand \geq \frac{(1-\TargetFillRate\ExpDemand)}{\TargetFillRate}}\right).$$

By Cantelli's inequality (a generalization of Chebyshev's inequality for single-tailed distributions), if the variance of $\sum_{i \in [\NumAgents]}\Demand_i$ is given by $\sigma$, then for any $\delta \geq 0$, $\prob{\sum_{i \in [\NumAgents]}\Demand_i - \ExpDemand \geq \delta} \leq  \frac{\sigma^2}{\sigma^2 + \delta^2}$.

Taking $\delta = \frac{(1-\TargetFillRate\ExpDemand)}{\TargetFillRate}$ and letting $\TargetFillRate = \xvar/ \ExpDemand$, we can lower bound the expected minimum fill rate with 
$$\frac{\xvar}{\ExpDemand} \left(1- \frac{\sigma^2}{\sigma^2 + \left(\frac{1-\xvar}{\xvar}\right)^2 \ExpDemand^2}\right) \quad \geq \quad \frac{\xvar}{\ExpDemand} \left(\frac{\left(\frac{1-\xvar}{\xvar} \right)^2}{\coefvar^2 + \left(\frac{1-\xvar}{\xvar} \right)^2 }\right),$$
where the inequality comes from the assumption that the coefficient of variation is at most $\coefvar$.

We now optimize over feasible thresholds (i.e., over the domain $\xvar \in [0,\min\{1, \ExpDemand\}]$, which ensures $\delta \geq 0$ and $\TargetFillRate \in [0,1]$). The ex-post minimum fill rate of the optimal \fixedthresh\ policy is thus lower-bounded by
\begin{align}
    \max_{\xvar \in [0,\min\{1, \ExpDemand\}]} \frac{\xvar}{\ExpDemand} \left(\frac{\left(\frac{1-\xvar}{\xvar} \right)^2}{\coefvar^2 + \left(\frac{1-\xvar}{\xvar} \right)^2 }\right)
\end{align}
Scaling by the normalization factor of $\DetGuar = \min\{1, 1/\ExpDemand\}$ completes the proof of \Cref{prop:coefvar}. \hfill \halmos
}

\section{Proof of \texorpdfstring{\Cref{thm:exante}}{} (\texorpdfstring{\Cref{subsec:ex-ante}}{})}
\label{apx:ex-ante}
We prove this theorem by first providing two instances which together show that no policy achieves ex-ante fairness greater than $\LowerboundFunExAnte$. We then show that our \PolicyNameAb\ policy achieves ex-ante fairness of at least $\LowerboundFunExAnte$ for any demand distribution, which means the bound is tight.

{\bf Upper bound:} We prove the hardness result by considering two separate cases corresponding to $\ExpDemand < 2$ and $\ExpDemand \geq 2$. For each case, we provide an instance of the problem under which no policy can obtain ex-ante fairness larger than $\LowerboundFunExAnte$.

\begin{enumerate}[label=(\roman*)]
    \item If $\ExpDemand < 2$, consider a joint demand distribution for an arbitrary number of agents such that with probability $1-\frac{\ExpDemand}{2}$ there is no demand, and with probability $\frac{\ExpDemand}{2}$ the total demand is $2$ (arbitrarily and deterministically split among agents). In this case, with probability $1-\frac{\ExpDemand}{2}$, each agent achieves an FR of $1$, and with probability $\frac{\ExpDemand}{2}$, the expected FR cannot exceed $\frac{1}{2}$ for at least one agent. Therefore, in this instance the minimum expected FR is upper-bounded by $1-\tfrac{\ExpDemand}{4}$.
    \item If $\ExpDemand \geq 2$, consider a deterministic demand distribution where total demand is equal to its expectation $\ExpDemand$ (arbitrarily split among $\NumAgents$ agents). In this case, the minimum expected FR is clearly upper-bounded by $\frac{1}{\ExpDemand}$.
\end{enumerate}

Taken together, these instances provide an upper bound on the minimum expected FR one can hope to achieve with any policy. We then scale each instance by our normalization factor, namely $\DetGuar = \min\{1, 1/\ExpDemand\}$, which provides an upper bound of $\LowerboundFunExAnte$ on the ex-ante fairness guarantee (see \Cref{def:polperf}) of any policy. \hfill\Halmos

{\bf Lower bound:} We now show that the \PolicyNameAb\ policy achieves an ex-ante fairness guarantee of $\LowerboundFunExAnte$. First, we show via induction that after the arrival of any agent $i \in [\NumAgents]$, the ratio $\frac{\Demandi{i}+\ExpDemandi{i+1}}{ \Supplyi{i}}$, i.e., expected total remaining demand to remaining supply, is in expectation at most $\ExpDemand$. We then place a lower bound on the expected FR of agent $i$ when following the \PolicyNameAb\ policy which is a decreasing and convex function of the ratio $\frac{\Demandi{i}+\ExpDemandi{i+1}}{ \Supplyi{i}}$. We conclude by applying Jensen's inequality to lower bound the expected FR of agent $i$. Since this lower bound is the same for each agent, it constitutes a lower bound on the minimum expected FR.

\begin{claim}[Upper Bound on Demand-to-Supply Ratio]
\label{clm:exanteproof}
When following the \PolicyNameAb\ policy, for all $i \in [\NumAgents]$, $\expect[\vec{\Demand}\sim\DemandJoinDist]{\frac{\Demandi{i} + \ExpDemandi{i+1}}{\Supplyi{i}}} \leq \ExpDemand$.
\end{claim}

\begin{proof}{Proof:}
We proceed by induction. Clearly, when $i=1$, $\expect[\vec{\Demand}\sim\DemandJoinDist]{\frac{\Demandi{1} + \ExpDemandi{2}}{\Supplyi{1}}} = \expect[\vec{\Demand}\sim\DemandJoinDist]{\frac{\ExpDemandi{1} }{\Supplyi{1}}} = \ExpDemand$. We now assume that this holds for $i = k$ and attempt to prove the claim for $i = k+1$. According to the \PolicyNameAb\ policy, $\Alloci{k} \leq \Supplyi{k}\frac{\Demandi{k}}{\Demandi{k} +\ExpDemandi{k+1}}$. Thus, $\Supplyi{k+1}$ is at least $\Supplyi{k} \frac{\ExpDemandi{k+1}}{\Demandi{k} +\ExpDemandi{k+1}}$. Consequently, 
$$\expect[\vec{\Demand}\sim\DemandJoinDist]{\frac{\Demandi{k+1} +\ExpDemandi{k+2}}{\Supplyi{k+1}}} = \expect[\vec{\Demand}\sim\DemandJoinDist]{\frac{\ExpDemandi{k+1}}{\Supplyi{k+1}}} \geq \expect[\vec{\Demand}\sim\DemandJoinDist]{\frac{\ExpDemandi{k+1}}{\Supplyi{k} \frac{\ExpDemandi{k+1}}{\Demandi{k} +\ExpDemand{k+1}}}} =  \expect[\vec{\Demand}\sim\DemandJoinDist]{\frac{\Demandi{k} + \ExpDemandi{k+1}}{\Supplyi{k}}} \geq \ExpDemand.$$
The final inequality comes from our inductive hypothesis, which completes the proof by induction. \hfill\Halmos
\end{proof}

Given a current demand $\Demandi{i}$ and expected future demand $\ExpDemandi{i+1}$, the FR of agent $i$ is $\min\left\{1, \frac{\Supplyi{i}}{\Demandi{i} + \ExpDemandi{i+1}}\right\}$. It is straightforward to show that this FR is lower-bounded by the following function of the ratio $\frac{\Demandi{i} + \ExpDemandi{i+1}}{\Supplyi{i}}$:
\begin{equation*} \hfun\left(\frac{\Demandi{i} + \ExpDemandi{i+1}}{\Supplyi{i}}\right) = \begin{cases}
1 - \frac{1}{4}\frac{\Demandi{i} + \ExpDemandi{i+1}}{\Supplyi{i}}, &\frac{\Demandi{i} + \ExpDemandi{i+1}}{\Supplyi{i}} < 2\\
\frac{\Supplyi{i}}{\Demandi{i} + \ExpDemandi{i+1}}, &\frac{\Demandi{i} + \ExpDemandi{i+1}}{\Supplyi{i}} \geq 2
\end{cases}
\end{equation*}
We remark that $\hfun$ is a decreasing and convex function of its argument. Hence, by Jensen's inequality and Claim \ref{clm:exanteproof},
$$\expect[\vec{\Demand}\sim\DemandJoinDist]{\hfun\left(\frac{\Demandi{i}+\ExpDemandi{i+1}}{\Supplyi{i}}\right)}   \quad \geq \quad \hfun\left(\expect[\vec{\Demand}\sim\DemandJoinDist]{\frac{\Demandi{i}+\ExpDemandi{i+1}}{\Supplyi{i}}}\right) \quad \geq \quad \hfun(\ExpDemand).$$ Since this lower bound on the expected FR holds for each agent $i$, we have shown a lower bound on the minimum expected FR when following the \PolicyNameAb\ policy. When scaled by our normalization factor, namely $\DetGuar = \min\{1, 1/\ExpDemand\}$, this lower bound exactly matches the upper bound of $\LowerboundFunExAnte$ established above, and thus completes the proof of \Cref{thm:exante}.\hfill\Halmos

 \section{\texorpdfstring{Missing Proofs of \Cref{subsec:variants}}{}}
 \label{apx:discussionsection}
 \subsection{Proof of Corollary \ref{cor:generalfunctions}}
\label{apx:cor:generalfunctions}
Suppose the allocation given by the \PolicyNameAb\ policy is $\vec{\Alloc}$. By Theorem \ref{thm:expost}, the PPA policy achieves ex-post fairness of $\LowerboundFunEx$ when the social welfare function is the minimum FR, which implies $\SocialWelfare_{+\infty}(\vec{\Alloc}) = \min_{i \in[\NumAgents]}\{\frac{\Alloci{i}}{\Demandi{i}}\} \geq \LowerboundFunEx \DetGuar$. 

Now consider a new allocation vector $\vec{\Alloc}'$ such that for all $i \in [\NumAgents]$, $\frac{\Alloci{i}'}{\Demandi{i}} = \min_{j \in [\NumAgents]}\{\frac{\Alloci{j}}{\Demandi{j}}\}$. We remark that $\Alloci{i} \geq \Alloci{i}'$ for all $i \in [\NumAgents]$. Further, it is easy to verify that $\SocialWelfare_\FairParam(\vec{\Alloc}') = \SocialWelfare_{+\infty}(\vec{\Alloc})$ for any $\FairParam \in [0,+\infty)$. Since $\SocialWelfare_\FairParam(\vec{\Alloc})$ is non-decreasing in each $\Alloci{i}$, $ \SocialWelfare_\FairParam(\vec{\Alloc}) \geq \SocialWelfare_{+\infty}(\vec{\Alloc}) $.

Therefore, the expected WPM social welfare when following the \PolicyNameAb\ policy is weakly greater than the minimum FR (regardless of the fairness parameter $\FairParam$). Scaling by the achievable social welfare when demand is deterministic, we have shown that ex-post fairness, i.e., $\expect[\vec{\Demand}\sim\DemandJoinDist]{\SocialWelfare_\FairParam(\vec{\Alloc})}/\DetGuar$,  must be at least $\LowerboundFunEx$. \hfill\Halmos

\subsection{Proof of Corollary \ref{cor:multiplegoods}}
\label{apx:cor:multiplegoods}
\Cref{thm:expost} establishes that the \PolicyNameAb\ policy guarantees an expected minimum FR for resource $j$ of at least $ \LowerboundFunPnoarg\left(\frac{\ExpDemand^j}{\Supply^j}, \NumAgents\right) \max\left\{1, \frac{\ExpDemand^j}{\Supply^j}\right\}$. Consequently, we can place a lower bound on the expected minimum weighted FR:
\begin{equation}
 \expect[\vec{\Demand}\sim\DemandJoinDist]{\text{min}_{i \in [\NumAgents]} \sum_{j \in [\NumGoods]} \Weightj{j} \frac{\Alloci{i}^j}{\Demandi{i}^j}} \geq \sum_{j \in [\NumGoods]} \Weightj{j} \expect[\vec{\Demand}\sim\DemandJoinDist]{\text{min}_{i \in [\NumAgents]}  \frac{\Alloci{i}^j}{\Demandi{i}^j}} \geq \sum_{j \in [\NumGoods]} \Weightj{j} \LowerboundFunPnoarg\left(\frac{\ExpDemand^j}{\Supply^j}, \NumAgents\right) \max\left\{1, \frac{\ExpDemand^j}{\Supply^j}\right\}. \tag*{\hfill\Halmos}
\end{equation}

\subsection{Proof of Corollary \ref{cor:multiplegoodsupper}}
\label{apx:cor:multiplegoodsupper}
Consider a correlated distribution for demands---correlated across agents and resource types---where the marginal distribution of demands for each resource type matches the worst-case joint distribution of the single-type problem described in the proof of \Cref{thm:hardness-ex-post}. These marginal distributions are then coupled such that the last-arriving agent with non-zero demand for resource $j$ is the same as the last-arriving agent with non-zero demand for resource $j'$, for any resources $j, j' \in [\NumGoods]$ for which at least one agent has non-zero demand. Given the marginal distributions, each agent is equally likely to be this last-arriving agent. For any sample path drawn from this distribution, it is without loss of generality to only consider policies where the allocation is decreasing (i.e., where this last-arriving agent has the worst FR) for every resource.\footnote{For intuition as to why this is without loss of generality, note (i) each agent with non-zero demand for resource $j \in [m]$ has identical demand for that resource, (ii) each agent linearly aggregates their FRs by the same weights $\{\Weightj{j}\}_{j\in[m]}$, and (iii) each agent is equally likely to be the last-arriving agent. Consequently, if agent $i$ has a strictly larger allocation than agent $i'$ (where $i' < i$) for any resource $j$, a policy which switches their allocations for that resource will have a weakly greater expected minimum weighted FR.}

If supply of resource $j$ is $\MultiGoodSupply^j$, \Cref{thm:hardness-ex-post} establishes an upper-bound of $\LowerboundFunPnoarg\left(\frac{\ExpDemand^j}{\MultiGoodSupply^j},\NumAgents\right)\max\left\{1, \frac{\ExpDemand^j}{\MultiGoodSupply^j}\right\}$ on the expected minimum FR of any policy in the single-type problem corresponding to resource $j$. Since for every sample path the agent with the worst FR (i.e., the last-arriving agent) is the same across resources, aggregating these bounds establishes an upper-bound of $\sum_{j \in [m]} \Weightj{j} \LowerboundFunPnoarg\left(\frac{\ExpDemand^j}{\MultiGoodSupply^j}, \NumAgents\right) \max\left\{1, \frac{\ExpDemand^j}{\MultiGoodSupply^j}\right\}$ on the expected minimum weighted FR. \hfill\Halmos
 \end{APPENDIX}
\end{document}